\documentclass[twoside,11pt]{article}

\usepackage{blindtext}

\usepackage[preprint]{jmlr2e}
\usepackage{mathrsfs}
\usepackage{subcaption}
\usepackage{placeins}

\newtheorem{assumption}{Assumption}

\usepackage{hayesmacros}

\usepackage{lastpage}
\jmlrheading{26}{2025}{1-\pageref{LastPage}}{10/23; Revised
11/24}{1/25}{23-1317}{Alex Hayes, Mark M. Fredrickson, and Keith Levin}
\ShortHeadings{Network-mediated causal effects}{Hayes, Fredrickson and Levin}

\begin{document}

\title{Estimating Network-Mediated Causal Effects via Principal Components Network Regression}

\author{\name Alex Hayes \email alex.hayes@wisc.edu \\
	\addr Department of Statistics\\
	University of Wisconsin--Madison\\
	Madison, WI, USA
	\AND
	\name Mark M. Fredrickson \email mfredric@umich.edu \\
	\addr Department of Statistics \\
	University of Michigan\\
	Ann Arbor, MI, USA
	\AND
	\name Keith Levin \email kdlevin@wisc.edu \\
	\addr Department of Statistics\\
	University of Wisconsin--Madison\\
	Madison, WI, USA}

\editor{Ilya Shpitser}

\maketitle

\begin{abstract}%
	We develop a method to decompose causal effects on a social network into an indirect effect mediated by the network, and a direct effect independent of the social network. To handle the complexity of network structures, we assume that latent social groups act as causal mediators. We develop principal components network regression models to differentiate the social effect from the non-social effect. Fitting the regression models is as simple as principal components analysis followed by ordinary least squares estimation. We prove asymptotic theory for regression coefficients from this procedure and show that it is widely applicable, allowing for a variety of distributions on the regression errors and network edges. We carefully characterize the counterfactual assumptions necessary to use the regression models for causal inference, and show that current approaches to causal network regression may result in over-control bias. The method is very general, so that it is applicable to many types of structured data beyond social networks, such as text, areal data, psychometrics, images and omics.
\end{abstract}


\section{Introduction}

Recent years have seen a concerted effort to study causal effects on networks, motivated by striking claims about contagions in social networks \citep{christakis2007}. One of the key ideas to emerge from this push is the need to account for clustering in networks \citep{shalizi2011}. Sociologists have long known that people in social networks are mostly connected to other people like themselves, which is often expressed informally as ``birds of a feather flock together,'' and more formally called ``homophily''. To identify and estimate causal effects in social settings, it is thus fundamental to model how social groups form in networks, as well any downstream effects of social group membership. This is challenging, as social groups in a network are typically unobserved.

To account for unobserved social structure in networks, \citet{shalizi2011} proposed using latent space models, where each node has an embedding that determines its propensity to connect with other nodes. More recently, \cite{mcfowland2021} showed that certain types of causal effects on networks can be estimated by controlling for latent node embeddings in linear network regression models. There is now a rapidly growing literature investigating how embeddings can be used for causal inference on networks \citep{paul2022, veitch2019, veitch2020, louizos2017, guo2020, chu2021, guo2020, cristali2022, chen2022a, ogburn2022, liu2022, omalley2014, leung2019, egami2021a, ogburn2022,ogburn2020}.

We make two contributions to this literature. First, we develop broad semi-parametric theory for network regression, showing how to incorporate node embeddings into linear regression. These network regression models are useful for causal inference on networks, as they allow practitioners to account for homophily, but they are also of substantial independent interest as an observation model. Similar models have previously appeared in \cite{li2019d,le2022, fosdick2015, he2019a} and in concurrent work \citep{nath2023, chang2024a}. Compared to these approaches, we require substantially weaker assumptions on the edge distribution in the network and the error distribution in the regression models. We allow weighted (i.e., real-valued) edges and allow edges to be observed with noise. Edges in the network and regression errors can be sampled from any distribution that satisfies a sub-gamma tail bound. Further, our regression models allow for heteroscedastic errors. Altogether, our results show that estimating network embeddings with principal components analysis, together with ordinary least squares to estimate regression coefficients, is applicable under broad semi-parametric conditions. This contrasts with previous approaches that have largely focused on specific parametric models.

Our second contribution is to characterize how latent node-level variables (i.e., embeddings) can act as causal mediators \citep{imai2010}. To date, causal methods for network data have mainly focused on homophily as a confounding factor. In contrast, we articulate the causal mechanisms involved when treatments influence latent social group membership. This is perhaps best illustrated by an example. Suppose we are interested in understanding and reducing adolescent smoking \citep{dimaria2022a,michell1997}. Understanding why adolescents smoke is challenging, as smoking is both a sexually differentiated behavior and a social behavior. To study the effect of sex on adolescent smoking, we decompose the effect of sex into a direct effect independent of the social network among adolescents and an indirect effect that operates through the network. For instance, sex may directly cause higher cigarette consumption through sex-specific expectations about smoking. Sex may also have an indirect effect mediated by the social network: adolescents may prefer to form sex-homophilous friendships, and social norms about the acceptability of smoking may vary across friend groups. That is, sex may influence an adolescent's social circumstances, which in turns causes a change in smoking behavior. These effects correspond to natural direct and indirect effects, specialized to the network setting by treating latent community memberships as mediators.

To estimate network-mediated effects, we propose an approach based on a generalization of the random dot product graph \citep{athreya2018,levin2022a} and principal components network regression \citep{le2022, cai2021, paul2022, mcfowland2021, upton2017}. To represent the mediating effect of the network, we embed the network into a low-dimensional Euclidean space via a singular value decomposition \citep{sussman2014}. We then use the embeddings to develop two network regression models: (1) an outcome model that characterizes how nodal outcomes vary with nodal treatment, controls, and position in latent space; and (2) a mediator model that characterizes how latent positions vary with nodal treatment and controls. We estimate the regressions using ordinary least squares and Huber-White robust standard errors, and prove that the resulting coefficients are asymptotically normal under general semi-parametric conditions. Once estimation is complete, the coefficients from the outcome model can be used to estimate the direct effect of treatment, and the coefficients from the outcome model and the mediator model can be combined to estimate the indirect effect of treatment \citep{vanderweele2014a}. These estimators are essentially the well-known product-of-coefficients estimators \citep{vanderweele2014a, nguyen2021c}, but using the network embedding as a mediator. We anticipate that our causal estimators will be accessible to many social and natural scientists, as they are based on familiar product-of-coefficients mediation techniques. Importantly, the mediation estimands we consider in this work are distinct from peer effects such as contagion and interference, explored elsewhere in the literature \citep{ogburn2020,ogburn2022}. Our mediational estimands measure the effect of meso-level social structures, in contrast to more micro-level peer influence. In Section \ref{subsec:peer} we explore the relationship between these mechanisms.

While the product-of-coefficient estimator is familiar to many, it relies on strong functional form assumptions. We want to emphasize that our characterization of latent mediation is valuable even to readers who suspect these assumptions are too strong to hold in practice. Following \cite{mcfowland2021}, practitioners and theorists are increasingly using linear regressions for causal inference on networks, and careful characterizations of the causal structure of these models is increasingly important \citep{paul2022, mcfowland2021, veitch2019, veitch2020, louizos2017, guo2020, chu2021, guo2020, cristali2022, chen2022a, ogburn2022, liu2022, omalley2014, leung2019, egami2021a, ogburn2022,ogburn2020}. In our development of the product-of-coefficients estimator, we correct the misconception that homophily can only have a confounding effect, we describe the over-control bias that can be induced by causal misspecification, and we document a network data set where causal misspecification leads to substantially misleading results. Indeed, this work was originally motivated by a misinterpretation of coefficients in a principal components network regression model.

We expect our estimators to have significant causal and non-causal applications beyond network setting. Similar methods for low-rank data have been employed, for example, on spatial networks \citep{gilbert2021, tiefelsdorf2007, doreian1981, ord1975}, text data \citep{keith2021, veitch2020, gerlach2018}, psychometric surveys \citep{freier2022, thurstone1947}, imaging data \citep{zhao2020,levin2022a}, and omics panels \citep{listgarten2010, alter2000}. Indeed, in the sequel, we apply our method to psychometric data, in addition to our running adolescent smoking example. In data applications, we find that adolescent girls end up smoking more than adolescent boys primarily due to an indirect network effect. This suggests that public health interventions might want to disrupt or alter social group formation. In an application to psychometrics data, we find that a meditation app reduces anxiety primarily by helping study participants defuse (i.e., step back and examine) from their emotions and by helping them feel less lonely. Our findings suggest that the meditation program might improve mental health outcomes by replacing meditation modules aimed at increasing a sense of meaning in life with additional modules targeting defusion or loneliness.

Our work is related to several extant lines of research. Past authors have heuristically proposed regression estimators for latent mediation in networks under the Hoff model \citep{hoff2002}, albeit without asymptotic theory or much elaboration of causal mechanisms~\citep{dimaria2022a,liu2021, che2021}. These ideas are related to work on non-parametric mediation analysis \citep{tchetgentchetgen2012, farbmacher2022, fulcher2020,zheng2012} and semi-parametric methods for hidden mediators such as \cite{cheng2022}, as well as methods for proximal mediation such as \cite{dukes2021} and \cite{ghassami2021b}. Similarly, \cite{sweet2019, sweet2022, guha2021, zhao2022} consider mediation in networks, albeit treating entire networks as mediators, an approach that requires observing an entire network per unit of analysis.

Beyond causal considerations, this paper proposes semi-parametric methods for principal components network regression. Similar results have previously been established in more restrictive parametric settings. \cite{le2022} and \cite{paul2022} show a related result in binary networks, under the assumption that regression errors are Gaussian, and \cite{cai2021} considers Gaussian networks with a one-dimensional latent space. We substantially relax these assumptions to allow for latent spaces of arbitrary dimension and to permit far more general edge and regression error distributions.  The regression models that we propose are related to several other forms of regression studied within the network literature, such as classic spatial and econometric interference models \citep{land1992, manski1993} and regressions with network-coherence penalties \citep{li2019d}. For a review, we refer the interested reader to \cite{le2022}. Also related is a large body of work on network association testing, such as \cite{ehrhardt2019, fredrickson2019, lee2019b, gao2022, su2020}.

\subsection*{Notation} \label{notation} For a matrix $A$, let $\norm*{A}, \norm*{A}_F$ and $\norm*{A}_{2, \infty}$ denote the spectral, Frobenius, and two-to-infinity norms, respectively. We write $A^\dagger$ for the Moore-Penrose pseudoinverse, $A_{i \cdot}$ for the $i$-th row, $A_{\cdot j}$ for the $j$-th column, and $\vecc(A)$ for the column-wise vectorization of $A$, i.e., $\vecc(A) = (A_{\cdot 1}^T, A_{\cdot 2}^T, \dots, A_{\cdot n}^T)^T$ for a matrix with $n$ columns. We use $\otimes$ to denote the Kronecker product. We write $[n]$ to denote the set $\set{1, 2, \dots, n}$. $\bbO_d$ denotes the set of $d \times d$ orthogonal matrices. When we define a new symbol inline, we use $\equiv$. We use standard Landau notation, e.g., $\bigoh{a_n}$ and $\littleoh{a_n}$ to denote growth rates, as well as the probabilistic variants $\Op{a_n}$ and $\op{a_n}$. For example, $g(n) = \bigoh{f(n)}$ means that for some constant $C>0$, $|g(n)| < C f(n)$ for all suitably large $n$.  In proofs, $C$ denotes a constant not depending on the number of vertices $n$, whose precise value may change from line to line, and occasionally within the same line.

\section{The Causal Structure of Mediation in Latent Spaces}

A central goal of this paper is to formalize and estimate natural direct and indirect effects as mediated by latent positions in a network. Briefly, the idea is that node-level treatments can affect formation of social groups, and that membership in social groups can further influence nodal outcomes.

\subsection{Motivating Example: Social and Non-Social Elements of Teenage Smoking} \label{subsec:motivating-example}

Let us return to the adolescent smoking example mentioned previously, which is based on the \emph{Teenage Friends and Lifestyle Study}, reported in \cite{michell1996}, \cite{michell1997}, \cite{michell1997a}, and \cite{michell2000a}. The \emph{Teenage Friends and Lifestyle Study} collected three waves of survey data in a secondary school in Glasgow, beginning in January 1995. Students in the study filled out a questionnaire about their lifestyle and risk-taking behaviors, including alcohol, tobacco and drug use, and additionally were asked to list six of their friends. Beyond this quantitative data, researchers also conducted in-depth qualitative work, investigating the social dynamics of the school through focus groups, classroom observations, and interviews.

In the Glasgow data, as in prior investigations, researchers found that adolescent girls smoked more than adolescent boys. \cite{michell1996} and \cite{michell2000a} proposed two distinct effects on smoking: one about athletic expectations, and the other about social influence. \cite{michell2000a} proposed that athletic expectations affected smoking as follows: both boys and girls were subject to general societal pressure to smoke, but adolescent boys were subject to additional pressure to be good at sports \citep{michell1996}. Since smoking was generally accepted as reducing athletic performance, this influenced adolescent boys to abstain from smoking in favor of pursuing athletic social capital \citep{michell2000a}.

The second proposed effect was a social effect. Smoking is known to be a collective behavior, where group members typically all partake or all abstain \citep{michell2000a}. Researchers found evidence of this form of group decision-making for various risk-taking behaviors in the Glasgow study, including tobacco consumption. This group-level decision-making led to differentiated behavior among adolescent boys and girls because adolescent friendships are highly sex-homophilous until the onset of puberty \citep{mehta2009}. Taken all together, this meant that the effect of sex on smoking was mediated by friend group membership: adolescent friend group membership was heavily determined by sex, and smoking was heavily determined by friend group. More precisely, \cite{michell2000a} found that smoking was mostly concentrated in friend groups composed of popular girls, unpopular students, and trouble-makers: ``risk taking behaviour was heavily polarized within social categories so that, for instance, groups of individuals (and their peripherals) were in general either risk-taking or non-risk-taking [...] generally groups (and their peripherals) were either all boys or all girls.''

\begin{figure}[ht!]
	\begin{subfigure}{0.49\textwidth}
		\centering
		\includegraphics[width=\textwidth]{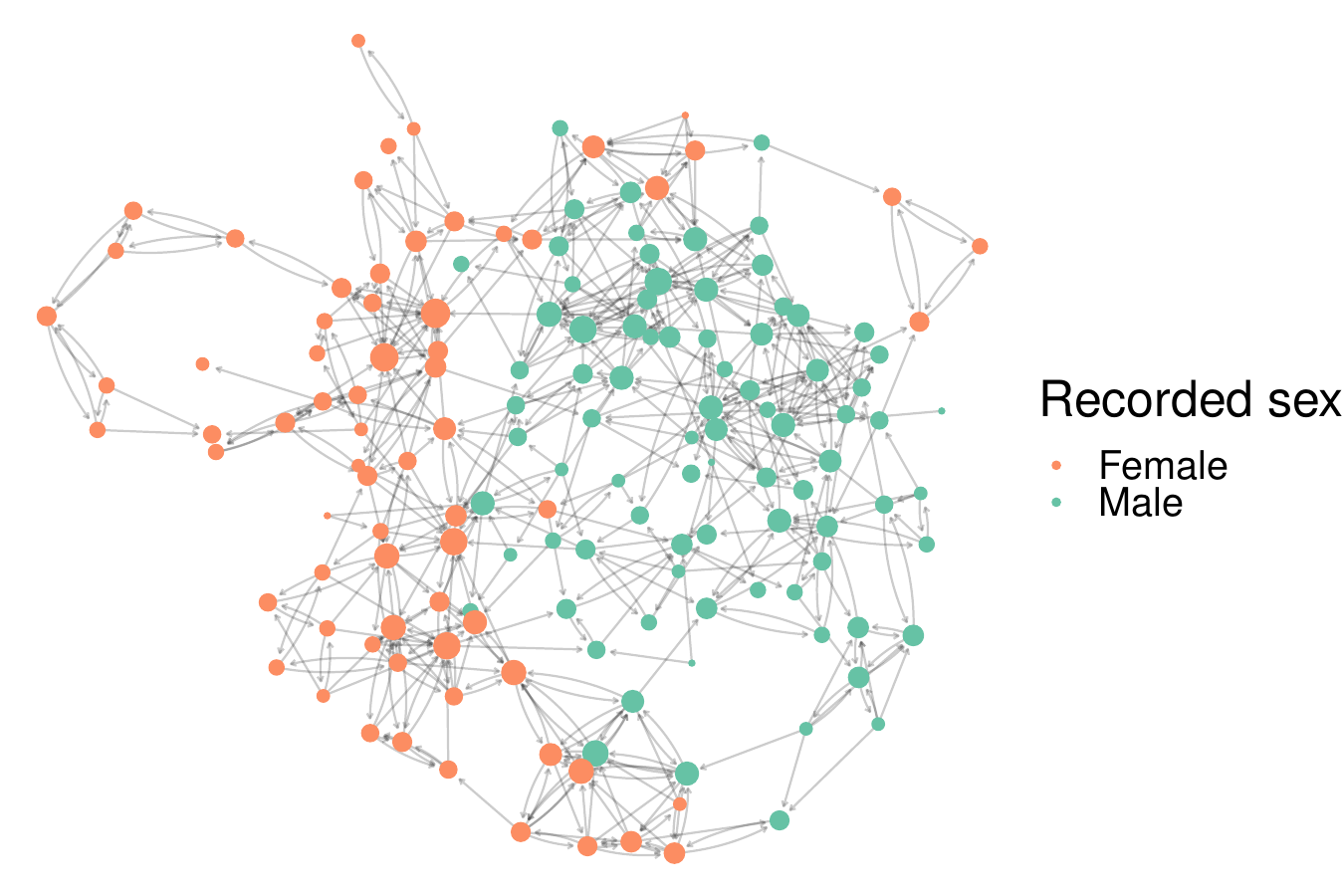}
	\end{subfigure}
	\begin{subfigure}{0.49\textwidth}
		\centering
		\includegraphics[width=\textwidth]{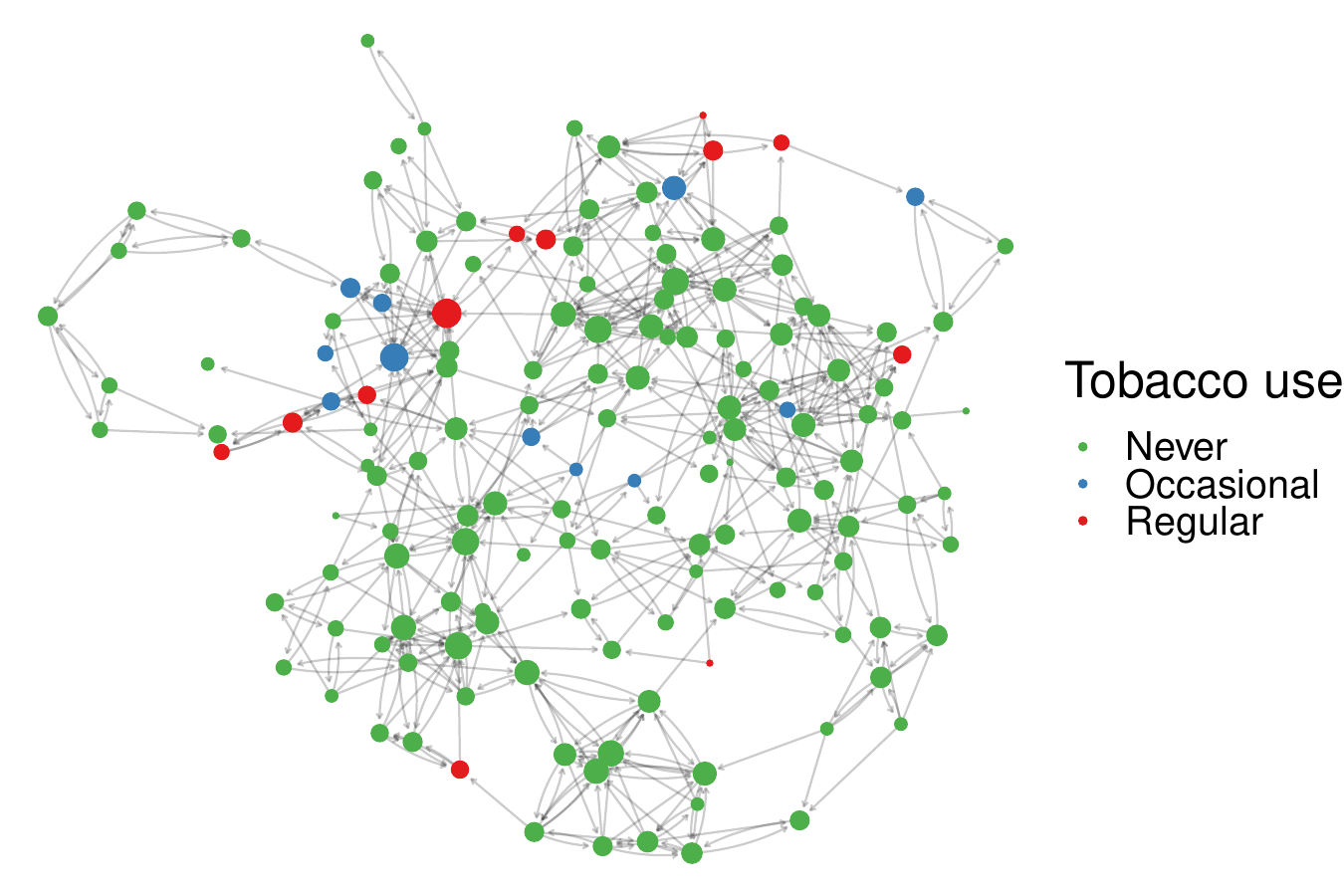}
	\end{subfigure}
	\caption{Directed friendships in a secondary school in Glasgow, reported in the Teenage Friends and Lifestyle Study (wave 1). Each node represents one student. An arrow from node $i$ to node $j$ indicates student $i$ claimed student $j$ as a friend. Node size is proportional to in-degree.}
	\label{fig:glasgow}
\end{figure}

The goal of this paper is to formalize a causal model for network-linked data that allows us to quantify and estimate causal effects like those proposed in the example above, a setting that we refer to as \emph{homophilous mediation}. The first mechanism corresponds to a direct effect, independent of social considerations, and the second mechanism corresponds to an indirect effect, with friend group membership mediating the causal effect of sex on smoking. In the Glasgow study, it seems intuitive that sex exerted both a direct and indirect effect on tobacco consumption, driving adolescent girls to smoke more than adolescent boys. Our formalization, developed in Sections~\ref{subsec:problem-setup} and~\ref{subsec:semi-param-model}, also accommodates observed confounders influencing both the direct and the indirect pathways. For example, some students in the Glasgow study reported attending church, which has the potential to influence both their social group and attitudes about tobacco.

\subsection{Causal Mediation in Latent Social Spaces}
\label{subsec:problem-setup}

To formalize a model for mediated effects on a network, we begin by introducing some notation. We assume that we have a network with $n$ nodes (corresponding to our experimental units), which we label according to the integers $[n] = \{1,2,\dots,n\}$. Let $T \in \set{0, 1}^n$ be a vector of observed binary treatment indicators for nodes $i \in [n]$, and let $Y \in \R^n$ be a vector of observed node-level outcomes. Let $\X \in \R^{n \times d}$ be a matrix describing the unobserved ``friend groups'' of each node (to be formalized below), with the row vector $\X_{i \cdot} \in \R^{1 \times d}$ encoding the friend group of node $i$ for each $i \in [n]$. Similarly, let $\C \in \R^{n \times p}$ be a matrix of observed confounders, with the row vector $\C_{i \cdot} \in \R^{1 \times p}$ denoting the confounders associated with node $i$. Lastly, let $A_{ij} \in \R$ denote the strength of the friendship between node $i$ and node $j$. For simplicity, our model takes friendships to be symmetric, such that $A_{ij} = A_{ji}$ for all $i, j \in [n]$, but this condition can be relaxed to allow for directed networks (see Section~\ref{sec:applications}).

To define causal estimands, we introduce notation for the necessary counterfactual quantities, defined in the sense of structural causal models \citep{pearl2009a}.

\begin{definition}
	Let $Y_i(t)$ be the counterfactual value of the outcome measured for the $i^{th}$ node when $T_i$ is set to $t$. Similarly, let $Y_i(t, x)$ be the counterfactual value of the $i^{th}$ outcome when $T_i$ is set to $t$ and $\X_{i \cdot}$ is set to $x$, and let $\X_{i \cdot}(t)$ be the counterfactual value of the mediator $\X_{i \cdot}$ when $T_i$ is set to $t$.
	The \emph{average treatment effect}, \emph{natural direct effect} and \emph{natural indirect effect} are defined, respectively, as
	\begin{equation*}
		\begin{aligned}
			\atef & = \E{Y_i(t) - Y_i(t^*)},                                                    \\
			\ndef & = \E{Y_i(t, \X_{i \cdot}(t^*)) - Y_i(t^*, \X_{i \cdot}(t^*))}, \text{ and } \\
			\nief & = \E{Y_i(t, \X_{i \cdot}(t)) - Y_i(t, \X_{i \cdot}(t^*))}.
		\end{aligned}
	\end{equation*}
\end{definition}

Note that the average treatment effect decomposes into the sum of the natural direct effect and the natural indirect effect: $\atef = \ndef + \nief$. We interpret all three estimands following Chapter 2 of \cite{vanderweele2015}: the average treatment effect $\ate$ describes how much the outcome $Y_i$ would change on average over all nodes $i \in [n]$ if the treatment $T_i$ were changed from $T_i = t$ to $T_i = t^*$. The natural direct effect describes how much the outcome $Y_i$ would change on average over all nodes $i \in [n]$ if the exposure $T_i$ were set at level $T_i = t^*$ versus $T_i = t$ but for each individual the mediator $\X_{i \cdot}$ were kept at the level it would have taken for that individual, had $T_i$ been set to $t^*$. The natural indirect effect describes how much the outcome $Y_i$ would change on average over all nodes $i = [n]$ if the exposure were fixed at level $T_i = t^*$ but the mediator $\X_{i \cdot}$ were changed from the level it would take under $T_i = t$ to the level it would take under $T_i = t^*$.

In slightly more plain language, we can interpret the natural direct effect as capturing the effect of the exposure on the outcome when the mediating pathway is disabled. In the smoking example discussed in Section~\ref{subsec:motivating-example}, this would correspond to the effect of sexed expectations alone. Similarly, we can interpret the natural indirect effect as capturing the effect of the exposure on the outcome that operates by changing the mediator while keeping treatment fixed \citep{vanderweele2015}. In the smoking example, this corresponds to the effect of sex on friend group, and then friend group on smoking behavior. The total effect of sex on smoking is the sum of the effects from sexed athletic expectations and friend group pressures.

In order to estimate counterfactual quantities, we must make identifying assumptions to relate unobserved counterfactual quantities to the observable data. Identification is implied by the properties of consistency, sequential ignorability, and positivity \citep{imai2010}, which are standard sufficient conditions within the causal inference literature.

\begin{assumption}[Non-parametric Identification of Natural Direct and Indirect Effects]
	\label{ass:mediation}
	The random variables $(Y_i, Y_i(t, x), \X_{i \cdot}, \X_{i \cdot}(t), C_{i \cdot}, T_i)$ are independent over $i \in [n]$ and obey the following three properties.
	\begin{enumerate}
		\item Consistency: \vspace{-3mm}
		      \begin{equation*} \begin{aligned}
				       & \text{if $T_i = t$, then $\X_{i \cdot}(t) = \X_{i \cdot}$ with probability 1, and}    \\
				       & \text{if $T_i = t$ and $\X_{i \cdot} = x$, then $Y_i(t, x) = Y_i$ with probability 1}
				      \vspace{-2mm}
			      \end{aligned} \end{equation*}
		\item Sequential ignorability:
		      \begin{equation*}
			      \set{Y_i(t^*, x), \X_{i \cdot}(t)} \indep T_i \cond \C_{i \cdot}
			      ~~~\text{ and }~~~
			      \set{Y_i(t^*, x)} \indep \X_{i \cdot}  \cond T_i = t, \C_{i \cdot}
		      \end{equation*}
		\item Positivity:
		      \begin{equation*}
			      \begin{aligned}
				      \P[T_i, \C_{i \cdot}]{x} & > 0 \text{ for each }  x \in \supp(\X_{i \cdot}) \\
				      \P[\C_{i \cdot}]{t}      & > 0 \text{ for each }  t \in \supp(T_i)
			      \end{aligned}
		      \end{equation*}
	\end{enumerate}

\end{assumption}

Roughly speaking, a sufficient condition for natural direct and indirect effects to be non-parametrically identified is that the observed controls $\C_{i \cdot}$ contain all confounders of the exposure-outcome ($T_i \to Y_i$), the exposure-mediator ($T_i \to \X_{i \cdot}$) and the mediator-outcome ($\X_{i \cdot} \to Y_i$) relationships. One structural causal model satisfying these requirements is given in Figure \ref{fig:mediating}. Note that, crucially, our model differs from traditional mediation models because the mediators (i.e., the group memberships $\X_{i \cdot}$) are unobserved.

Three assumptions are particularly important from a counterfactual perspective. As in tabular settings, the sequential ignorability assumption is strong and may not hold due to mediator-outcome confounding. Two other concerns are more specific to the network setting. In particular, positivity may be a problem if friend groups are highly homophilous. That is, conditional on treatment and controls, some regions of the latent space might have zero probability mass. In the context of the adolescent smoking example, this is potentially an issue, as friend groups are highly sexually homophilous. Since empirical networks often exhibit high degrees of homophily, positivity violations may present a larger challenge in network settings than in non-network settings. A third crucial assumption is that there are not peer effects such as contagion or interference, a topic we discuss in detail in Section \ref{subsec:peer}.

\begin{figure}[t!]
	\centering
	\includegraphics[width=0.6\textwidth]{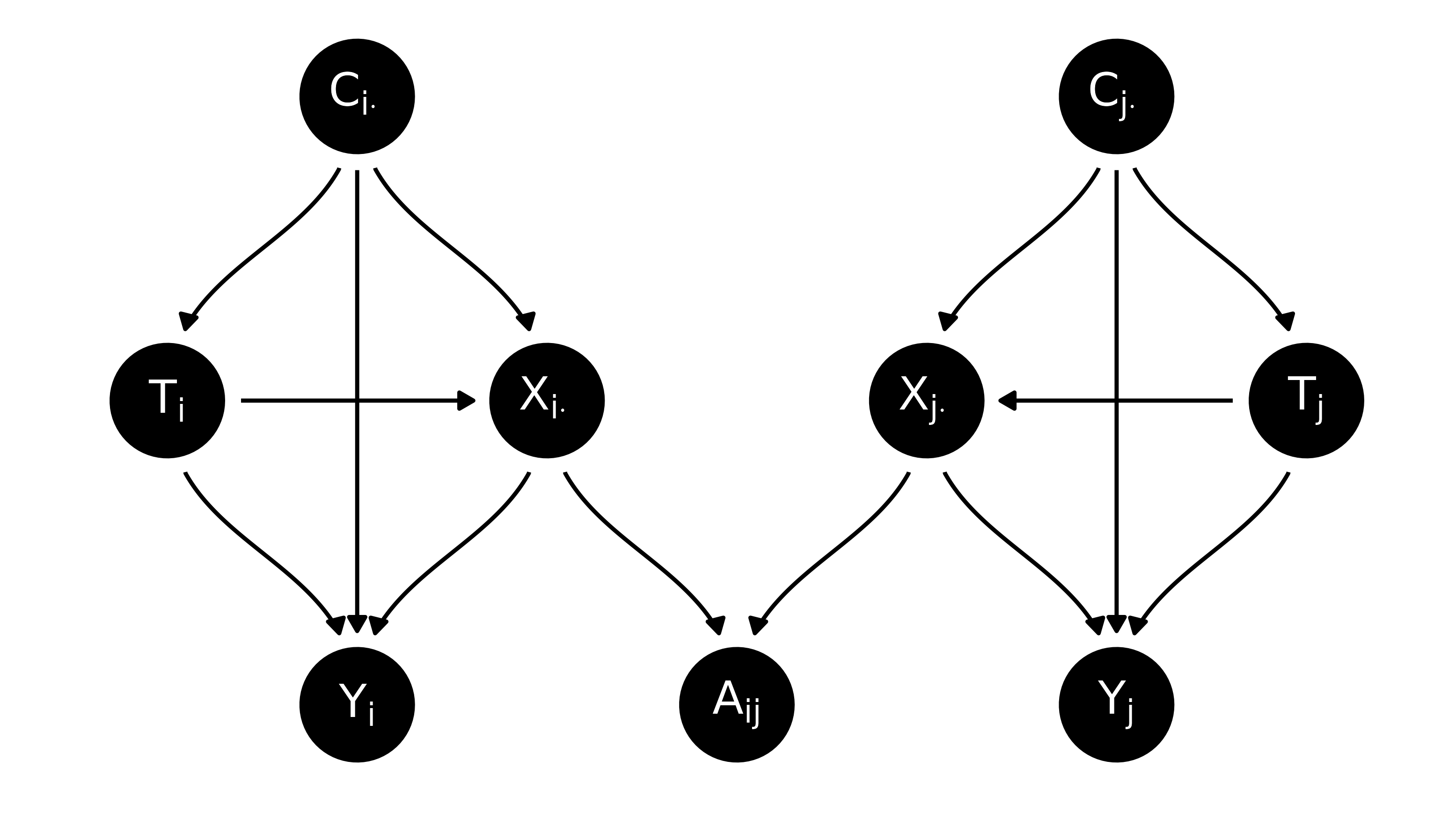}
	\caption{A directed acyclic graph (DAG) representing the causal pathways of latent mediation in a network with two nodes called $i$ and $j$. Each node in the figure corresponds to a random variable, and edges indicate which random variables may cause which other random variables. We are interested in the causal effect of $T_i$ on $Y_i$, as mediated by the latent position $\X_{i \cdot}$.}
	\label{fig:mediating}
\end{figure}

\subsection{Semi-Parametric Latent Space Structure}
\label{subsec:semi-param-model}

The counterfactual model we have presented thus far is a standard causal model for mediation, with the exception that we have not clarified the role of $\X_{i \cdot}$, which we claimed should correspond to a latent measure of social group membership. We now fix ideas about $\X_{i \cdot}$ and clarify the details of the network model by developing a statistical model for low-rank matrices. This low-rank model can be thought of as a generalization of the random dot product model~\citep{athreya2018,nickel2006,bonato2007} to networks with weighted (i.e., non-binary) edges.

Describing (possibly) weighted edges requires a brief technical pre-requisite.
\begin{definition}
	Let $Z$ be a mean-zero random variable with cumulant generating function $\psi_Z(t) = \log \E{e^{t Z}}$.

	\begin{enumerate}
		\item $Z$ is $\nu$-sub-Gaussian for $\nu > 0$ if $\psi_Z(t) \le t^2 \nu / 2$ for all $t \in \R$.

		\item $Z$ is $(\nu,b)$-sub-gamma for $\nu, b \ge 0$ if $\psi_Z(t) \le \frac{t^2 \nu}{2 (1 - b t)}$ and $\psi_{-Z}(t) \le \frac{t^2 \nu}{2 (1 - b t)}$ for all $t < 1 / b$.
	\end{enumerate}
\end{definition}

The class of sub-gamma distributions is broad, and includes as special cases the Bernoulli, Poisson, Exponential, Gamma, and Gaussian distributions, as well as any sub-Gaussian or squared sub-Gaussian distribution, and all bounded distributions \citep[see][for a detailed treatment]{boucheron2013}. Our primary assumption on the network structure is that the edges are sampled according to sub-gamma distributions and the network is low-rank in expectation \citep{boucheron2013,tropp2015}.

\begin{assumption}[Sub-gamma network]
	\label{ass:subgamma}
	Let $A \in \R^{n \times n}$ be a random symmetric matrix, such as the adjacency matrix of an undirected graph. Let $\Apop = \E[\X]{A} = \X \X^T$ be the expectation of $A$ conditional on $\X \in \R^{n \times d}$, which has independent and identically distributed rows $\X_{1 \cdot}, \dots, \X_{n \cdot}$.
	The matrix $\Apop$ has $\rank \paren*{\Apop} = d$ and is positive semi-definite with eigenvalues $\lambda_1 \ge \lambda_2 \ge \cdots \ge \lambda_d > 0 = \lambda_{d+1} = \cdots = \lambda_n$.
	Conditional on $\X$, the upper-triangular elements of $A - \Apop$ are independent $(\nu_n, b_n)$-sub-gamma random variables.
\end{assumption}

This characterization of the network structure is very general and encompasses a range of popular network models as special cases. The most important example of the sub-gamma model for our purposes is the random dot product graph.

\begin{example}[Random Dot Product Graph; \cite{athreya2018}]
	\label{ex:rdpg}
	Let $F$ be a distribution on $\R^d$ such that $x^T y \in [0, 1]$ whenever $x, y \in \supp F$. Draw $\X_{1 \cdot}, \X_{2 \cdot}, \dots, \X_{n \cdot}$ i.i.d.\ according to $F$, and collect these $n$ points in the rows of $\X \in \R^{n \times d}$. Conditional on $\X$, the edges of graph $G$ are generated independently, with probability of an edge $i \sim j$ given by $\X_{i \cdot} \X_{j \cdot}^T$. That is, conditional on $\X$, the entries of the symmetric adjacency matrix $A$ above the diagonal are independent with $A_{ij} \sim \Bern \paren{\X_{i \cdot} \X_{j \cdot}^T}$.
	Then we say that $A$ is distributed according to a random dot product graph with latent position distribution $F$ and write $(A,X) \sim \RDPG( F, n )$.
\end{example}

Under the random dot product graph, each node in a network is associated with a latent vector, and these latent vectors characterize propensities to form edges with other nodes. Specifically, nodes close to each other in latent space are more likely to form connections, and nodes far apart are unlikely to form connections. When nodes cluster in the latent space, the result is that edges in the network are more likely to form between nodes with similar latent characteristics. This manifests as homophily in the resulting network.

Another particularly important sub-gamma model is the stochastic blockmodel, which is in fact a sub-model of the random dot product graph.

\begin{example}[Stochastic Blockmodel]\label{ex:sbm}
	The stochastic blockmodel~\citep[SBM;][]{holland1983} is a model of community membership, in which each vertex is assigned to a community (sometimes called a ``block''). Conditional on assignments to communities, edges are generated independently between every pair of vertices in the network, and the probability of forming an edge between vertices $i$ and $j$ depends only on the community memberships of nodes $i$ and $j$.

	Let $B \in [0, 1]^{d \times d}$ denote a fixed, positive semi-definite matrix of inter-block edge probabilities and let $\Z_{i \cdot} \in \R^d$ be a vector encoding the block membership of node $i$.  Conditional on the matrix $B \in [0, 1]^{d \times d}$ and on the block memberships encoded by the rows of $\Z \in \set{0, 1}^{n \times d}$, the behavior of the stochastic blockmodel is characterized by
	\begin{equation*}
		\P[\Z, B]{A_{ij} = 1} = \Z_{i \cdot} B \Z_{j \cdot}^T.
	\end{equation*}
	The basic stochastic blockmodel, as introduced in \cite{holland1983}, forces each node to belong to exactly one block. That is, it requires $\Z_{i \cdot} \in \set{0, 1}^d$ to have exactly one entry equal to one. The degree-corrected stochastic blockmodel \citep{karrer2011} relaxes this restriction by giving each node a ``degree-heterogeneity'' parameter that encodes a node's propensity to form edges. This is equivalent to requiring $\Z_{i \cdot} \in \R_+^d$, where $d - 1$ entries of $\Z_{i \cdot}$ are still zero. The overlapping stochastic blockmodel \citep{latouche2011} allows $\Z_{i \cdot} \in \set{0, 1}^d$ with no additional restrictions, such that nodes can belong to multiple blocks. The mixed membership stochastic blockmodel \citep{airoldi2008} restricts $\Z_{i \cdot}$ to lie on the $d - 1$ dimensional simplex, such that nodes can have partial membership in multiple blocks, but their total block participation must sum to one. The overlapping and mixed-membership variants may also be extended to include degree-correction, in which case $\Z_{i \cdot}$ can be a nearly arbitrary vector in $\R^d_+$ \citep{jin2024a,zhang2020b}.

	To express the stochastic blockmodel in more general sub-gamma form, take $\X = Z B^{1/2}$. Thus, the latent positions $\X$ encode (1) block participation as characterized by $\Z$, (2) degree-adjustment (i.e., popularity) captured by the row scales of $\Z$, and (3) intra-block edge formation propensities contained in $B$.
\end{example}

Before moving on to describe how the latent positions $\X$ are related to other network covariates, some remarks about the general sub-gamma model are warranted. First, the sub-gamma model allows for edges to be observed with noise.

\begin{example}[Noisily Observed Random Dot Product Graph] \label{ex:noisy-rdpg}
	Let $\paren*{\mathscr A, \X} \sim \RDPG(F, n)$ and let $\{ E_{ij} : 1 \le i < j \le n\}$ be independent, mean-zero sub-gamma random variables for $1 \le i \le j \le n$. For example, $E_{ij}$ might correspond to (centered) Gaussian or Bernoulli noise. Then the network $A_{ij} = \mathscr A_{ij} + E_{ij}$ satisfies Assumption~\ref{ass:subgamma}, as the sum of sub-gamma random variables remains sub-gamma.
\end{example}

We note that all our results presented below can be extended to asymmetric $\Apop$, rectangular $\Apop$, and $\Apop$ with negative eigenvalues. The assumption that $\Apop$ is symmetric and positive semi-definitive is primarily to simplify notation, and our proofs can be extended to the general case using the techniques of \cite{rubin-delanchy2022} and \cite{rohe2023}. Thus, the sub-gamma model can be extended to handle bipartite and directed graphs~\citep{qing2021, rohe2016}. Other natural extensions include models such as Gaussian mixtures with identity covariance, latent Dirichlet allocation~\citep{rohe2023, blei2003}, topic models~\citep{gerlach2018}, and psychometric factor models~\citep{thurstone1947,thurstone1934}.

It is also important to note that, like the random dot product graph, the sub-gamma model considered here is subject to orthogonal non-identifiability. Since $\Apop = \X \X^T = (\X Q) (\X Q)^T$ for any $d \times d$ orthogonal matrix $Q$, the latent positions $\X$ are only identifiable up to an orthogonal transformation (see \citealt{athreya2018} for further discussion). Luckily, this non-identifiability of latent positions does not influence identifiability of the natural direct and indirect effects.

Lastly, under the sub-gamma model, the latent positions $\X$ are sufficient for $A$. That is, the nodal covariates $T_i, C_{i \cdot}$ and $Y_i$ do not directly influence the formation of edges of $A_{ij}$. The covariates $T_i$ and $C_{i \cdot}$ {\em can} influence edge formation, but only via the intermediary $\X$.  Relaxing this constraint is an interesting topic for future work, but requires different estimators for the latent positions $\X$ than those we consider here \citep{mele2023,binkiewicz2017}.

With the network structure established, we now relate the latent positions $\X_{i \cdot}$ to the node-level observations $(T_i, C_{i \cdot}, Y_i)$.

\begin{assumption}[Linear Conditional Expectations]
	\label{ass:causal-linearity}

	The outcome regression functional is linear in $T_i, \C_{i \cdot}$, and $\X_{i \cdot}$ and the mediator regression functional is linear in $T_i$ and $\C_{i \cdot}$:
	\begin{equation*}
		\begin{aligned}
			\underbrace{\E[T_i, \C_{i \cdot}, \X_{i \cdot}]{Y_i}}_{\R}
			 & = \underbrace{\betazero}_{\R}
			+ \underbrace{T_i}_{\{0, 1\}} \underbrace{\betat}_{\R}
			+ \underbrace{\C_{i \cdot}}_{\R^{1 \times p}} \underbrace{\betac}_{\R^{p}}
			+ \underbrace{\X_{i \cdot}}_{\R^{1 \times d}} \underbrace{\betax}_{\R^d},
			 & \text{(outcome model)}                      \\
			\underbrace{\E[T_i, \C_{i \cdot}]{\X_{i \cdot}}}_{\R^{1 \times d}}
			 & = \underbrace{\thetazero}_{\R^{1 \times d}}
			+ \underbrace{T_i}_{\{0, 1\}} \underbrace{\thetat}_{\R^{1 \times d}}
			+ \underbrace{\C_{i \cdot}}_{\R^{1 \times p}} \underbrace{\Thetac}_{\R^{p \times d}}
			 & \text{(mediator model)}
		\end{aligned}
	\end{equation*}
	The columns of $T, \C$ and $\X$ must be linearly independent for regression coefficients to be identifiable. The latent dimension $d$ and the number of nodal controls $p$ are constants that do not vary with sample size.
\end{assumption}

In this model, $\betazero$ and $\thetazero$ play the role of intercept terms, while $\betat$ and $\betac$ encode the average associations between nodal covariates $T, \C$ and nodal outcomes $Y$, conditional on the effect of the latent positions $\X$.  $\thetat$ and $\Thetac$ describe how latent positions in the network vary with nodal covariates.  In the general sub-gamma edge model that we consider here, it is difficult to provide an interpretation of $\betax$ and the mediator coefficients because the latent positions $\X$ can take on several roles depending on the precise parametric sub-model under consideration (e.g., under the SBM as compared to the more general RDPG).  Nonetheless, $\betax$ represents the average association between latent network structure and nodal outcomes, conditional on nodal covariates. Since $\X$ is only identified up to orthogonal rotation, $\betax$ and $\Theta$ are similarly only identified up to orthogonal rotation.

When we combine the statistical model entailed by Assumption \ref{ass:causal-linearity} with the previous counterfactual assumptions entailed by Assumption \ref{ass:mediation}, the regression coefficients $\beta$ and $\Theta$ have known causal interpretations.

\begin{proposition}[\citealt{vanderweele2014a}]
	\label{prop:identification-mediated}
	Under Assumptions \ref{ass:mediation} and \ref{ass:causal-linearity}
	\begin{align}
		\ndef & = \paren*{t - t^*} \, \betat, ~ \text{ and } ~    \label{eq:nde} \\
		\nief & = \paren*{t - t^*} \, \thetat \, \betax \label{eq:nie}
	\end{align}
\end{proposition}

\begin{remark}
	An immediate question is whether Proposition~\ref{prop:identification-mediated} is useful, given that $\thetat$ and $\betax$ are subject to an unknown orthogonal transformation. Luckily, non-identifiability of the regression coefficients does not impact identifiability of causal estimands. $\nde$ depends only on $\betat$, which is fully identified. On the other hand, $\nie$ is a function of $\thetat$ and $\betax$, and the unknown orthogonal transformations for these estimates cancel (a result that follows immediately from Theorem~\ref{thm:main} below). This is because $\nie$ depends fundamentally on the projection of $Y$ onto the span of $\{T, \C, \X\}$ and the projection of $\X$ onto the span of $\{T, \C\}$, rather than the precise bases of those subspaces.
\end{remark}

\begin{remark}
	Proposition \ref{prop:identification-mediated} identifies the natural direct and indirect effects based solely on outcome and mediator regression models, and does not require a propensity score model. It is possible to construct more complicated estimators that additionally leverage the propensity score, and these estimators can be more robust and efficient than the product-of-coefficients estimator considered here \citep{nguyen2021c,tchetgentchetgen2012}. We are primarily interested in network regression in this work, and leave exploration of other estimators to future work.
\end{remark}

\subsection{Causal Mechanisms in Latent Space}
\label{subsec:causal-interpretation}

A first challenge when considering mediation in a latent space is to understand how latent mediation works in the abstract. To characterize one possible form of latent mediation, we develop some intuition in the context of degree-corrected mixed-membership stochastic blockmodels, as in Example~\ref{ex:sbm}. Recall that, under this model, the latent position $\X_{i \cdot}$ encodes the group memberships of node $i$, and the ``popularity'' of node $i$ (i.e., the propensity of node $i$ to form connections with other nodes)\footnote{More precisely, under the degree-corrected mixed-membership stochastic blockmodel, $\P[\Z, B]{A_{ij} = 1} = \Z_{i \cdot} B \Z_{j \cdot}^T$, where $\Z \in \R_+^{n \times d}$ is a matrix of popularity-scaled group membership weights, and $B$ is a positive semi-definite mixing matrix. In particular, under this model, $\X = Z B^{1/2}$. The $B^{1/2}$ factor in $\X$ can be slightly counter-intuitive at first. For the sake of intuition, one can conceptualize interventions as applying to the block memberships $\Z$, while holding the mixing matrix $B$ constant. In Proposition~\ref{prop:alt-param}, we show that any intervention on $\Z$ satisfying the parametric constraints of Assumption~\ref{ass:causal-linearity} is equivalent to an intervention on $\X$, which also satisfies the parametric constraints of Assumption~\ref{ass:causal-linearity}. Thus, for concreteness, one can think about interventions applying directly to the block memberships $\Z$. Under Assumption \ref{ass:causal-linearity}, interventions can also simultaneously influence $B$, but these interventions may feel less intuitive.}. The latent position $\X_{i \cdot}$ is involved in several causal pathways.

\begin{enumerate}
	\item The homophily-inducing pathway $T_i \to \X_{i \cdot}$. Since $\X_{i \cdot}$ encodes the group membership and popularity of node $i$, intervening on $\X_{i \cdot}$ can cause node $i$ to participant in different communities, or change its popularity, or simultaneously translate node $i$ to new communities while also scaling its popularity. Assumption \ref{ass:causal-linearity} implies that intervention must, on average, cause a translation in the latent space.

	\item The network-formation pathway $\X_{i \cdot} \to A_{i \cdot}$. Intervening on $\X_{i \cdot}$~simultaneously modifies $\P[\X]{A_{ij} = 1}$ for all $j \in [n]$. That is, intervening on node $i$'s community membership changes node $i$'s probability of connecting to every other node.

	\item The social outcome effect pathway $\X_{i \cdot} \to Y_i$. This encodes the idea that community membership and popularity influence outcomes.
\end{enumerate}

To develop intuition for interventions on the latent positions, consider a degree-corrected mixed membership stochastic blockmodel with five blocks, where nodes are primarily members of a single block and there is a small amount of degree heterogeneity. Interventions can increase or decrease participation in particular communities. For instance, in Figure \ref{fig:intervention-latent}, we visualize an intervention that decreases participation in the second community while increasing participation in the third community, in the latent space.

\begin{figure}[t!]
	\centering
	\includegraphics{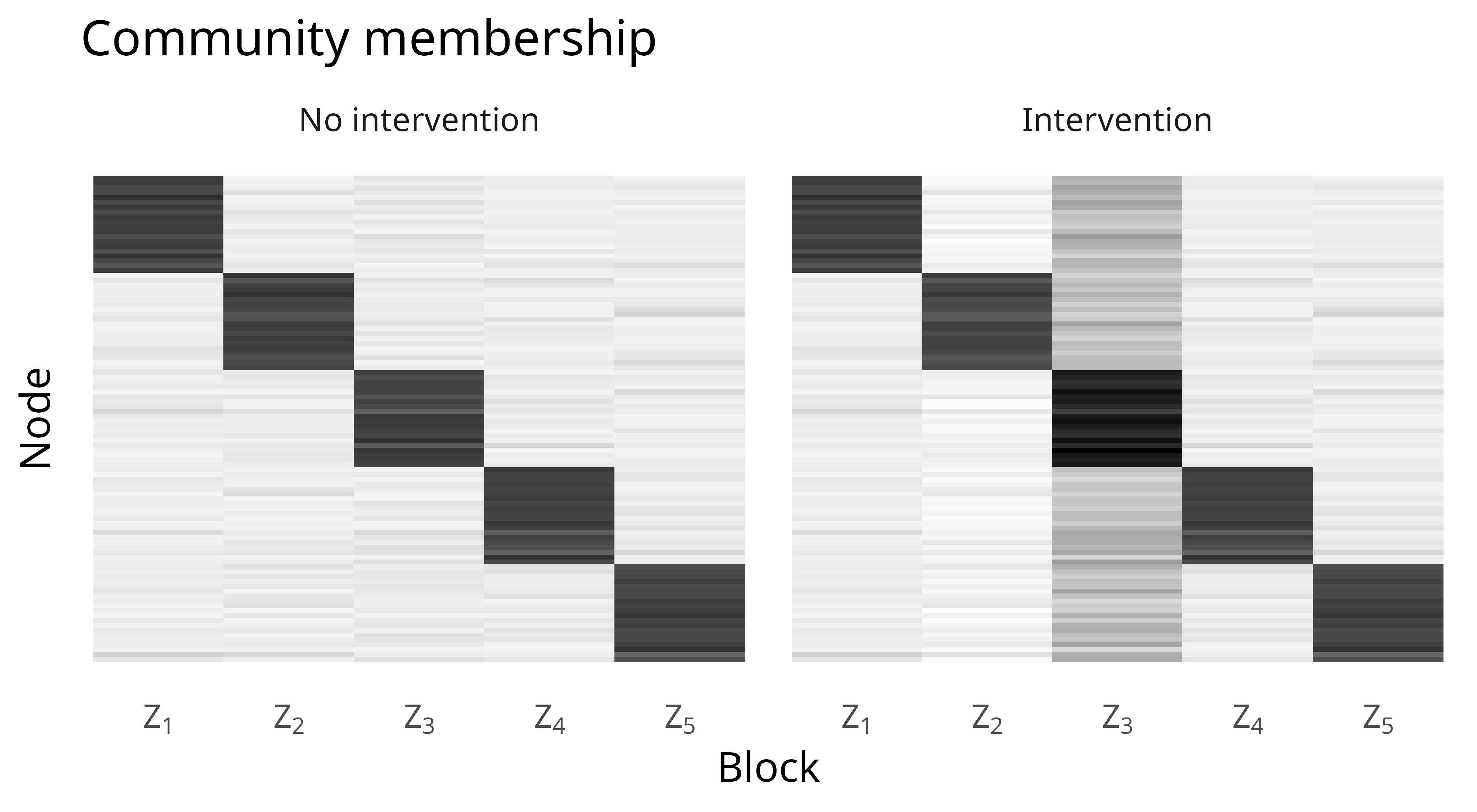}
	\caption{Block memberships in a degree-corrected mixed-membership stochastic blockmodel before and after intervention. Each row corresponds to a node in the network and each column corresponds to a block in the blockmodel. Darker colors indicated increased participation in a given block. The intervention decreases participation in the second community slightly and increases participation in the third community more dramatically. Nodes are sorted by community membership. }
	\label{fig:intervention-latent}
\end{figure}

Intervention in turn alters the probability of friendship between pairs of nodes, and thus leads to different counterfactual networks via the pathway $T \to \X \to A$. To understand the impact of the intervention on the network itself, Figure \ref{fig:intervention-network} visualizes the difference between counterfactual networks where intervention does and does not occur. In the top left panel, we see friendship probabilities when $T$ does not impact $X$, and in the top right panel, we see friendship probabilities when $T$ causes a change in $X$. Intervention causes all nodes to participate in the third community. The mixing matrix $B$ is a diagonal matrix in this example, so nodes participating in the third community have some probability to connect with other nodes also participating in the third community. This means that, under intervention, all nodes in the network are more likely to connect with one another, due to mutual participation in the third block. The effect is especially pronounced for nodes that began in the third block. The bottom panels of Figure \ref{fig:intervention-network} show corresponding random samples from the networks conditional on latent positions $\X$.

\begin{figure}[t!]
	\begin{subfigure}{0.5\textwidth}
		\centering
		\includegraphics[width=0.7\textwidth]{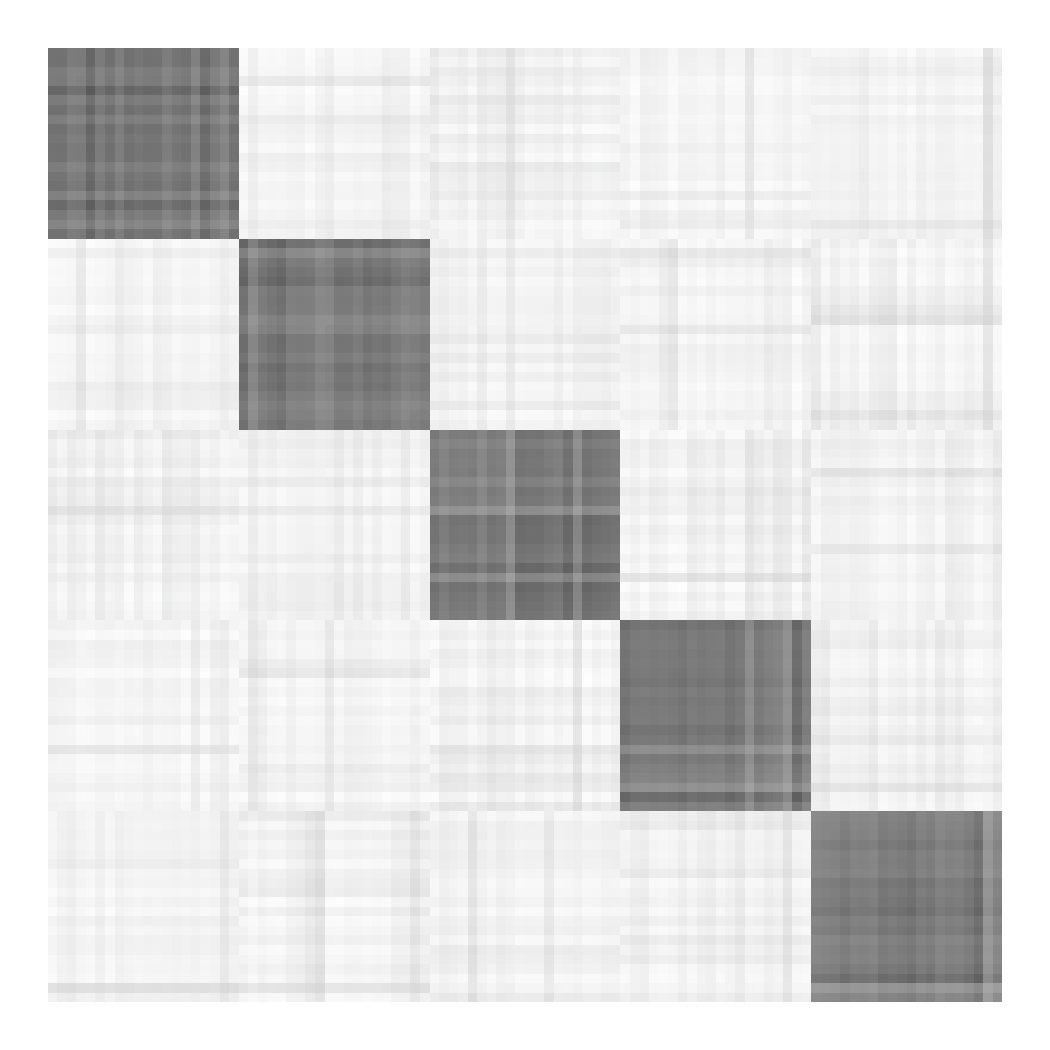}
		\caption{$\Apop$ without intervention.}
		\label{fig:intervention-expected-a-pre}
	\end{subfigure}
	\begin{subfigure}{0.5\textwidth}
		\centering
		\includegraphics[width=0.7\textwidth]{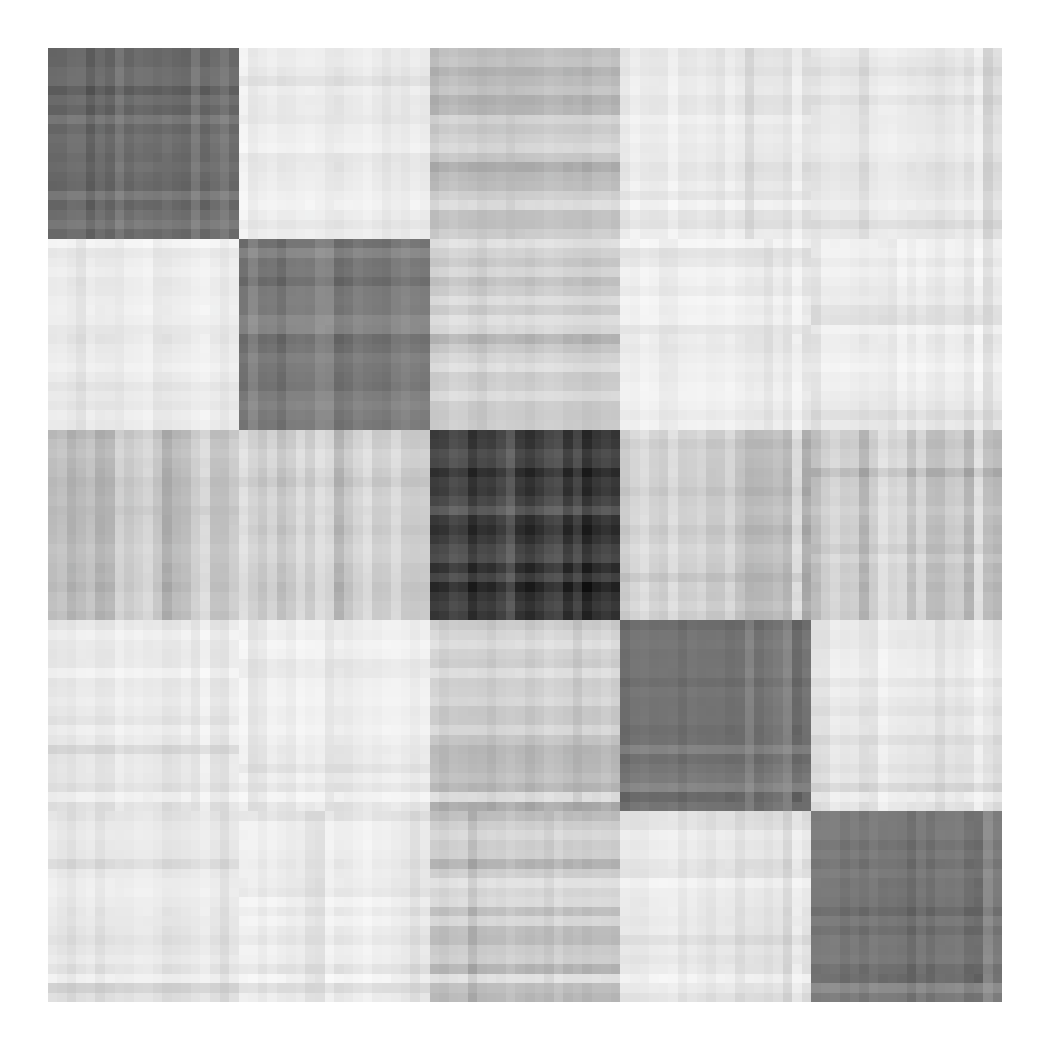}
		\caption{$\Apop$ with intervention.}
		\label{fig:intervention-expected-a-post}
	\end{subfigure}
	\begin{subfigure}{0.5\textwidth}
		\centering
		\includegraphics[width=0.7\textwidth]{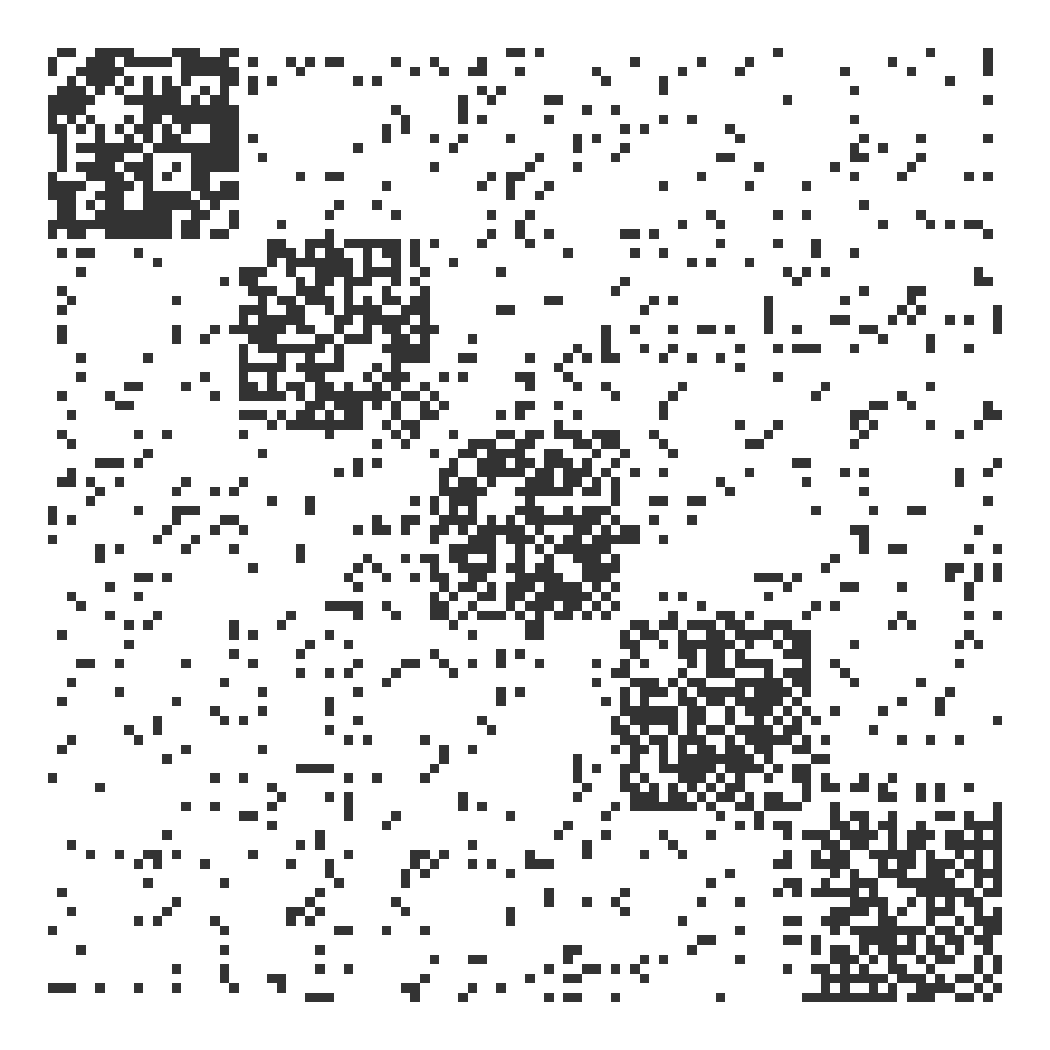}
		\caption{Realized network without intervention}
		\label{fig:intervention-observed-a-untreated}
	\end{subfigure}
	\begin{subfigure}{0.5\textwidth}
		\centering
		\includegraphics[width=0.7\textwidth]{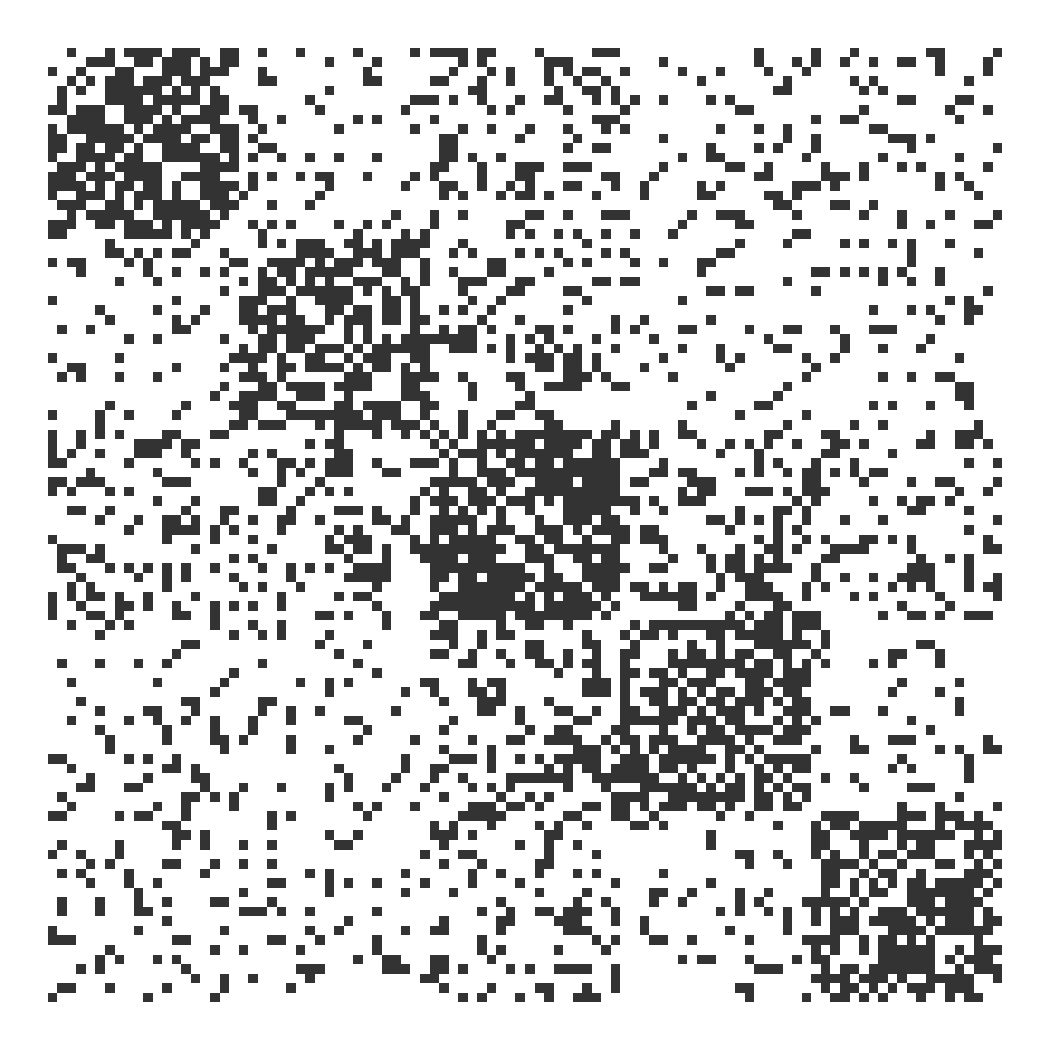}
		\caption{Realized network with intervention}
		\label{fig:intervention-observed-a-treated}
	\end{subfigure}
	\caption{A visualization of how intervention impacts the probability of edge formation $\Apop$, and ultimately the realized network $A$. In all cases matrices as visualized as heatmaps. Each element of the heatmap corresponds to one edge in the network. Sub-figure~\subref{fig:intervention-expected-a-pre} shows edge formation probabilities pre-intervention, and Sub-figure~\subref{fig:intervention-expected-a-post} shows edge formation probabilities post-intervention. Sub-figure~\subref{fig:intervention-observed-a-untreated} shows a network realized when the intervention doesn't occur, and Sub-figure~\subref{fig:intervention-observed-a-treated} shows a network realized when the intervention does occur. Nodes are sorted by community membership.}
	\label{fig:intervention-network}
\end{figure}

This synthetic example is primarily to develop intuition about the latent positions, and downstream consequences of an intervention. Even with this intuition, it may be difficult to understand what it means to intervene on a latent position. In practice, we believe that the best way to resolve this ambiguity is by interpreting the latent positions in the context of the relevant data application. In practice, we compute estimates of the latent positions using the adjacency spectral embedding (see Section \ref{sec:statistical}; one can think of $A$ as a collection of proxy measurements that we use to infer $X$). With estimates of the latent positions in hand, it is often much easier to reason about causal mechanisms, because the estimated positions can be given a concrete interpretation via inspection, auxiliary data and domain knowledge. In the psychometric data application of Section \ref{subsec:healthy-minds}, for instance, the latent positions have a very clear interpretation. To interpret the latent positions, we find that Gaussian mixture models and varimax rotation are often valuable tools \citep{rohe2023,rubin-delanchy2022, zhang2021d, priebe2019}.

To summarize: we believe that the degree-corrected mixed-membership stochastic blockmodel offers compelling intuition about the structures that latent positions can represent. This intuition is helpful to determine when a latent mediation model may be applicable: for instance, in settings where one suspects that social groups act as mediators. However, there is an additional confirmatory step necessary for empirical work. It is critical to confirm that the estimated latent positions do in fact capture the hypothesized mediating constructs. Estimates of latent positions might reflect other network structure unrelated to hypothesized mediators. This is a weakness of our smoking example: we know from extensive sociology research that the adolescent social network is composed of numerous small social groups. However, we cannot confirm that the social structure observed by sociologists is the same as the social structure captured by the network principal components, an issue that we discuss in more detail in Section \ref{subsec:glasgow}. In contrast, in the psychometric example of Section \ref{subsec:healthy-minds}, the estimated latent positions are clearly measures of latent constructs of interest.

\subsection{Relation to Peer Effects}
\label{subsec:peer}

The mediation mechanism that we have proposed considers how network effects manifest through group-level dynamics rather than individual-level interactions. In network settings, it is commonly assumed that there may be peer effects operating at the level of individual interactions. Two especially pertinent peer effects are contagion and interference \citep{hu2022c, ogburn2020}. Contagion occurs when the outcome of neighbor $j$ impacts the outcome of node $i$ (i.e., there is a path $Y_j \to Y_i$). Interference occurs when the treatment applied to neighbor $j$ impacts the outcome of node $i$ (i.e., there is a path $T_j \to Y_i$). In the context of our smoking example, eliding some simultaneity issues, contagion would occur if a student's tobacco usage were influenced by their friends' tobacco usage. Interference would occur if a student's tobacco usage were influenced by their friends' sex.

Both our counterfactual assumptions and the parametric form of our outcome model assume that there are no peer effects. This restriction clearly limits the applicability of our method, but there are some caveats. First, in ongoing work, we are extending the network mediation model to allow for these peer effects. The extension is quite involved, but preliminary results and related work suggest that, at least some of the time, it is possible to estimate both mediated effects and peer effects using ordinary least squares and estimated latent positions $\Xhat$, as we propose here \citep{chang2024a, lee2002, paul2022a, trane2023}. That is, one can view the network mediation model, and our characterization of mediation in a latent space, as descriptions of an important subset of causal mechanisms in networks. In many cases, these mechanisms should be jointly modelled, but developing methods to do so is a complicated technical endeavor that is outside the scope of this paper.

Part of the challenge in jointly modelling latent mediation and peer effects is that the mechanisms can be difficult or impossible to distinguish. \cite{shalizi2011} showed this in a non-parametric longitudinal setting, but a very similar consideration arises in the parametric cross-sectional setting. To see this, consider the ``linear-in-sums'' model, one popular way to model contagion and interference \citep{paul2022,mcfowland2021,egami2021a, hu2022c,bramoulle2020}:
\begin{equation*}
	\E[T_i, \C_{i \cdot}, \X, A, Y_{-i}]{Y_i}
	= \betazero + T_i \betat + \C_{i \cdot} \betac + \X_{i \cdot} \betax + \beta_\text{Ay} \sum_{j \neq i}^n A_{ij} Y_j + \beta_\text{At} \sum_{j \neq i}^n A_{ij} T_j.
\end{equation*}
The coefficient $\beta_\text{Ay}$ encodes a contagion effect, and the coefficient $\beta_\text{At}$ encodes an interference effect. In forthcoming work, we show that columns of the design matrix corresponding to $\betax$, $\beta_\text{Ay}$ and $\beta_{At}$ are collinear in the asymptotic limit. That is, the interference column $AT$ and the contagion column $AY$ of the design matrix are both contained in the column space of $\X$ in large samples. This occurs because the column spaces of $A$ and $\X$ get closer with increasing sample size under the sub-gamma model, and thus peer effects that are linear functions of $A$ can be equivalently represented as linear functions of $\X$. From a parametric perspective, the coefficient $\betax$ generalizes ``linear-in-sums'' style peer effects. In the context of our smoking example, the causal effect of belonging to a particular friend group is asymptotically indistinguishable from a contagion effect that depends on the number of friends who are smokers. Thus, the causal effects along the pathway $\X_{i \cdot} \to Y_i$ can be thought of as effects of group memberships, or alternatively, as consequences of diffusions over the network. The fact that this equivalence holds only in the asymptotic limit, however, introduces a number of subtle caveats, which we describe in a forthcoming manuscript.

In addition to ``linear-in-sums'' style peer effect, it is also possible to consider ``linear-in-means'' style peer effects,
\begin{equation*}
	\E[T_i, \C_{i \cdot}, \X, A, Y_{-i}]{Y_i}
	= \betazero + T_i \betat + \C_{i \cdot} \betac + \X_{i \cdot} \betax + \beta_\text{Ay} \sum_{j \in \mathcal N(i)} \frac{A_{ij}}{d_i} Y_j + \beta_\text{At} \sum_{j \in \mathcal N(i)} \frac{A_{ij}}{d_i} T_j
\end{equation*}
\noindent where $\mathcal N(i) = \{j \in [n]: A_{ij} = 1\}$ denotes the neighbors of node $i$ and $d_i = \sum_j A_{ij}$ is the degree of node $i$ \citep{bramoulle2009}. In the linear-in-means approach, $\betax$, $\beta_\text{Ay}$ and $\beta_{At}$ are identified, even in the asymptotic limit, under some non-trivial assumptions about the network structure. Recent work has considered these models from a statistical point of view \citep{chang2024a, lee2002, paul2022a, trane2023}, although we believe there are still substantial identifications concerns to address, from both a causal and a statistical perspective.

\section{Estimation Theory for Principal Components Network Regression}
\label{sec:statistical}

Having established a model, we have several estimation targets. First, there are the network regression coefficients $\beta$ and $\Theta$ from Assumption~\ref{ass:causal-linearity}. Once we estimate the coefficients $\beta$ and $\Theta$, we can plug them into Equations~\eqref{eq:nde} and~\eqref{eq:nie} to obtain estimates of the natural direct and indirect effects, respectively.

\subsection{Principal Components Network Regression}
\label{subsec:network-regression}

Before discussing estimation, we simplify notation by collecting the node-level covariates into $\W = \begin{bmatrix} 1 & T & \C \end{bmatrix} \in \R^{n \times (p + 2)}$. Assumption~\ref{ass:causal-linearity} can then be re-written as
\begin{equation} \label{eq:def:W} \begin{aligned}
		\E[\W_{i \cdot}, \X_{i \cdot}]{Y_i}
		 & = \W_{i \cdot} \betaw + \X_{i \cdot} \betax,
		 & \text{(outcome model)}                       \\
		\E[\W_{i \cdot}]{\X_{i \cdot}}
		 & = \W_{i \cdot} \Theta.
		 & \text{(mediator model)}
	\end{aligned} \end{equation}

We would like to estimate $\betaw, \betax$ and $\Theta$ by applying ordinary least squares regression to the vertex-level latent positions $\X$.  Unfortunately, the latent positions $\X$ are unobserved. To contend with this, we will estimate $\X$ from the observed network, then plug in $\Xhat$ for $\X$ in subsequent regressions. We use the adjacency spectral embedding \citep[ASE;][]{sussman2012a} to achieve this estimation, but we note that several other methods are available \citep{xie2020,xie2021,wu2022}.

\begin{definition}[ASE; \cite{sussman2014}]
	\label{def:ase}
	Given a network with adjacency matrix $A$, the $d$-dimensional \emph{adjacency spectral embedding} (ASE) of $A$ is defined as
	\begin{equation*}
		\Xhat = \Uhat \Shat^{1/2} \in \R^{n \times d},
	\end{equation*}
	where $\Uhat \Shat \Vhat^T$ is the rank-$d$ truncated singular value decomposition of $A$. That is, $\Shat \in \R^{d \times d}$ is diagonal, with entries given by the $d$ leading singular values of $A$, and $\Uhat, \Vhat \in \R^{n \times d}$ have the corresponding $d$ orthonormal singular vectors as their columns.
\end{definition}

The spectral embeddings $\Xhat$ converge to the true $\X$ uniformly over the rows of $\X$, up to orthogonal non-identifiability \citep{levin2022a, lyzinski2014}. We thus construct least squares estimators for the regression coefficients based on $\Xhat$ rather than $\X$, as follows.

\begin{definition}[Regression Point Estimators]
	\label{def:reg-point-estimators}

	Define $\Dhat = \begin{bmatrix} \W & \Xhat \end{bmatrix} \in \R^{n \times (2 + p + d)}$. We estimate $\betaw$ and $\betax$ via ordinary least squares as follows
	\begin{equation*}
		\begin{bmatrix}
			\betawhat \\
			\betaxhat
		\end{bmatrix}
		= \paren*{\Dhat^T \Dhat}^{-1} \Dhat^T Y.
	\end{equation*}

	Similarly, we estimate $\Theta$ via ordinary least squares as
	\begin{equation*}
		\Thetahat
		= \paren*{\W^T \W}^{-1} \W^T \Xhat.
	\end{equation*}
\end{definition}

\begin{definition}[Regression variance estimators]
	\label{def:covariance}

	With notation as above, we define covariance estimators
	\begin{equation*} \begin{aligned}
			\Sigmahatbeta
			 & = \Abetahat^{-1} \cdot \Bbetahat \cdot \paren*{\Abetahat^{-1}}^T
			~~~\text{ and }                                                         \\
			\Sigmahattheta
			 & = \Athetahat^{-1} \cdot \Bthetahat \cdot \paren*{\Athetahat^{-1}}^T,
		\end{aligned} \end{equation*}
	where $I_d$ is a $d \times d$ identity matrix, and letting $\xihat_{i \cdot} = \Xhat_{i \cdot} - \W_{i \cdot} \Thetahat$, we define
	\begin{equation*} \begin{aligned}
			 & \Abetahat = \frac{\Dhat^T \Dhat}{n},
			 & \Bbetahat = \frac{1}{n} \sum_{i=1}^n \paren*{Y_i - \Dhat_{i \cdot} \betahat}^2 \Dhat_{i \cdot}^T \Dhat_{i \cdot}, \\
			 & \Athetahat = \frac{I_d \otimes \W^T \W}{n}, \quad \text{and}
			 & \Bthetahat
			= \frac{1}{n} \sum_{i=1}^n \xihat_{i \cdot}^T \, \xihat_{i \cdot} \otimes \W_{i \cdot}^T \W_{i \cdot}.
		\end{aligned} \end{equation*}
\end{definition}

Our main technical results state that the ordinary least squares estimates based on $\Xhat$ converge to the same asymptotic distribution as the ordinary least squares estimates based on $\X$. That is, the estimated regression coefficients asymptotically behavior as if we had access to the true latent positions, even though we do not. In order for plug-in estimation to be useful, the estimates based on the true latent positions $\X$ must themselves be well-behaved. Under fairly weak conditions, estimates based on the true latent positions are asymptotically normal.

\begin{assumption}[Regularity Conditions for Ordinary Least Squares $M$-estimation] \label{ass:regularity}
	We require standard regularity conditions for ordinary least squares estimates for both the outcome model and the mediator model. See Chapter 7 of \cite{boos2013} and Chapter 5 of \cite{vaart1998} for additional discussion.

	\begin{enumerate}
		\item  Define $\xi = \X - \E[\W]{\X} \in \R^{n \times d}$ to be the matrix of errors in the mediator regression. $\xi_{ij}$ and $\xi_{i' j}$ are independent for $i \neq i'$ and all $j \in [d]$, and that $\E[\W]{\xi_{ij}} = 0$ for all $i \in [n], j\in[d]$. Further,
		      \begin{equation*}
			      \Athetavec \equiv \E{\lim_{n \to \infty}
				      \frac{1}{n} \sum_{i=1}^n \W_{i \cdot}^T \W_{i \cdot}} \in \R^{p \times p}
		      \end{equation*}
		      exists and is non-singular, and the matrix $\Bthetavec \in \R^{d \times d}$ defined according to
		      \begin{equation*}
			      [ \Bthetavec ]_{j,j'} \equiv \E{\lim_{n \to \infty} \frac{1}{n} \sum_{i=1}^n \xi_{ij} \, \W_{i \cdot}^T \W_{i \cdot} \, \xi_{i j'}}
			      ~~~j,j' \in [d]
		      \end{equation*}
		      exists and has all entries finite.

		\item  Define $\Dhat = \begin{bmatrix} \W & \Xhat \end{bmatrix} \in \R^{n \times (2 + p + d)}$, let $\varepsilon = Y - \E[\D]{Y} \in \R^n$ be the vector of errors in the outcome regression. The $\varepsilon_i$ are independent, and obey $\E[\D]{\varepsilon_i} = 0$ for all $i \in [n]$. Further, $\Abeta \equiv \E{\lim_{n \to \infty} \frac{1}{n} \sum_{i=1}^n \D_{i \cdot}^T \D_{i \cdot}} \in \R^{(p + d) \times (p + d)}$ exists and is non-singular, and
		      \begin{equation*}
			      \Bbeta \equiv \E{\lim_{n \to \infty} \frac{1}{n} \sum_{i=1}^n \varepsilon_i^2 \, \D_{i \cdot}^T \D_{i \cdot}}
		      \end{equation*}
		      exists and is finite.
	\end{enumerate}
\end{assumption}

For consistent estimation of $\beta$ and $\Theta$ based on {\em known} $\X$, the key requirement is that the residuals $\varepsilon_i$ and $\xi_{ij}$ have (conditional) mean zero, given the covariates $\D$ and $\W$, respectively. Fundamentally, this means that $\E[\W_{i \cdot}, \X_{i \cdot}]{Y_i}$ must be linear in both $\W_{i \cdot}$ and $\X_{i \cdot}$ and that $\E[\W_{i \cdot}]{\X_{i \cdot}}$ must be linear in $\W_{i \cdot}$. We note that the errors $\varepsilon$ and $\xi$ need not come from a particular distribution: any distribution satisfying the assumptions on $\Athetavec, \Bthetavec, \Abeta$, and $\Bbeta$ will do. The rows of $\xi$ must be independent of one another, but arbitrary dependence is acceptable within the rows of $\xi$. Once the regularity conditions in Assumption~\ref{ass:regularity} are satisfied, these least-squares estimates for $\beta$ and $\Theta$ are asymptotically normal about their estimands.  Further, the well-known ``robust'' or ``sandwich'' covariance estimator is a consistent estimator for the covariance structures of these asymptotic distributions and the ordinary least squares estimates and covariance estimators can be used to obtain asymptotically valid confidence intervals for $\beta$ and $\Theta$.

Since we do not have access to the latent positions $\X$, we must estimate them using the adjacency matrix, and as a result we need additional structure beyond Assumption~\ref{ass:regularity}.

\begin{assumption}[Conditions for Two-stage Estimability]
	\label{ass:second-stage}

	With notation as above, we assume that
	\begin{enumerate}
		\item The smallest non-zero eigenvalue $\lambda_d$ of $P$ grows at a sufficiently fast rate as a function of $n$ relative to the noise parameters $\nu_n$ and $b_n$. In particular,
		      \begin{equation*}
			      \paren*{\nu_n + b_n^2} = \littleoh{\frac{\lambda_{d}^{2}}{n \log^2 n}}.
		      \end{equation*}

		\item The entries of the observation error vector $\varepsilon$ have bounded second moments, i.e., $\max_{i \in [n]} \E{\varepsilon_i^2} < B$ for $B > 0$ not depending on $n$. Further, we assume that the error vector is such that $n^{-1} \sum_i \E{\varepsilon_i^4}$ is bounded as a function of $n$.

		\item The rows $\W_{1 \cdot}, \W_{2 \cdot}, \dots, \W_{n \cdot}$ of $\W$ are independent. Within each row $\W_{i, \cdot}$, the elements may be dependent, but are marginally sub-Gaussian with fixed, shared parameter $\sigma > 0$.

		\item The latent positions $X_1,X_2,\dots,X_n$ are such that $n^{-1} \sum_i \E{\norm*{\xi_{i \cdot}}^2}$ and $n^{-1} \sum_i \E{\norm*{\X_{i \cdot}}^4}$ are bounded as functions of $n$.
	\end{enumerate}
\end{assumption}

The first condition is a sufficient condition for $\Xhat$ to concentrate around $\X$ \citep[see][for further discussion]{levin2022a}. This condition primarily places requirements on the density of edges in the network, and in particular requires that the expected average degree is $\omega(\sqrt n \log n)$. In a random dot product graph, all degrees are asymptotically of the same order as the average degree, so our results are for dense networks. See Remark \ref{rem:sparsity} for some additional comments on sparsity.

The second condition puts some additional (weak) conditions on the error in the outcome regression, beyond those already required for $M$-estimation. While the second condition strengthens Assumption~\ref{ass:regularity}, it still makes no distributional assumptions. For example, if the $\varepsilon_i$ are independent and identically distributed sub-gamma random variables, our bounded second moment condition is satisfied. The third condition is necessary for control over $\norm*{\W}$ in our proofs. The fourth-moment assumptions in the second and fourth conditions are needed to ensure convergence of our covariance estimator based on the spectral estimates. They ensure that error terms between this spectral-based covariance estimate and the covariance estimate based on the true (but unobserved) latent positions are suitably close. We anticipate that the moment bounds in the second, third and fourth conditions can be relaxed, but at the expense of additional proof complexity.

Our theoretical results, in their most general form, require certain growth rates on the sub-gamma parameters $\nu_n$ and $b_n$, and the largest and smallest non-zero eigenvalues $\lambda_1$ and $\lambda_d$ of $P$, described in Assumption \ref{ass:growth-rates}. For ease of presentation, we note that these rates are satisfied by random dot product graphs (and therefore stochastic blockmodels, since stochastic blockmodels are a submodel of random dot product graphs) and present Theorem~\ref{thm:main} in the setting of a random dot product graph. Fully general bounds in terms of $n, \nu_n, b_n, \lambda_1$ and $\lambda_d$ along  with proof details may be found in the Appendix.

\begin{theorem}
	\label{thm:main}
	If Example~\ref{ex:rdpg}, and Assumptions~\ref{ass:causal-linearity},~\ref{ass:regularity}, and~\ref{ass:second-stage} hold, then there exists a sequence of orthogonal matrices $\set{Q_n}_{n=1}^\infty$ such that
	\begin{equation*}
		\begin{aligned}
			\sqrt{ n } \,
			\Sigmahattheta^{-1/2}
			\begin{pmatrix}
				\vecc \paren*{\Thetahat \, Q_n^T} - \Thetavec
			\end{pmatrix}
			 & \to
			\Normal{0}{I_{p d}}, \text{ and } \\
			\sqrt{ n } \,
			\Sigmahatbeta^{-1/2}
			\begin{pmatrix}
				\betawhat - \betaw \\
				Q_n \, \betaxhat - \betax
			\end{pmatrix}
			 & \to
			\Normal{0}{I_d}.
		\end{aligned}
	\end{equation*}
\end{theorem}

Theorem~\ref{thm:main} shows that the ordinary least squares estimates based on $\Xhat$ are asymptotically normal about their estimands, up to orthogonal non-identifiability. Since $\Xhat$ only recovers $\X$ up to some unknown orthogonal transformation, $\Thetahat$ and $\betaxhat$ are only recovered up to this transformation, but $\betaw$ can be fully recovered.

\begin{remark}[Specifying the latent dimension $d$]
	\label{rem:model-selection}
	Theorem~\ref{thm:main} assumes, implicitly, that the latent dimension $d$ is known or consistently estimated. In general, estimating the dimension of the latent space $d$ is a challenging problem, but it can be addressed independently of network regression. Indeed, there are a variety of techniques for estimating the rank of random dot product graphs, such as those in \cite{chen2021, han2015, fishkind2013, landa2021, li2020c, han2020}, among others. Theoretically, any consistent estimator of the rank of a network suffices for Theorem~\ref{thm:main} to hold.

	Practically, we propose that analysts use a consistent estimator of $d$, but also that they conduct a sensitivity analysis to investigate how much results vary with the embedding dimension $d$. When results are indeed sensitive to $d$, we recommend erring on the side of over-estimating $d$. It is well-known in the random dot product graph literature that over-estimating the rank of $\X$ can lead to estimates $\Xhat$ that are still useful for downstream tasks \citep{fishkind2013}. We show via simulations that over-estimating $d$ still leads to interval estimates of $\nde$ and $\nie$ with correct coverage in the network mediation setting. We also note that these estimates tend to vary with embedding dimension $d$ when $d$ is under-estimated, but they stabilize once the embedding dimension has been reached. This suggests that a reasonable way to choose the embedding dimension $d$ is to look for a plateau in $\ndehat$ and $\niehat$ as a function of $d$ (see Section~\ref{sec:simulations} and Figure \ref{fig:bias_trajectories} for details).
\end{remark}

\begin{remark}[Finite sample bias]
	\label{rem:finite-sample-bias}
	The noise in $\Xhat$ around $\X$ induces bias in the estimates of $\betahat$ at finite sample sizes. In random dot product graphs, the coefficients $\betahat$ will be shrunken towards zero, as the rows of $\Xhat - \X \Q$ are approximately normally distributed with mean zero \citep{athreya2015}, such that standard results on noisily observed regressors hold. Asymptotically, however, the noise in $\Xhat$ around $\X$ vanishes, such that bias induced by measurement error disappears.
\end{remark}

\begin{remark}[Sparsity]
	\label{rem:sparsity}
	Under Assumption~\ref{ass:second-stage}, the average degree of a binary network must be $\omega(\sqrt n \log n)$, rather than the more typical $\omega(\log^c n)$. In turn, under Assumption \ref{ass:subgamma}, all degrees in the network are the same order, and so all degrees must be $\omega(\sqrt n \log n)$. This restriction to dense networks is a consequence of the proof technique, rather than a fundamental limit of the estimator: the sub-gamma bounds used in the present work are not tight when applied to networks with sparse Bernoulli edges. It is straightforward, however, to replace the sub-gamma bounds in our proofs with tighter bounds specialized to Bernoulli random variables~\citep{lei2015,boucheron2013} and to thereby relax assumptions on average degree in binary networks. Given our focus on general semi-parametric results, we do not perform these calculations here. See Remark 15 of \cite{levin2022a} for further discussion.
\end{remark}

\begin{remark}[Generalized linear models]
	\label{rem:M-generalization}
	We expect that results similar to Theorem~\ref{thm:main} hold for generalized linear models, and indeed for general regression $M$-estimators. Past work on random dot product graphs has established a functional central limit theorem for the latent positions $\Xhat$ \citep[Theorem 4]{tang2017} and convergence of $M$-estimates that are functions of $\Xhat$ only \citep[Theorem 4]{athreya2021}. These results should extend to the sub-gamma model and regression $M$-estimators, and we anticipate this to be a fruitful avenue for future work. In particular, these results would enable practitioners to use many of the estimators considered in \cite{nguyen2021c} to estimate network-mediated causal effects, or alternatively may allow incorporation of spatial or network error structures \citep{lesage2009,bramoulle2009}. Similarly, such results would imply that $\Xhat$ could be used in place of $\X$ in the regression-based estimators proposed in \cite{vanderweele2015}. Indeed, \cite{vanderweele2015} can be seen as a template for how our model could be extended to included non-linear terms or link functions.
\end{remark}

\subsection{Network Regression for Causal Estimation}

The semi-parametric identification results of Proposition~\ref{prop:identification-mediated} suggest a regression estimator for the natural direct and indirect effects.

\begin{definition}[Causal Point Estimators]
	\label{def:causal-estimators}
	To estimate $\nde$ and $\nie$, we combine regression coefficients from the network regression models
	\begin{equation*} \begin{aligned}
			\cdehat & = \ndehat = \paren*{t - t^*} \, \betathat      & \text{and} \\
			\niehat & = \paren*{t - t^*} \, \thetathat \, \betaxhat.
		\end{aligned} \end{equation*}
\end{definition}

\begin{remark}
	When $\Xhat = \X$, $\niehat$ reduces to the multivariate product-of-coefficients estimator introduced in \cite{vanderweele2014a}. As such, there are numerous methods for sensitivity analysis that can be immediately applied to $\ndehat$ and $\niehat$ \citep[Chapter 3]{vanderweele2015}.
\end{remark}

\begin{definition}[Causal Variance Estimators]
	\label{def:causal-variance}

	To estimate the variances of $\nde$ and $\nie$ in our semi-parametric setting, we combine coefficients from the network regression models:
	\begin{align*}
		\sigmahatnde & = \paren*{t - t^*}^T \cdot \Sigmahatbetat \cdot \paren*{t - t^*}
	\end{align*}
	\noindent where $\Sigmahatbetat$ denotes the element of $\Sigmahatbeta$ corresponding to $\betat$. Using analogous notation, let
	\begin{equation*}
		\sigmahatnie
		= \paren*{t - t^*}^T
		\begin{bmatrix}
			\betaxhat \\
			\thetathat
		\end{bmatrix}^T
		\begin{bmatrix}
			\widehat{\Sigma}_{\thetat} & 0                         \\
			0                          & \widehat{\Sigma}_{\betax} \\
		\end{bmatrix}
		\begin{bmatrix}
			\betaxhat \\
			\thetathat
		\end{bmatrix}
		\paren*{t - t^*}.
	\end{equation*}
\end{definition}

As with Theorem~\ref{thm:main}, we present Theorems \ref{thm:nde} and \ref{thm:nie} in the setting of a random dot product graph. Fully general versions may be found in the Appendix.

\begin{theorem}
	\label{thm:nde}

	In the setting of Example~\ref{ex:rdpg} under Assumptions~\ref{ass:mediation},~\ref{ass:causal-linearity},~\ref{ass:regularity} and ~\ref{ass:second-stage},
	\begin{equation*}
		\sqrt{n \, \sigmahatnde} \paren*{\ndehat - \nde}
		\to
		\Normal{0}{1}.
	\end{equation*}
\end{theorem}

The theorem follows from Definition~\ref{def:causal-estimators} and Theorem~\ref{thm:main} together with Slutsky's theorem and an application of the delta method. A similar distributional result holds for the natural indirect effect. Proofs for both results are given in the Appendix.

\begin{theorem} \label{thm:nie}
	In the setting of Example~\ref{ex:rdpg} under Assumptions~\ref{ass:mediation},~\ref{ass:causal-linearity},~\ref{ass:regularity} and ~\ref{ass:second-stage},
	\begin{equation*}
		\sqrt{n \, \sigmahatnie} \paren*{\niehat - \nie}
		\to
		\Normal{0}{1}.
	\end{equation*}
\end{theorem}
The form of the variance estimator $\sigmahatnie$ follows from an application of the delta method to the regression estimators $\betahat$ and $\Thetahat$, which can be shown to have an asymptotically normal joint distribution via a stacked M-estimator argument \citep{boos2013,nguyen2021c,vanderweele2014a, he2024b}. Note that rotational non-identifiability of the regression coefficients does not impact our ability to recover $\nie$, as the unknown matrix $Q$ cancels with a corresponding $Q^T$ in the product of regression coefficients.

\section{Simulations}
\label{sec:simulations}

We now turn to a brief exploration of our estimators' performance when applied to simulated data. In our results below, we find that our two-stage regression estimators are able to reliably recover regression coefficients and mediated effects, up to orthogonal non-identifiability where appropriate. We conduct simulations using two separate models to generate network structure, both based on the degree-corrected stochastic blockmodel.

We consider a degree-corrected SBM with $d$ blocks, $n$ nodes, and degree heterogeneity parameters $\gamma$ sampled from a continuous uniform distribution on the interval $[1, 3]$. Block assignment is random and nodes have equal probability of assignment to all blocks. The mixing matrix $B$ is set to $0.8$ on the diagonal, and $0.03$ off the diagonal, corresponding to strong assortative structure. Once the block memberships $\Z$, degree heterogeneity parameters $\gamma$ and mixing matrix $B$ are known, we compute the latent positions $\X$ numerically based on the singular value decomposition of $\E[\Z, \gamma, B]{A}$.

To generate data for our simulations, we first sample a network $A$ and latent positions $\X$ according to a degree-corrected stochastic blockmodel. Then we sample the nodal covariates $W$, according to one of two different models:
\begin{enumerate}
	\item In the ``uninformative'' model, the nodal covariates are three-dimensional samples from a standard multivariate normal distribution, independent of all other parameters in the model. These are combined with an intercept column. One of the Gaussian columns is taken to be the treatment and the others are taken to be controls.
	\item In the ``informative'' model, the nodal covariates are dummy-coded block membership indicators, using treatment coding and including an intercept column. The treatment $T$ is taken to be the column corresponding to the indicator for the second block, and the controls are taken to be all other block membership indicators.
\end{enumerate}
Then, we infer the implied mediator coefficients $\Theta$ via a linear regression of nodal covariates on the latent positions\footnote{This process may seem counter-intuitive, since it gives up precise control over the mediator coefficients $\Theta$. The upside is that we do not need to specify a generative model for the mediator regression errors $\xi$. Specifying $\xi$ is challenging in binary networks where $\X$ must follow an inner product distribution to maintain $P_{ij} \in [0, 1]$.}. In the uninformative model, there is no association between $\W$ and $\X$ on average, so $\W$ is only idiosyncratically associated with $\X$. In the informative model, $\W$ is a coarsened version of $X$ where degree-correction information has been omitted, and so we expect a strong association between nodal covariates and latent positions.

Next we sample $\beta$ from a multivariate Gaussian distribution with mean equal to the vector of all ones, and covariance equal to a diagonal matrix with $1/4$ on the diagonal. Finally, to generate the nodal outcomes, we generate errors $\varepsilon$ from a $t_5$ distribution and use $\W, \X, \beta$ and $\varepsilon$ to produce $Y$ satisfying the regression condition in Assumption~\ref{ass:causal-linearity}. At this point, we also determine the induced direct and indirect effects based on $\beta$ and $\Theta$ via Equations \eqref{eq:nde} and \eqref{eq:nie}.

For each model, we sample $(A, Y, \W)$ for varying number of nodes $n$ and latent dimensions $d$, and compute point estimates and confidence intervals for $\Thetahat, \betahat, \ndehat$ and $\niehat$. We repeat this procedure 100 times for each combination of parameters. We focus here on the causal estimators $\ndehat$ and $\niehat$. Refer to the Appendix for further results on the consistency and finite sample bias of the regression coefficients.

In Figure~\ref{fig:loss_average}, we consider the mean squared error of $\ndehat$ and $\niehat$. We observe that the point estimates $\ndehat$ and $\niehat$ converge to $\nde$ and $\nie$, as expected per Theorem \ref{thm:main}. In Figure~\ref{fig:causal_coverage} we observe that the proposed asymptotic confidence intervals achieve close to their nominal coverage rates in finite samples. This verifies that variance estimators accurately quantify the uncertainty in $\ndehat$ and $\niehat$, also as expected given Theorems \ref{thm:nde} and \ref{thm:nie}. In the uninformative setting, coverage for the indirect effect is higher than the nominal rate, which is unsurprising, given that confidence intervals for the indirect effect based on the delta method can be overly conservative \citep{he2024b}.

In Figure \ref{fig:misspecification_coverage}, we investigate the coverage of our asymptotic confidence intervals when the rank $d$ of the network is misspecified. We see that underestimating the latent dimension dramatically degrades coverage of confidence intervals of $\nde$ and $\nie$. However, when $d$ is overestimated, confidence intervals for $\nde$ and $\nie$ obtain nominal coverage rates. The negative effect of underestimating $d$ is more pronounced in the informative model, where treatment is strongly associated with latent position in the network, and weaker in the uninformative model, where treatment is weakly associated with latent position in the network. On the basis of these results, we suggest that practitioners err on the side of over-estimating, rather than under-estimating, the rank $d$ of the latent positions $\X$. Intuitively, as $d$ increases, $\Xhat$ captures more and more of the latent community structure in the network, until eventually $\Xhat$ captures all the latent structure in the network and the effects stabilize (see Figure~\ref{fig:misspecification_mean_squared_error} in the Appendix for additional simulation results in this vein).

Lastly, in Figure \ref{fig:bias_trajectories}, we investigate the bias of $\ndehat$ and $\niehat$ as a function of the embedding dimension $d$. We observe that estimates of $\ndehat$ and $\niehat$ vary with the embedding dimension $d$ when $d$ is under-estimated. However, once the embedding dimension $d$ is correctly specified, the estimates $\ndehat$ and $\niehat$ stabilize. This suggests that practitioners can use sensitivity curves (such as those in Figures \ref{fig:glasgow-estimates} and \ref{fig:hm-curve}) to estimate that embedding dimension $d$: in particular, they should look for the embedding dimension $d$ that stabilizes the estimates $\ndehat$ and $\niehat$; any estimate using this embedding dimension or higher is likely to have good coverage properties.

\begin{figure}[t!]
	\centering
	\includegraphics[width=0.9\textwidth]{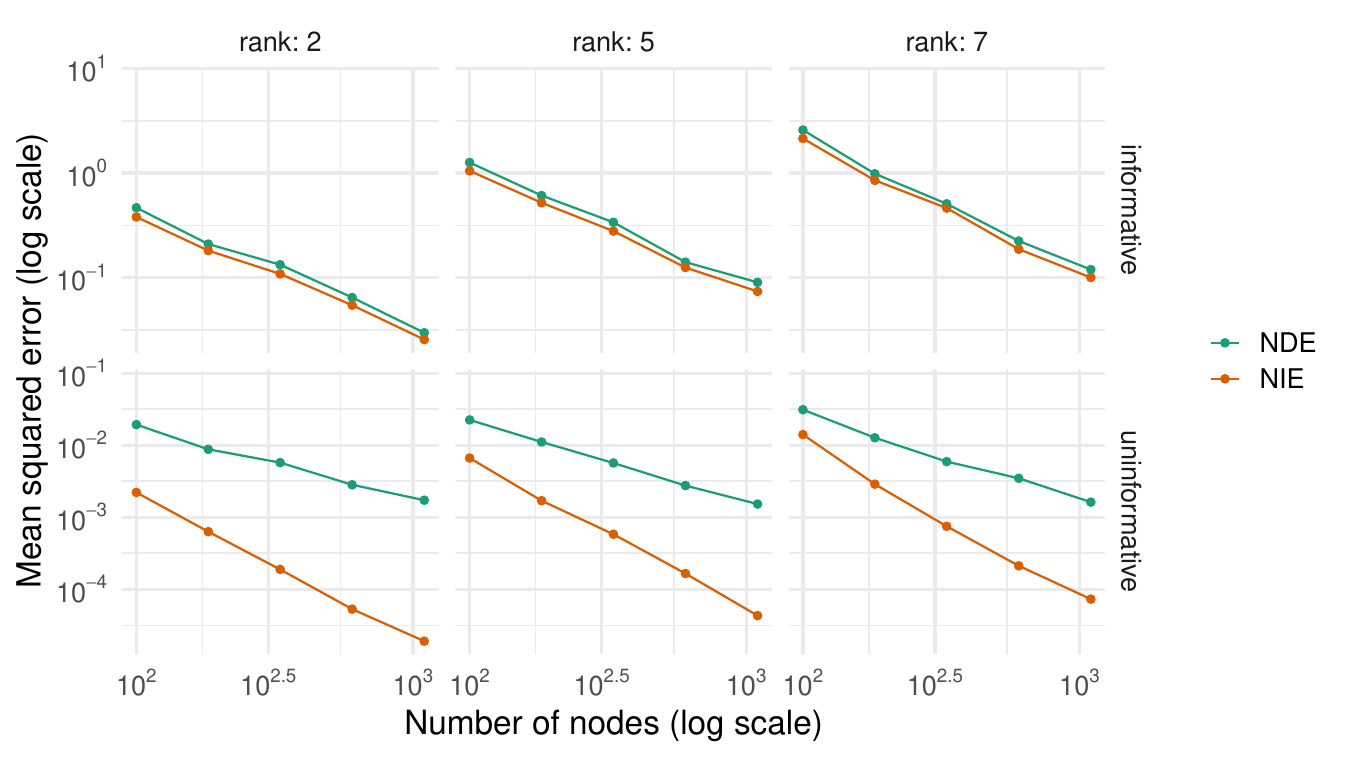}
	\caption{Convergence of $\ndehat$ to $\nde$ and $\niehat$ to $\nie$. Each panel shows the mean squared error (vertical axis, log scale) of $\ndehat$ (teal) and $\niehat$ (orange) as a function of the number of nodes in the network (horizontal axis, log scale). Panels vary horizontally by number of latent communities (left: two blocks, middle: five block, right: seven blocks) and vertically by the simulation model (top: informative, bottom: uninformative).}
	\label{fig:loss_average}
\end{figure}

\begin{figure}[ht!]
	\centering
	\includegraphics[width=0.9\textwidth]{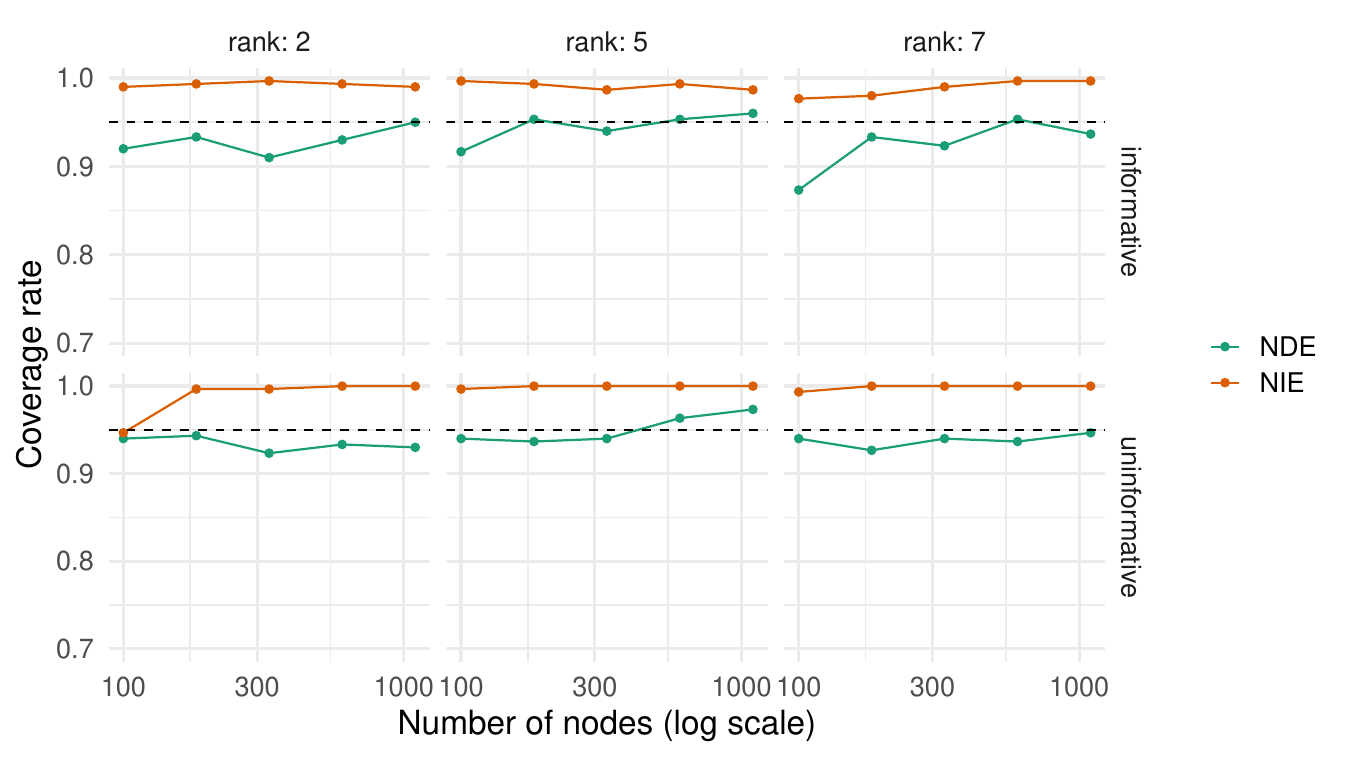}
	\caption{Finite sample coverage of asymptotic confidence intervals for $\nde$ and $\nie$. Each panel shows coverage (vertical axis) of $\nde$ (teal) and $\nie$ (orange) as a function of the number of nodes in the network (horizontal axis, log scale). The dashed horizontal line denotes the nominal coverage rate of 95\%. Panels vary horizontally by number of latent communities (left: two blocks, middle: five block, right: seven blocks) and vertically by the simulation model (top: informative, bottom: uninformative).}
	\label{fig:causal_coverage}
\end{figure}

\begin{figure}[ht]
	\centering
	\includegraphics[width=0.8\textwidth]{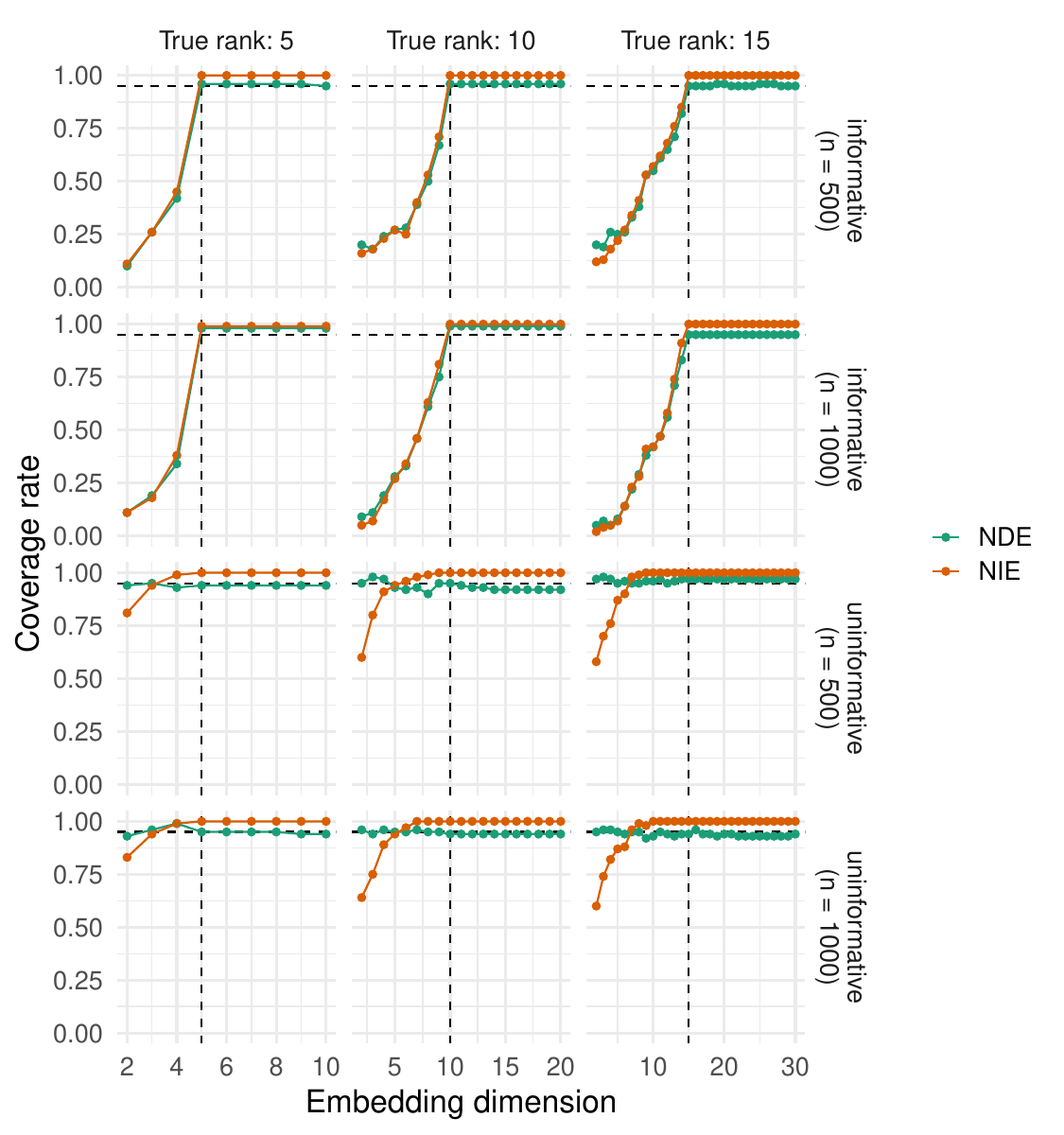}
	\caption{Coverage of confidence intervals for $\nde$ and $\nie$ when the dimension $d$ is misspecified. Each panel shows coverage (vertical axis) of $\nde$ (teal) and $\nie$ (orange) as a function of the embedding dimension $d$ (horizontal axis). The dashed horizontal line denotes the nominal coverage rate of 95\% and the dashed vertical line denotes the true latent dimension. Panels vary horizontally by number of latent communities (left: five, middle: ten, right: fifteen) and vertically by the simulation model and number of nodes in the network.}
	\label{fig:misspecification_coverage}
\end{figure}

\begin{figure}[ht]
	\centering
	\includegraphics[width=0.8\textwidth]{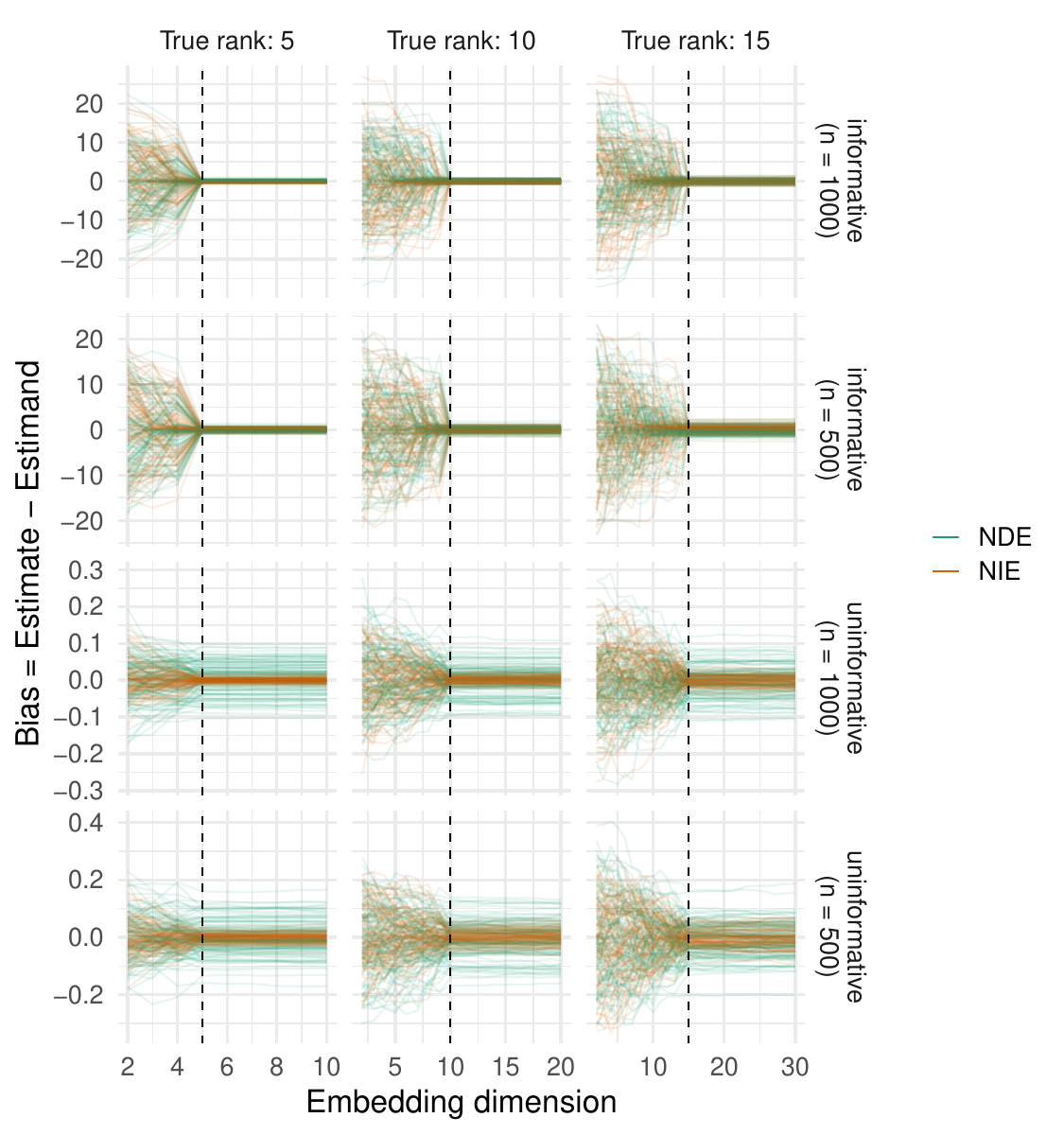}
	\caption{Stability of point estimates for $\nde$ and $\nie$ when the dimension $d$ is misspecified. Each panel shows bias (vertical axis) of $\ndehat$ (teal) and $\niehat$ (orange) as a function of the embedding dimension $d$ (horizontal axis). The dashed vertical line denotes the true latent dimension. Panels vary horizontally by number of latent communities (left: five, middle: ten, right: fifteen) and vertically by the simulation model and number of nodes in the network.}
	\label{fig:bias_trajectories}
\end{figure}

\section{Data Applications}
\label{sec:applications}

We now illustrate our method by applying it to two data sets, one previously considered by \cite{dimaria2022a}, and the other previously considered by \cite{hirshberg2022, hirshberg2024}.

\subsection{Smoking in an Adolescent Social Network} \label{subsec:glasgow}

We first revisit the Teenage Friends and Lifestyle Study described in Section \ref{subsec:motivating-example}, focusing on the causal effect of sex on smoking during the first wave of the study. Recall that the social network was collected by asking students ``who are your best friends'', and allowing them to list up to six responses. Sex and tobacco use were self-reported as nominal features with levels ``Male'' and ``Female''; and ``Never'', ``Occasional,'' and ``Regular,'' respectively. To match the analysis of \cite{dimaria2022a}, for the tobacco use measure we combined ``Occasional'' and ``Regular'' into a single level, and compared smokers with non-smokers. We treated age (continuous) and church attendance (nominal) as possible confounders.

We began by computing the adjacency spectral embedding of the social network $A$. In the Glasgow data, the social network is directed: an edge $i \to j$ indicates that student $i$ listed student $j$ as friend. This directedness means that students have two distinct co-embeddings corresponding to their propensity to send out-edges and receive in-edges. Letting $\widehat{A} \approx \Uhat \Shat \Vhat^T$ be the truncated singular value decomposition of $A$, the left co-embedding $\Xhat = \Uhat \Shat^{1/2}$ describes how students in the network send edges (i.e., claim friends), and the right co-embedding $\Lhat \equiv \Vhat \Shat^{1/2}$ describes how students receive edges (i.e., are claimed as friends). Here we report results using the right co-embeddings $\Lhat$.  We did not select any particular dimension $d$ for the latent space. Instead, we repeated our analysis for many values of $d$, to investigate the sensitivity of our results to the dimension of the latent space (see Remark \ref{rem:model-selection} and the simulation study in Section \ref{sec:simulations} for additional commentary).

Once we obtained embeddings $\Lhat$ via the singular value decomposition, we performed two ordinary least squares regressions. Using the formula notation of \cite{wilkinson1973} to specify the design and outcome matrices, we obtained least squares estimates for the following specifications:
\begin{equation*}
	\begin{aligned}
		\texttt{smoking} & \sim \texttt{sex + age + church + Fhat} \\
		\texttt{Fhat}    & \sim \texttt{sex + age + church}.
	\end{aligned}
\end{equation*}
We then combined the regression coefficients (and covariances) per Definitions~\ref{def:causal-estimators} and~\ref{def:causal-variance} to obtain point and interval estimates for the natural direct and indirect effects of sex on tobacco use. We visualized these results as a function of the embedding dimension $d$ in Figure~\ref{fig:glasgow-estimates}.

\begin{figure}[ht!]
	\centering
	\includegraphics[width=0.7\textwidth]{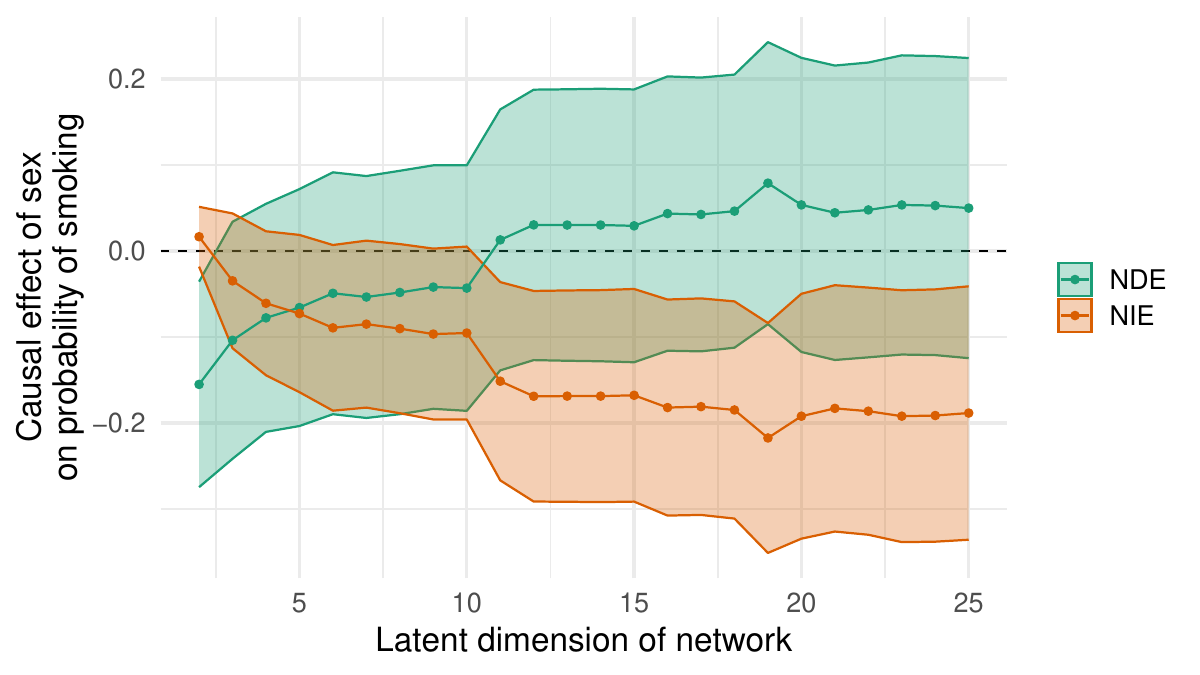}
	\caption{Estimated direct (teal) and indirect (orange) effects of sex on tobacco usage in the Glasgow social network. The estimated effects (vertical axis) vary with the dimension $d$ of the latent space (horizontal axis), and are adjusted for possible confounding by age and church attendance. Positive values indicate a greater propensity for adolescent boys to smoke, while negative values indicate a greater propensity for adolescent girls to smoke.}
	\label{fig:glasgow-estimates}
\end{figure}

The estimates of $\ndehat$ and $\niehat$ stabilized as a function of the embedding dimension around $d = 12$, and the qualitative interpretation of the results is effectively the same for all $d \ge 12$. Since over-estimating $d$ results in less estimation error than under-estimating $d$ (see Remark \ref{rem:model-selection} and simulation results in Section \ref{sec:simulations}), we first interpreted interval estimates when the latent dimension of the network is $d = 15$, under the assumption that $d = 15$ is correctly specified. For this latent dimension, we estimated a 95\% CI for the direct effect of male sex, relative to female sex, to be $(-0.14, 0.17)$ and a 95\% CI for the indirect effect to be $(-0.28, -0.04)$. That is, the estimates are consistent with no direct effect of sex on probability of tobacco usage.  Nonetheless, there is substantial uncertainty in the estimate of the direct effect: the data is consistent with a direct effect in either direction of up to roughly $0.15$. In contrast, the estimates are consistent with a negative indirect effect, though the scale of this effect is fairly uncertain.

There are, however, several reasons to be cautious about these estimates. First and foremost, we did not have auxiliary information about the social network that allowed us to directly interpret the embeddings $\Lhat$. This potentially leads to an issue with ill-defined interventions on $\L$, as we discussed in Section \ref{subsec:causal-interpretation}. While substantial sociological research confirms the presence of meaningful social groups in the social network \citep{michell2000a}, as well as the fact that smoking predominantly occurs in majority female social groups \citep{michell2000a, michell1997a}, we could not verify that the social groups observed by sociologists match the social groups encoded by $\Lhat$. We must hope that low-rank structure is an appropriate way to capture these social groups. There is a second issue, namely that we have no particularly compelling way to adjudicate between the sending co-embeddings $\Xhat$ and the receiving co-embeddings $\Lhat$. Our choice of $\Lhat$ is based primarily on folklore and personal experience that, in social networks, the receiving co-embeddings $\Lhat$ are more informative than the sending co-embeddings $\Xhat$. A sensitivity analysis using $\Xhat$ in place of $\Lhat$ yields smaller estimates of the indirect effect, which are not statistical distinguishable for zero (Figure \ref{fig:glasgow-estimates-u} in the Appendix). A third and final reason to be cautious about this analysis is that the positivity assumption may be violated (Figure \ref{fig:glasgow-positivity}, also in the Appendix).

In light of these considerations, we consider this analysis to be illustrative. That said, the results in Figure \ref{fig:glasgow-estimates} do align with previous sociological analyses, as well as the results obtained in \cite{dimaria2022a}, who used a related estimator to conclude that ``the probability of smoking tobacco regularly is higher for 13-year-old girls than for boys. In contrast, the total indirect effect for girls is negative, [such that] the effect of gender on the chance of smoking reduces through friendship relationships.'' Collectively, these analyses are suggestive of potential public health interventions, which could be further investigated. To reduce smoking, a public health intervention could focus on the indirect causal pathway, and could intervene on either the friend group formation process, or localized smoking within friend groups. For instance, students who smoke could be encouraged to connect socially with students who do not smoke. Alternatively, public health campaigns could focus on locating social groups where smoking is prevalent, and then performing more resource intensive interventions on those social groups.

\subsection{Psychological Mediators of Anxiety in a Randomized Controlled Trial on Meditation} \label{subsec:healthy-minds}

We next use our method to re-analyze data from a randomized controlled trial of a smartphone-based well-being training called the Healthy Minds Program, originally reported in \cite{hirshberg2022}. We call this trial the Healthy Minds trial for convenience, but note that numerous trials have been conducted based on the Healthy Minds smartphone app and methodology.

In the trial, 662 adults were randomly assigned to a four-week meditation program or a control condition. During the intervention, participants were surveyed on a weekly basis to assess their psychological well-being (psychological distress, anxiety and depression) and four anticipated psychological mediators of well-being (mindful action, loneliness, cognitive defusion and purpose). Participants were also surveyed three months after the end of the intervention period.

Meditation based interventions are known to improve anxiety and depression, an observation empirically verifiable in the Healthy Minds data itself (see Figure \ref{fig:hm-weekly-ate}). Further, there are theoretical reasons to believe that meditation can improve mindful action, loneliness, cognitive defusion and purpose, and that improvements in these dimensions can reduce depression and anxiety. In the four week long intervention program, one week is devoted to improving each of the hypothesized psychological mediators. The main goal of the Healthy Minds study was to investigate these mechanisms, and to improve knowledge of psychological mechanisms in order to design more effective interventions. We re-analyzed the data with this same goal in mind, focusing on the anxiety outcome at the end of the four week intervention. Following the original analysis, we control age and sex as potential confounders.

\begin{figure}[t]
	\centering
	\includegraphics{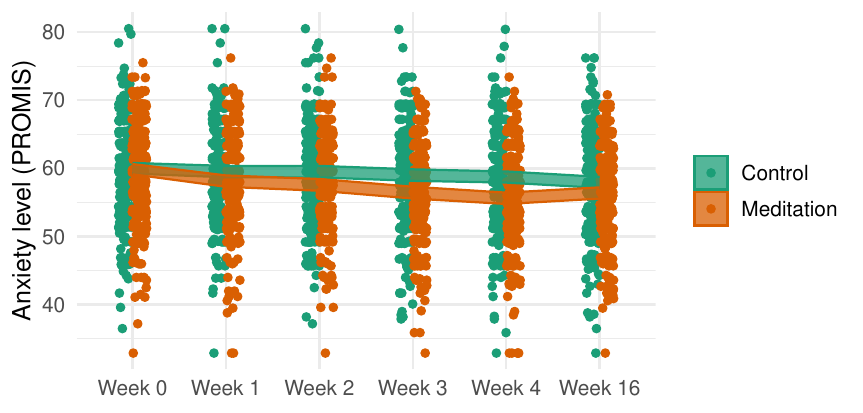}
	\caption{Anxiety levels (measured by the computer adaptive PROMIS test) for intervention and control groups in the Healthy Minds Study. Ribbons represent 95\% confidence intervals for mean anxiety level in each group at each time point.}
	\label{fig:hm-weekly-ate}
\end{figure}

While the Healthy Minds data at first glance seems unrelated to social networks, there is a close connection to network data, as psychological constructs are typically considered latent constructs that must be measured via surveys followed by factor analysis \citep{rohe2023}. In the Healthy Minds study, mindful action was measured via the Five Facet Mindfulness Questionnaire Act with Awareness subscale (8 questions), loneliness via the NIH Toolbox Loneliness Questionnaire (5 questions), defusion via the Drexel Defusion Scale (10 questions), and purpose via the Life Questionnaire Presence subscale (10 questions). We can represent the survey responses as a bipartite network with adjacency matrix $A \in \R^{533 \times 33}$, where $A_{ij}$ denotes participant $i$'s response to survey question $j$ (for convenience, we only consider the 533 study participants who responded to all survey questions at the end of the intervention period). The survey questions were validated as measures of the corresponding latent constructs at the time that the surveys were developed.

\begin{figure}[t]
	\centering
	\includegraphics{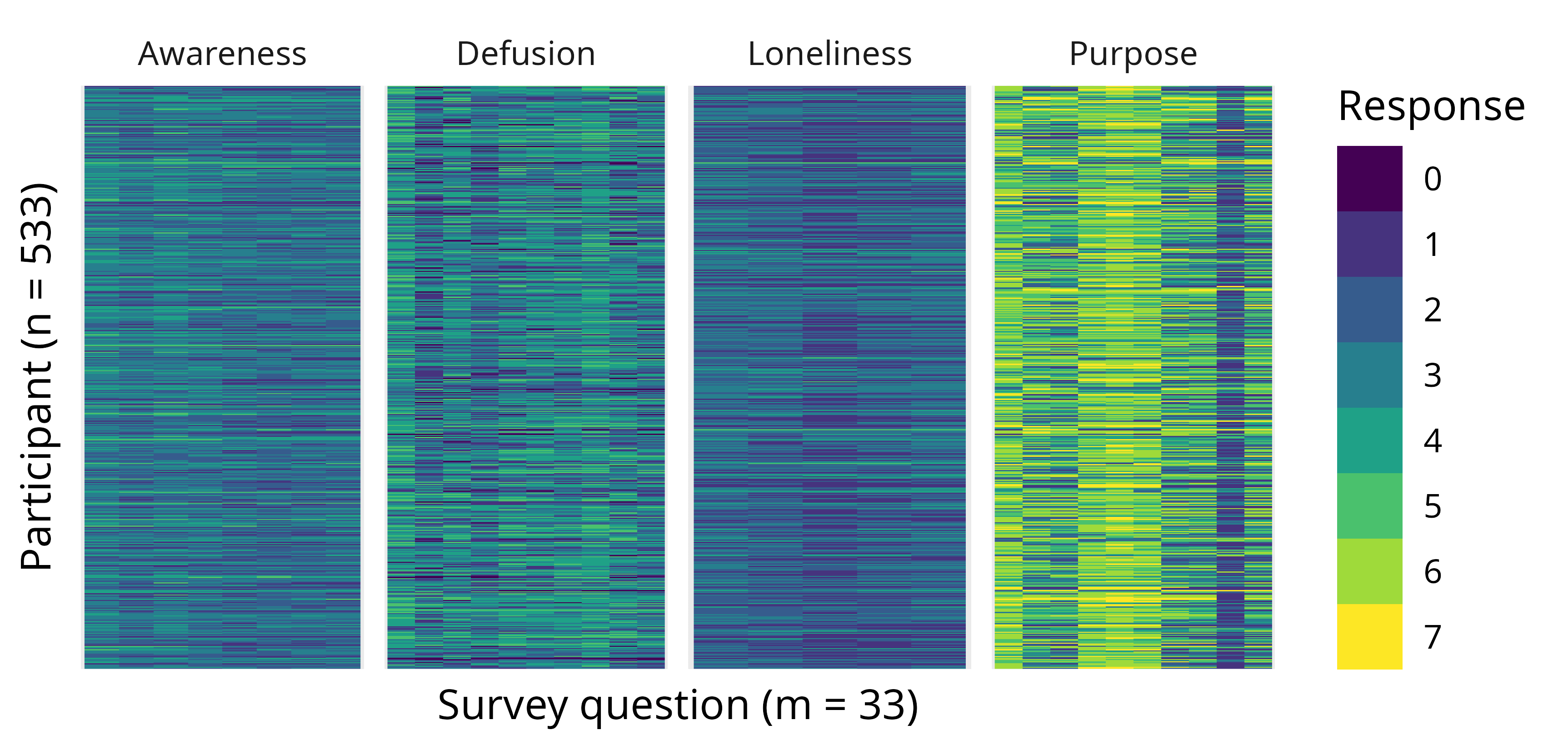}
	\caption{Survey responses for mediator measures at the end of four weeks.}
	\label{fig:hm-responses4}
\end{figure}

\begin{figure}[ht!]
	\centering
	\includegraphics{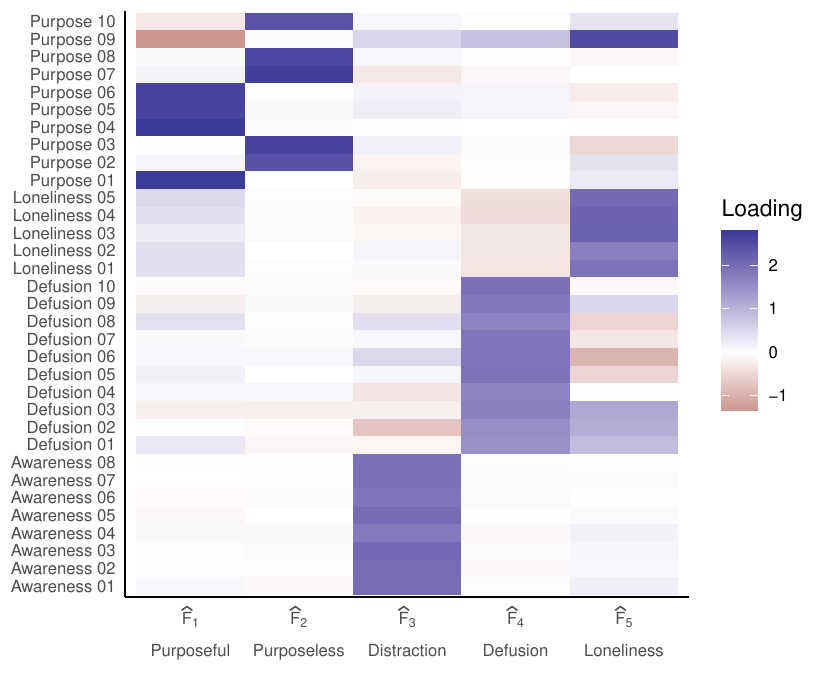}
	\caption{Survey item embeddings $\Lhat$ based on the survey responses at the end of four weeks.}
	\label{fig:hm-yhat}
\end{figure}

Since $A$ is rectangular, we compute a decomposition $A = \Xhat \Lhat^T$, where $\Xhat = \Uhat \Shat^{1/2}$ and $\Lhat  = \Vhat \Shat^{1/2}$. Since the rows of $A$ correspond to participants and the columns of $A$ correspond to survey items, the left co-embeddings $\Xhat$ describe participants, and the right co-embeddings $\Lhat$ describe survey items. When the dot product of $\Xhat_{i \cdot} \Lhat_{j \cdot}^T$ is large, that indicates that participant $i$ is expected to give a large response (e.g., ``absolutely agree'' rather than ``neither agree nor disagree'') to survey item $j$. When participants $i$ and $i'$ have embeddings $\Xhat_{i \cdot}$ and $\Xhat_{i' \cdot}$ that are close to each other, this indicates that they tended to respond to survey items in a similar manner. When survey items $j$ and $j'$ have embeddings $\Lhat_{j \cdot}$ and $\Lhat_{j' \cdot}$ that are close to each other, this means that participants responded to questions $j$ and $j'$ in a similar manner.

Our hope is that $\Lhat$ captures the hypothesized mediating constructs. Investigating if this is the case complicated by the fact that $\Xhat$ and $\Lhat$ are both subject to orthogonal non-identifiability, since $A = \X Q Q^T \L^T$ for any orthogonal $Q$. To interpret the question embeddings, we varimax rotate the right co-embeddings, as described in \citet{rohe2023}. That is, we compute a varimax rotation $R$ based on the unscaled right singular vectors $V$ and then take $\Xhat = \Uhat \Shat R$ and $\Lhat = R^T \Vhat^T$. Under the assumption that the embeddings are leptokurtic (i.e., more skewed than a Gaussian), the varimax rotated embeddings are identified up to column permutations and sign flips, making them much easier to interpret. We visualize the varimax-rotated results in Figure \ref{fig:hm-yhat}, where we see that a five dimensional embedding yields highly interpretable latent factors, which explain 70\% of the variance in the survey responses. The factors $\Lhat_{\cdot 3}, \Lhat_{\cdot 4}, \Lhat_{\cdot 5}$ load primarily on the awareness, defusion and loneliness survey questions.  The factors $\Lhat_{\cdot 1}$ and $\Lhat_{\cdot 2}$ load primarily on the purpose survey questions. However, two factors are need to capture the purpose latent construct because the survey questions are written with alternating valences:
\begin{enumerate}
	\setlength{\itemsep}{0pt}
	\setlength{\parskip}{0pt}
	\item I understand my life's meaning.
	\item I am looking for something that makes my life feel meaningful.
	\item I am always looking to find my life's purpose.
	\item My life has a clear sense of purpose.
	\item I have a good sense of what makes my life meaningful.
	\item I have discovered a satisfying life purpose.
	\item I am always searching for something that makes my life feel significant.
	\item I am seeking a purpose or mission for my life
	\item My life has no clear purpose.
	\item I am searching for meaning in my life.
\end{enumerate}
Items 1 and 2, for instance, are coded with opposite valences, but responses for all questions were on an integer scale from 1 (absolutely untrue) to 7 (absolutely true). Larger numerical responses to item 1 indicate a greater sense of purpose, so we call $\Lhat_{\cdot 1}$ the ``purposeful'' factor. Larger numerical responses to item 2 indicate an absence of purpose, so we call $\Lhat_{\cdot 2}$ the ``purposeless'' factor. Interestingly, the $9$-th item ``My life has no clear purpose' is not particularly associated with the purposeful and purposeless factors $\Lhat_{\cdot 1}$ and $\Lhat_{\cdot 2}$, but rather the loneliness factors $\Lhat_{\cdot 5}$ (survey items for other scales are available in Appendix \ref{app:healthyminds}). Using similar reasoning, we name $\Lhat_{\cdot 3}$ the ``distraction'' factor rather than the awareness factor, because larger numerical responses indicate lower levels of awareness (see Appendix \ref{app:healthyminds} for the survey items and response scale).

In order to conduct the mediation analysis, we fit the following two regression models on the study participants:
\begin{equation*}
	\begin{aligned}
		\texttt{anxiety} & \sim \texttt{meditation + sex + age + Xhat} \\
		\texttt{Xhat}    & \sim \texttt{meditation + sex + age}.
	\end{aligned}
\end{equation*}

\begin{figure}[t!]
	\centering
	\includegraphics{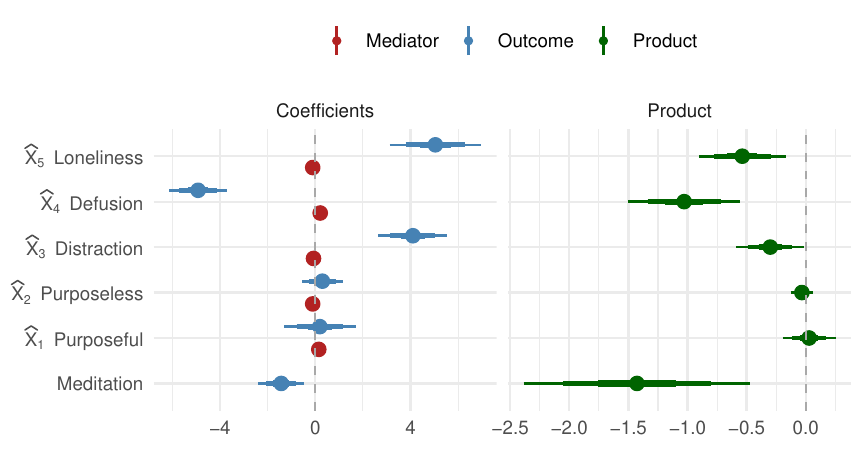}
	\caption{Left: Coefficients from the Healthy Minds mediator and outcome regression models, visualized as point intervals with interval widths 0.5, 0.8, 0.95, respectively. Coefficients for control variables not visualized. Intervals in the mediator model are very close to the point estimates, and thus not visible. All mediator coefficients are statistically distinguishable from zero. Right: confidence intervals for the corresponding coefficient products. The treatment coefficient, labelled ``Meditation,'' is simply repeated from the left panel for a comparative sense of scale.}
	\label{fig:hm-coefs}
\end{figure}

We visualize the coefficients for both regressions in Figure \ref{fig:hm-coefs}. First, consider the mediator model coefficients. The intervention has a small but statistically significant effect on each of the latent mediators: it reduces loneliness, increases cognitive defusion, decreases distraction, decreases purposelessness and increases purposefulness. This suggests that the Healthy Minds Program is effectively improving the theorized psychological mediators. However, not all of the psychological mediators cause anxiety. Both coefficients for purpose-related factors are consistent with zero causal effects. Loneliness and distraction (i.e., low awareness) both increased anxiety, and cognitive defusion (i.e., the ability to step back from thoughts and feelings and reflect on them) reduced anxiety. Meditation also appears to have a direct anxiety reducing effect, independent of its impact on the psychological mediators.

Considering the causal pathways, and the outcome regression and the mediation regression in aggregate, we estimate 95\% confidence intervals for the natural direct effect as $[-2.38, -0.45]$ and for the natural indirect effect as $[-2.69, -1.05]$. That is, about half the anxiety reducing effect of the intervention was due to indirect effect, and about half was due to the direct effect. A natural followup question is if products of factor-specific regression coefficients can be interpreted as factor-specific indirect effects.
This is the case, provided that we make the additional assumption that the latent psychometric factors are independent of one another conditional on intervention and controls. We visualize these products and their associated confidence intervals in the right panel of Figure \ref{fig:hm-coefs}.

If conditional independence of the factors is plausible, the defusion and loneliness pathways seem most important for reducing anxiety. The distraction estimate is smaller and nearly compatible with a null effect. The purpose effects are both fairly precisely measured zeroes. This factor-specific effects might motivate more effective meditation interventions. Each of the four weeks of the Healthy Minds intervention is devoted to improving one of the four hypothesized mediators, and our results suggest that it could be beneficial to replace the purpose module of the Healthy Minds program with an alternative, possibly doubling the time devoted to defusion or social connectedness skills. Alternatively, it may be worthwhile to investigate whether or not the factors are indeed independent; we find it particularly interesting that the Purpose 09 item (``My life has no clear purpose'') was so strongly associated with the loneliness factor, possibly suggesting a connection between sense of meaning and social connectedness that could be leveraged in future interventions.

The analysis thus far has exclusively considered the survey responses and anxiety levels at the end of the four week intervention period.  Participants were surveyed before the intervention period, then weekly for four weeks during the intervention period, and then three months after the end of the intervention. As a final step, we repeated our mediation analysis, considering each time point in the study independently. We compute direct and indirect effects at each stage of the study, and then plotted them in Figure \ref{fig:hm-trajectory}. We see that the direct effect appears at week 1 of the intervention, and persists at the same level throughout the program, before fading back down three months post intervention. This leads us to believe that the direct effect is capturing the beneficial effect of calming breathing exercises, which have immediate and short term benefits. In contrast, the indirect effect grows slowly over time, exactly as one would expect if participants were slowly learning new and healthier habits of mind. These effects also seem to persist more strongly after the end of the intervention period. This suggests that a longer program might have more beneficial effect, and perhaps that the larger indirect effects may persist longer.

\begin{figure}[ht!]
	\centering
	\includegraphics[width=0.7\textwidth]{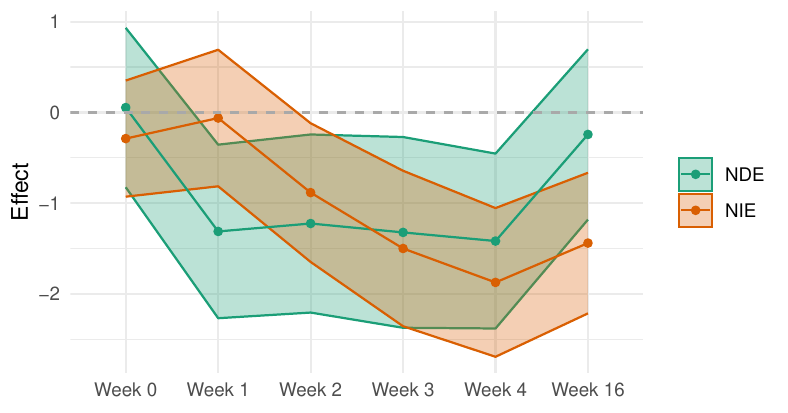}
	\caption{Natural direct and indirect effects of meditation on anxiety level, analyzed independently at each week in the Healthy Minds Trial, controlling for age and sex as confounders, and using a five dimensional embedding of survey responses at each time point.}
	\label{fig:hm-trajectory}
\end{figure}

\section{Discussion}
\label{sec:discussion}

In this paper, we have explored the use of principal components network regression for analyzing structured data and its potential applications in causal inference. We highlight four main takeaways from our research.

\textbf{Use principal components network regression.} Principal components network regression is distributionally agnostic, robust to noisily observed networks, and computationally straightforward. The general low-rank sub-gamma model that we consider accommodates a wide variety of parametric submodels, as well as noisily observed networks, and is appropriate for many kinds of structured data. We recommend including spectral network embeddings in ordinary least squares regressions. We have shown that including spectral network embeddings in ordinary least squares estimators only requires semi-parametric assumptions.  Asymptotically, it is equivalent to observing latent low-rank structure in the network and including population structure in regressions, although in finite samples there may be some bias induced by estimation error. Low-rank network regressions are consistent and asymptotically normal under weak and distribution-free assumptions. Although some regression coefficients are subject to an unknown orthogonal transformation, in practical applications this may not matter, or may be resolvable with varimax rotation \citep{rohe2023}.

\textbf{Principal components network regression can be used for causal inference.} Principal components network regression is useful in observational settings, and additionally for causal inference. We have carefully detailed the counterfactual and statistical assumptions required for regression coefficients to have causal interpretations. When latent network positions are mediators, coefficients from principal components network regression can be used in the product-of-coefficients method for estimating natural direct and indirect effects.  Much of the causal intuition for tabular regressions also applies to our network regressions \citep{vanderweele2014a}, but there some network-specific concerns such as homophily causing positivity violations, and the presence or absence of peer effects.

\textbf{Social groups can act as mediators, not just confounders.} In other words, latent homophily can have a mediating effect rather than a confounder effect. As a result, it is critical to carefully consider the role of latent positions in causal settings, because causal analysis should adjust for confounders, but should not adjust for mediators. Mistaking a mediator for a confounder and then adjusting for the mediator induces over-control bias. We empirically demonstrate how incorrect assumptions about the direction of causation can induce dramatic overcontrol bias into causal estimates. We suspect overcontrol bias is most likely to be an issue when considering causal effects of demographic features, which are likely to induce homophily in networks.

\textbf{Network embeddings need to be interpreted.} Network embeddings are not black boxes, magical controls, or inherent mediators. Network embeddings are estimates of latent structure in networks, and compelling inference using network embeddings requires an interpretation of that latent structure. It is not enough to claim that network embeddings capture the homophily relevant to a particular causal pathway. In applied projects, one must confirm that the network embeddings capture the important latent constructs, rather than noise, or latent constructs other than those originally hypothesized. As an important corollary, practitioners should not use network embeddings for causal inference unless they have sufficient domain knowledge or auxiliary data to give the embeddings a concrete and substantive interpretation. Pragmatically, interpreting network embeddings is complicated by orthogonal non-identifiability, but this identification challenge can often be resolved with varimax rotation \citep{rohe2023}, or mixture modeling in the latent space \citep{rubin-delanchy2022}.

\section*{Reproducibility}

A replication package for our simulations and data analysis is available at \url{https://github.com/alexpghayes/network-mediation-replication}.

\acks{We thank Felix Elwert, Hyunseung Kang, Karl Rohe, Steven Wright, S\'ebastien Roch, Sameer Deshpande, Adeline Lo, Yehzee Ryoo, Tianxi Li, Can Le, Elizabeth Ogburn, Edward McFowland III, Simon Goldberg, Matthew Hirshberg, Isabel Fulcher, Nina Varsava, Ralph Trane, Bennett Zhu, Ben Lystig, Dan Bolt, Katharine Scott, and Emily Case for helpful discussions and feedback. Several anonymous reviewers provided feedback that greatly improved the manuscript. Support for this research was provided by the University of Wisconsin--Madison, Office of the Vice Chancellor for Research and Graduate Education with funding, as well as NSF grants DMS 2052918, DMS 1646108 and DMS 2023239. Support for this research was also provided by American Family Insurance through a research partnership with the University of Wisconsin-Madison's American Family Insurance Data Science Institute.}

\newpage
\appendix

\section{Alternative Blockmodel Parameterization}
\label{app:interpretation}

We develop our theory of network regression around the latent positions $\X$. While this is mathematically convenient and allows us to relate our work to previous work on random dot product graphs, it is often difficult to develop good intuition for the latent positions $\X$, even in well-known settings such as the stochastic blockmodel, as discussed in Section \ref{subsec:causal-interpretation}. Instead of parameterizing a network in terms of latent positions $\X$, there is often a natural decomposition of $P = \Z B \Z^T$, which we explore here.

\begin{proposition}[Equivalent Parameterizations for Network Regression]
	\label{prop:alt-param}

	Suppose that $\Z \in \R^{n \times d}$ and $B \in \R^{d \times d}$ are arbitrary full-rank matrices such that $\E[\Z, B]{A} = \Z B \Z^T$ and let
	\begin{align}
		\label{eq:z-reg}
		\Z & = \W \Theta' + \xi',                  \\
		\label{eq:z-reg-2}
		Y  & = \W \betaw+ \Z \betaz + \varepsilon'
	\end{align}
	where $\xi' = \Z - \E[\W]{\Z} \in \R^{n \times d}$ and each row $\xi_{i \cdot}'$ is mean-zero and uncorrelated with the corresponding row $\W_{i \cdot}$, and the $\varepsilon_i$ are independent with bounded second moments.
	Then there exist $\Theta \in \R^{p \times d}, \xi \in \R^{n \times d}, \betax \in \R^{d}$ and $\varepsilon \in \R^n$ such that
	\begin{equation*} \begin{aligned}
			\X & = \W \Theta + \xi,  \qquad \qquad  \text{ and } \\
			Y  & = \W \betaw+ \X \betax + \varepsilon,
		\end{aligned} \end{equation*}
	where $\xi$ satisfies $\E[\W_{i \cdot}]{\xi_{ij}} = 0$ for $i \in [n], j \in [d]$, and the elements of $\varepsilon$ are independent with bounded second moments.
\end{proposition}

This proposition has several implications. First, if there is a linear regression \eqref{eq:z-reg} to establish how $\Z$ varies with nodal covariates $\W$, then there is another equivalent regression, also linear, to establish how $\X$ varies with nodal covariates $\W$. That is, every ``conveniently'' parameterized mediator model implies an ``inconveniently'' parameterized mediator model (our proofs are developed in the ``inconvenient'' setting). Further, provided that the errors in the $\Z$-regression are uncorrelated with $\W$, then the errors in the $\X$ regression are also uncorrelated with $\W$. Thus, if the coefficients in the convenient mediator regression are estimable, so are the coefficients in the inconvenient mediator regression. There is an analogous story for the outcome regression \eqref{eq:z-reg-2}.

\begin{remark}
	The reparameterization above comes at the cost of a potential loss of identifiability. For example, there are some law-rank models where $\X$ is identifiable up to multiplication by a signed permutation matrix \citep[e.g.,][Proposition 3.2]{rohe2022}, and much of the rotational ambiguity induced can be resolved via a varimax rotation. This in turn implies $\betax$ is identifiable up to sign flips in the regression coefficients $\betax$. Since our network regression model does not leverage any additional identifying information about latent positions $\X$, even when $\X$ is narrowly identified, $\betax$ is always subject to non-identifiability up to a full orthogonal rotation.

	We conjecture that plugging varimax rotated singular vectors into a nodal regression results in a consistent estimator of regression coefficients. We further anticipate that these estimates are normally distributed in the large-$n$ limit, effectively resolving the issue of rotational non-identifiability in network regression settings.
\end{remark}

\begin{proof}[Proof of Proposition~\ref{prop:alt-param}]
	First we consider the mediator model. We explicitly construct $\Theta'$ and $\xi'$ and then verify the dependence properties of $\xi'$. Let $\Sigma_Z = \Z^T \Z$ and let $R_U^T D R_U$ be the singular value decomposition of $\Sigma_Z^{1/2} B \Sigma_Z^{1/2}$. By Proposition 3.2 of \cite{rohe2022}, $D$ is a diagonal matrix containing the singular values of $\E[\Z, B]{A}$, and the singular vectors of $\E[\Z, B]{A}$ are given by
	\begin{equation*}
		\Upop = \Z \Sigma_Z^{-1/2} R_U^T.
	\end{equation*}
	Thus, $\Z$ and $B$ together imply that $\X = \Z \Sigma_Z^{-1/2} R_U^T D^{1/2}$. Then, by hypothesis,
	\begin{equation*}
		\X
		= \Z \Sigma_Z^{-1/2} R_U^T D^{1/2}
		= \paren*{W \Theta + \xi} \Sigma_Z^{-1/2} R_U^T D^{1/2}
		= W \Theta' + \xi'.
	\end{equation*}
	where $\Theta' = \Theta \, T_Z$ and $\xi' = \xi \, T_Z$ and $T_Z = \Sigma_Z^{-1/2} R_U^T D^{1/2}$ is an invertible matrix depending only on $\Z$ and $B$. The tower law then yields
	\begin{equation*}
		\E[\W_{i \cdot}]{\xi_{i \cdot}'}
		= \E[\W_{i \cdot}]{\xi_{i \cdot} \, T_Z}
		= \E[Z, B]{\E[\W_{i \cdot}]{\xi_{i \cdot}} \, T_Z}
		= \E[Z, B]{0 \cdot T_Z}
		= 0.
	\end{equation*}
	Now, turning to the outcome model, we have $\Z = \X D^{-1/2} R_U \Sigma_Z^{1/2}$. Let $\betax = D^{-1/2} R_U \Sigma_Z^{1/2} \betaz$. Then $\Z \betaz = \X \betax$ and we can take $\varepsilon = \varepsilon'$. Since $\X_{i \cdot}$ is a function of $\Z_{i \cdot}$ and $B$, and $B$ is fixed, independence of $\Z_{i \cdot}$ and $\varepsilon_i$ implies independence of $\X_{i \cdot}$ and $\varepsilon_i$, completing the proof.
\end{proof}

\section{Additional Details about the Glasgow Data Example}
\label{app:glasgow}

In the Glasgow example, recall that the adolescent social network was highly sexually homophilous (see Figure \ref{fig:glasgow}). This high level of homophily suggests that there may be positivity violations (Sections \ref{subsec:semi-param-model} and \ref{subsec:causal-interpretation}), and so we investigate positivity empirically by plotting $\Lhat$. In Figure \ref{fig:glasgow-positivity}, we see that the latent embeddings in the Glasgow data likely violate the positivity assumption, as some regions of the latent space are only occupied by male or female students. This implies that causal identification may not hold in the Glasgow data set.

Also recall that that the Glasgow social network is a directed network, and as a result each node has a set of left co-embeddings $\Xhat$ and a set of right co-embeddings $\Lhat$. In Section \ref{subsec:glasgow}, we used the right co-embeddings $\Lhat$, but we were unable to confirm that $\Lhat$ characterized the relevant social structure in the network. In Figure \ref{fig:glasgow-estimates-u}, we see that the results based on left co-embeddings $\Xhat$ are estimated to be smaller in magnitude than those based on the right co-embeddings $\Lhat$. Since we are unable to definitely adjudicate between $\Xhat$ and $\Lhat$, the difference in these analyses introduces ambiguity into the data analysis.

\begin{figure}
	\centering
	\includegraphics[width=\textwidth]{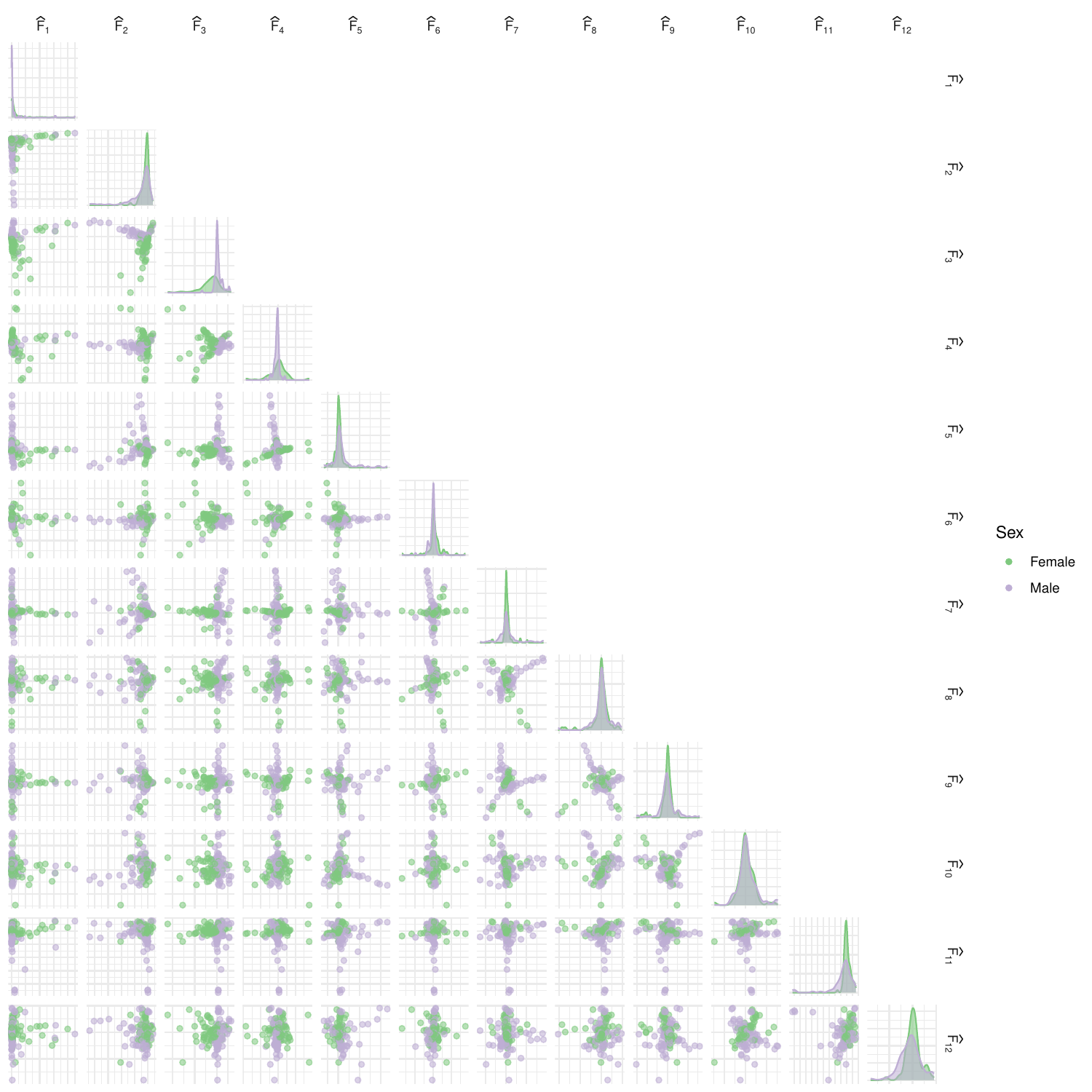}
	\caption{Right latent co-embeddings in the Glasgow data set, colored by reported sex.}
	\label{fig:glasgow-positivity}
\end{figure}

\begin{figure}
	\centering
	\includegraphics[width=0.8\textwidth]{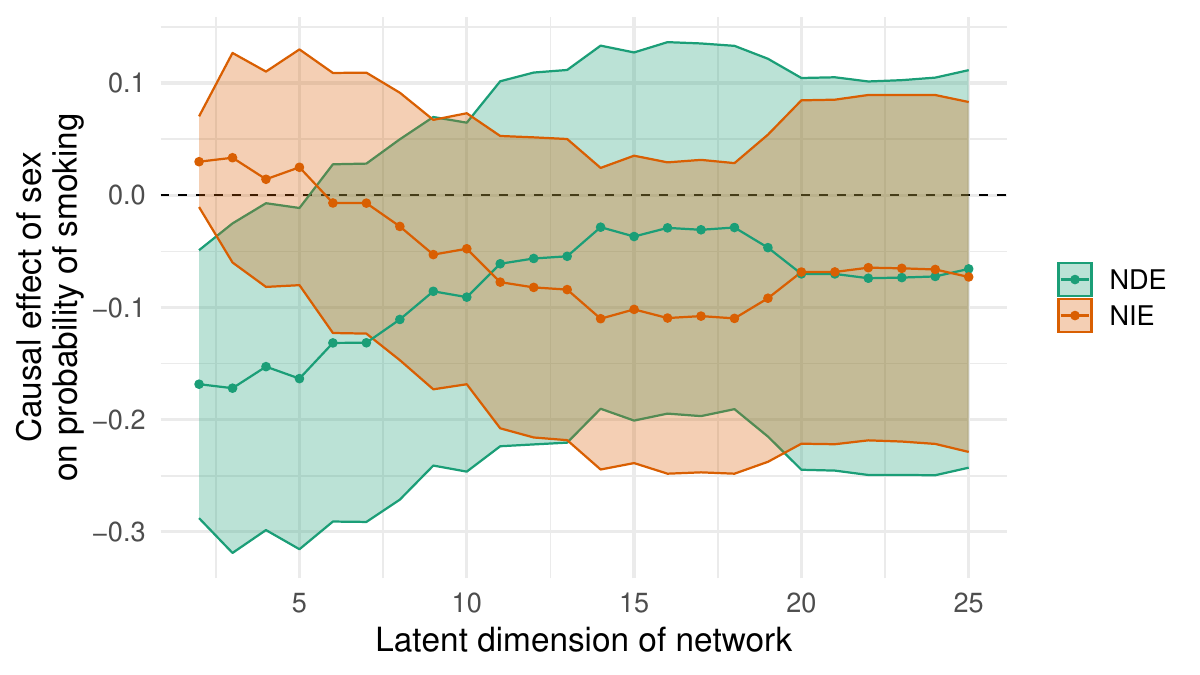}
	\caption{Estimated direct (teal) and indirect (orange) effects of sex on tobacco usage in the Glasgow social network, using the left co-embeddings $\Xhat$ rather than the right co-embeddings $\Lhat$ (compare with Figure \ref{fig:glasgow-estimates}). The estimated effects (vertical axis) vary with the dimension $d$ of the latent space (horizontal axis), and are adjusted for possible confounding by age and church attendance. Positive values indicate a greater propensity for adolescent boys to smoke, while negative values indicate a greater propensity for adolescent girls to smoke.}
	\label{fig:glasgow-estimates-u}
\end{figure}

\section{Additional Details about the Healthy Minds Data Example}
\label{app:healthyminds}

Figure~\ref{fig:bipartite-dag} characterizes the causal structure of mediation in a bipartite network, and Figure~\ref{fig:hm-curve} shows that $\ndehat$ and $\ndehat$ are insensitive to the embedding dimension in the Healthy Minds data application.

\begin{figure}
	\centering
	\includegraphics[width=0.45\linewidth]{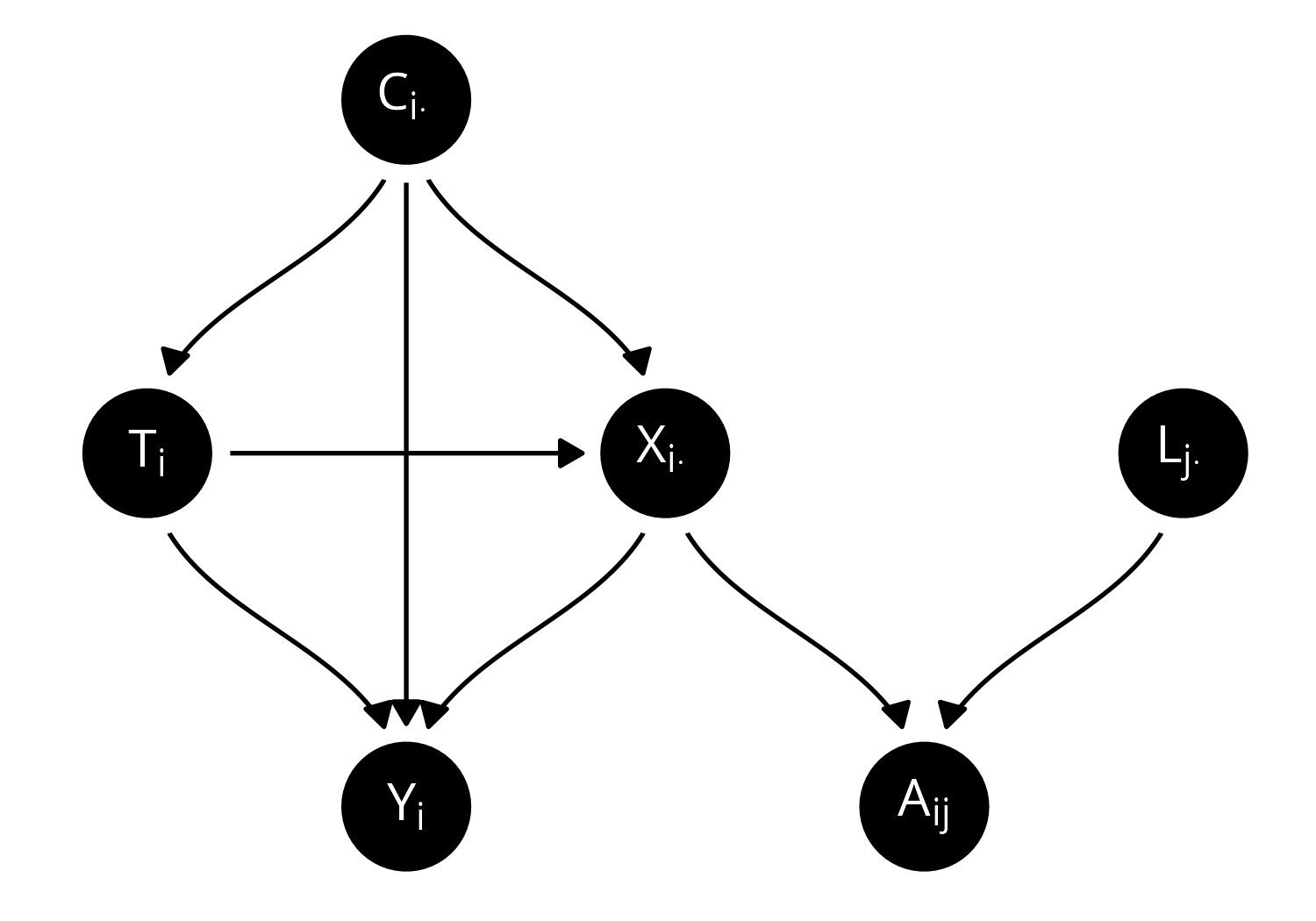}
	\caption{A directed acyclic graph (DAG) representing the causal pathways of latent mediation in a bipartite network. The network has two nodes called $i$ and $j$ and the node $i$ is the unit of interest. Each node in the figure corresponds to a random variable, and edges indicate which random variables may cause which other random variables. We are interested in the causal effect of $T_i$ on $Y_i$, as mediated by the latent position $\X_{i \cdot}$.}
	\label{fig:bipartite-dag}
\end{figure}

\begin{figure}
	\centering
	\includegraphics[width=0.7\textwidth]{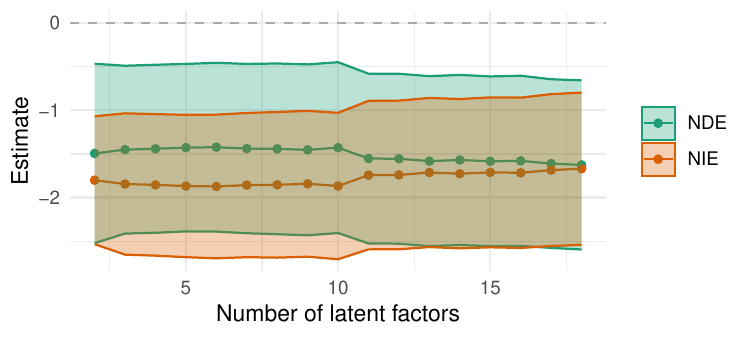}
	\caption{Estimated direct (teal) and indirect (orange) effects of the Healthy Minds program on anxiety levels at the end of the four week intervention period. The estimated effects (vertical axis) vary with the dimension $d$ of the latent space (horizontal axis), and are adjusted for possible confounding by age and sex attendance. Positive values indicate that the intervention increased anxiety, while negative values indicate that intervention decreases anxiety.}
	\label{fig:hm-curve}
\end{figure}

\FloatBarrier
\subsection{Five Facet Mindfulness Questionnaire Act With Awareness subscale}

\begin{enumerate}

	\setlength{\itemsep}{0pt}
	\setlength{\parskip}{0pt}
	\item When I do things, my mind wanders off and I'm easily distracted.
	\item I don't pay attention to what I'm doing because I'm daydreaming, worrying, or otherwise distracted.
	\item I am easily distracted.
	\item I find it difficult to stay focused on what's happening in the present.
	\item It seems I am `running on automatic' without much awareness of what I'm doing.
	\item I rush through activities without being really attentive to them.
	\item I do jobs or tasks automatically without being aware of what I'm doing.
	\item I find myself doing things without paying attention.
\end{enumerate}
\begin{center}
	\begin{tabular}{|c|c|}
		\hline
		\textbf{Score} & \textbf{Description}      \\
		\hline
		5              & Very often or always true \\
		\hline
		4              & Often true                \\
		\hline
		3              & Sometimes true            \\
		\hline
		2              & Rarely true               \\
		\hline
		1              & Never or very rarely true \\
		\hline
	\end{tabular}
\end{center}

\subsection{Drexel Defusion Scale}

\begin{enumerate}
	\setlength{\itemsep}{0pt}
	\setlength{\parskip}{0pt}
	\item Feelings of anger. You become angry when someone takes your place in a long line. To what extent would you normally be able to defuse from feelings of anger?
	\item Cravings for food. You see your favorite food and have the urge to eat it. To what extent would you normally be able to defuse from cravings for food?
	\item Physical pain. Imagine that you bang your knee on a table leg. To what extent would you normally be able to defuse from physical pain?
	\item Anxious thoughts. Things have not been going well at school or your job, and work just keeps piling up. To what extent would you normally be able to defuse from anxious thoughts like ``I'll never get this done.''?
	\item Thoughts of self. Imagine you are having a thought such as ``no one likes me.'' To what extent would you normally be able to defuse from negative thoughts about yourself?
	\item Thoughts of hopelessness. You are feeling sad and stuck in a difficult situation that has no obvious end in sight. You experience thoughts such as ``Things will never get any better.'' To what extent would you normally be able to defuse from thoughts of hopelessness?
	\item Thoughts about motivation or ability. Imagine you are having a thought such as ``I can't do this'' or ``I just can't get started.'' To what extent would you normally be able to defuse from thoughts about motivation or ability?
	\item Thoughts about your future. Imagine you are having thoughts like, ``I'll never make it'' or ``I have no future.'' To what extent would you normally be able to defuse from thoughts about your future?
	\item Sensations of fear. You are about to give a presentation to a large group. As you sit waiting for your turn, you start to notice your heart racing, butterflies in your stomach, and your hands trembling. To what extent would you normally be able to defuse from sensations of fear?
	\item Feelings of sadness. Imagine that you lose out on something you really wanted. You have feelings of sadness. To what extent would you normally be able to defuse from feelings of sadness?
\end{enumerate}
\begin{center}
	\begin{tabular}{|c|c|}
		\hline
		\textbf{Score} & \textbf{Description} \\
		\hline
		5              & Very much            \\
		\hline
		4              & Quite a lot          \\
		\hline
		3              & Moderately           \\
		\hline
		2              & Somewhat             \\
		\hline
		1              & A little             \\
		\hline
		0              & Not at all           \\
		\hline
	\end{tabular}
\end{center}

\subsection{NIH Toolbox Loneliness Questionnaire}

\begin{enumerate}
	\setlength{\itemsep}{0pt}
	\setlength{\parskip}{0pt}
	\item I feel alone and apart from others
	\item I feel left out
	\item I feel that I am no longer close to anyone
	\item I feel alone
	\item I feel lonely
\end{enumerate}
\begin{center}
	\begin{tabular}{|c|c|}
		\hline
		\textbf{Score} & \textbf{Description} \\
		\hline
		5              & Always               \\
		\hline
		4              & Usually              \\
		\hline
		3              & Sometimes            \\
		\hline
		2              & Rarely               \\
		\hline
		1              & Never                \\
		\hline
	\end{tabular}
\end{center}

\section{Case Study: Over-Control Bias}
\label{app:overcontrol-bias}

Treating latent positions as confounders when they are in fact mediators leads to biased causal estimates. This bias can be substantial in empirical networks.

When the latent positions are confounders (see the structural causal model in Figure \ref{fig:confounding}), $\betat$ is the average treatment effect:
\begin{equation*}
	\begin{aligned}                                                    \\
		\E[T_i, C_{i \cdot}, \X_{i \cdot}]{Y_i}
		 & = \betazero + T_i \underbrace{\betat}_{\substack{\text{average} \\ \text{treatment} \\ \text{effect}}} + ~C_{i \cdot} \betac + \X_{i \cdot} \betax.
	\end{aligned}
\end{equation*}
In contrast, when the latent positions are mediators (see the structural causal model in Figure \ref{fig:mediating}), $\betat$ is the natural direct effect:
\begin{equation*}
	\begin{aligned}                                                    \\
		\E[T_i, C_{i \cdot}, \X_{i \cdot}]{Y_i}
		 & = \betazero + T_i \underbrace{\betat}_{\substack{\text{natural}     \\ \text{direct} \\ \text{effect}}} + C_{i \cdot} \betac + \X_{i \cdot} \underbrace{\betax}_{\substack{\text{effect of} \\ \text{$X$ on $Y$}}}, ~\text{ and }~  \\
		\E[T_i, C_{i \cdot}]{\X_{i \cdot}}
		 & = \thetazero + T_i \underbrace{\thetat}_{\substack{\text{effect of} \\ \text{$T$ on $X$}}}  + C_{i \cdot} \Thetac.
	\end{aligned}
\end{equation*}
Treating the latent positions as confounders when they are mediators implies the mistaken identification result $\betat = \ate$. However, in truth, $\betat = \nde = \ate - \nie$, and using $\betat$ as an estimate of $\ate$ induces a bias of $\nie$ into the estimate of the average treatment effect. This bias is well-known as ``over-control bias'' \citep{cinelli2022}.

\begin{figure}
	\centering
	\includegraphics[width=0.6\textwidth]{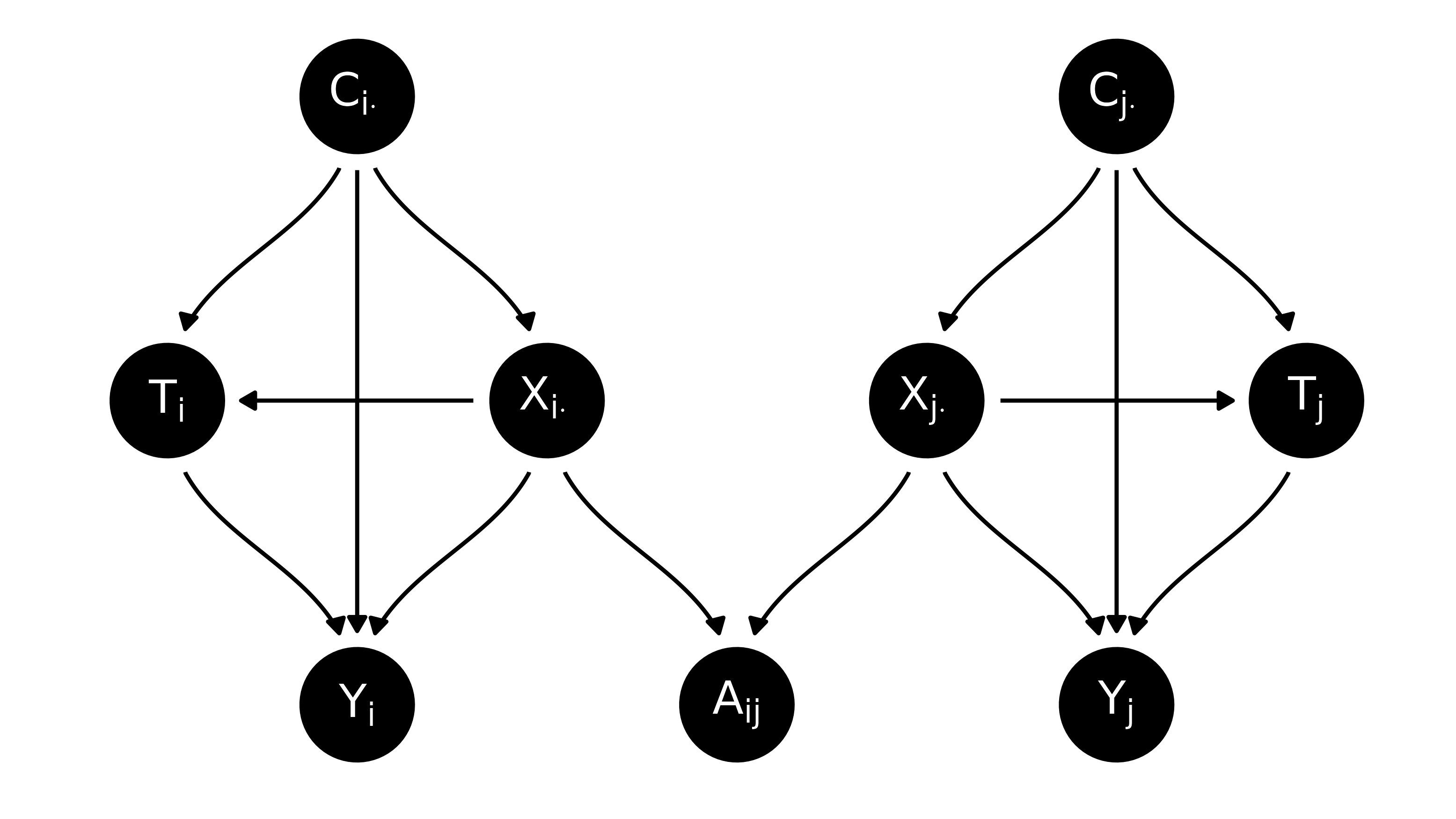}
	\caption{A DAG representing the causal pathways in a network with homophilous \emph{confounding}, for a network with two nodes called $i$ and $j$. Compare to Figure~\ref{fig:mediating}, where the direction of the $\X_{i \cdot} \to T_i$ and $\X_{j \cdot} \to T_j$ arrows are reversed.}
	\label{fig:confounding}
\end{figure}

Over-control bias can be large in network data. To demonstrate this, we re-analyzed the AddHealth data set investigated in the initial pre-print of \cite{le2022}. The AddHealth data consists of a self-reported social network of 2,152 high school students, along with grade level, sex, race, and a proxy measure of mental health for each student. Our goal is to investigate how mental health varies with race, controlling for grade level and sex.

In the original analysis, \cite{le2022} used a procedure that is equivalent to the ordinary least squares regression
\begin{equation*}
	\texttt{mental\_health} \sim \texttt{grade + race + sex + Xhat}
\end{equation*}
and found that race did not have a statistically significant effect on mental health. The original analysis did not interpret the regression coefficients causally, but it did suggest that the effect of race was plausibly zero.

Our mediation framework allows us to clarify the role of race. Race (as well as grade and sex) causes community structure, rather than the other way around, such that latent social groups are clearly mediators in the AddHealth network. Indeed, race precedes the network in time, and it is impossible for causal arrows to point backwards in time.

Using the framework developed in this paper, we fit models
\begin{equation*}
	\begin{aligned}
		\texttt{mental\_health} & \sim \texttt{grade + race + sex + Xhat} \\
		\texttt{Xhat}           & \sim \texttt{grade + race + sex}
	\end{aligned}
\end{equation*}
and then computed $\ndehat$ and $\niehat$, which we plotted with confidence intervals in Figure \ref{fig:addhealth}. Using the cross-validated eigenvalue method proposed in \cite{chen2021}, we determined that a reasonable choice of latent dimension was $d \approx 120$. Our estimates were stable in a neighborhood around this value of $d$, suggesting that they were reliable so long $d$ was not badly misspecified (see Remark \ref{rem:model-selection} and Section \ref{sec:simulations}).

\begin{figure}
	\centering
	\includegraphics[width=0.8\textwidth]{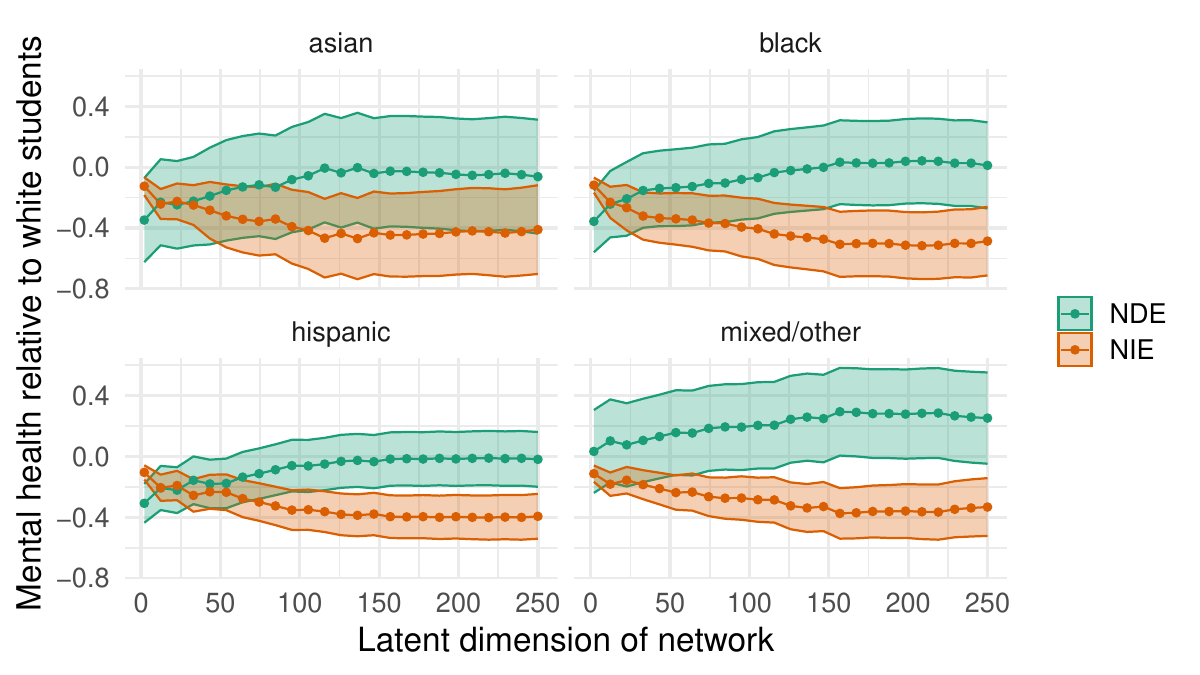}
	\caption{Confidence intervals (95\%) for natural direct (teal) and indirect (orange) effects in the AddHealth data example. The vertical axis corresponds to the value of the causal effect, and the horizontal axis encodes the embedding dimension $d$. All contrasts are relative to white students.}
	\label{fig:addhealth}
\end{figure}

At $d = 120$, the estimated direct effect of non-white race was zero, and the estimated indirect effects of non-white race was large and statistically significant. This in turn implied that there was a large average treatment effect of race on mental health. The effect simply operates along the indirect pathway, as race causes group membership.

The takeaways from this case study are two-fold. First, when latent positions are included in network regression models, they are likely to be interpreted, either implicitly or explicitly, as causal confounders. It's important to assess whether or not this is the case. Second, it's entirely plausible for causal effects in social networks to occur entirely via the indirect pathway. In this setting, mistakenly using the latent positions as confounders rather than mediators can result in over-control bias and very misleading estimates of causal effects.

\section{Proof of Theorem \ref{thm:main}}

In order to prove Theorem \ref{thm:main}, we will first prove a more general statement holds under the sub-gamma network model (Assumption \ref{ass:subgamma}), which generalizes the random dot product graph considered in Example~\ref{ex:rdpg}. To work with the general sub-gamma model, we require the following assumptions about the sub-gamma parameters:

\begin{assumption}[Growth rates] \label{ass:growth-rates}
	Under the model in Assumption~\ref{ass:subgamma}, the eigenvalues $\lambda_1$ and $\lambda_d$ and the sub-gamma parameters $\nu_n$ and $b_n$ grow with $n$ in such a way that
	\begin{equation} \label{eq:growth:lambda1:UB}
		\lambda_1 = \bigoh{ n },
	\end{equation}
	\begin{equation} \label{eq:growth:lambdad:LB}
		\lambda_1 = \Omega( 1 ),
	\end{equation}
	\begin{equation} \label{eq:growth:46MT1}
		\frac{ \sqrt{ \nu_n+b_n^2} n \log^{3/2} n }{ \lambda_d^{3/2} } = o(1),
	\end{equation}
	\begin{equation} \label{eq:growth:46MT5a}
		\frac{ \lambda_1 (\nu_n+b_n^2) n^2 \log^2 n }{ \lambda_d^{7/2} } = o(1),
	\end{equation}
	and
	\begin{equation} \label{eq:lambdaratio}
		\frac{ \lambda_d }{ \lambda_1^{1/2} } = \bigoh{ 1 }.
	\end{equation}

	We additionally take $d$, the rank of the latent positions $\X$, and $p$, and the number of nodal controls, to be fixed asymptotically.
\end{assumption}

Note that Assumption \ref{ass:growth-rates} holds for the random dot product graph (Example \ref{ex:rdpg}).

\begin{proposition}
	\label{prop:rdpg-subgamma-params}

	Suppose that $(A, \X) \sim \RDPG(F, n)$, as described in Example \ref{ex:rdpg}. Then $(A, \X)$ are generated according to a process that satisfies Assumption \ref{ass:subgamma}, and further $\lambda_1 = \Theta \paren*{n}, \lambda_d = \Theta \paren*{n}, \nu_n = c$, and $b_n = 1$ for some $c > 0$.
\end{proposition}

\begin{proof}
	See \citet[Remark 24]{athreya2018} or \citet[Prop 4.3]{sussman2014}.
\end{proof}

A straightforward application of Proposition \ref{prop:rdpg-subgamma-params} shows that random dot product graphs satisfy Assumption \ref{ass:growth-rates}. Thus, in order to prove Theorem \ref{thm:main}, it is sufficient to replace Example \ref{ex:rdpg} with Assumptions \ref{ass:subgamma} and \ref{ass:growth-rates}, as in the following. Recalling the $\W$ notation of Equation \eqref{eq:def:W}, we have the following theorem.

\begin{theorem}
	\label{thm:main-subgamma}

	If Assumptions~\ref{ass:subgamma},~\ref{ass:causal-linearity},~\ref{ass:regularity}, \ref{ass:second-stage}, and \ref{ass:growth-rates} hold, then there exists a sequence of orthogonal matrices $\set{Q_n}_{n=1}^\infty$ such that
	\begin{equation*}
		\begin{aligned}
			\sqrt{ n } \,
			\Sigmahattheta^{-1/2}
			\begin{pmatrix}
				\vecc \paren*{\Thetahat \, Q_n^T} - \Thetavec
			\end{pmatrix}
			 & \to
			\Normal{0}{I_{p d}}, \text{ and } \\
			\sqrt{ n } \,
			\Sigmahatbeta^{-1/2}
			\begin{pmatrix}
				\betawhat - \betaw \\
				Q_n \, \betaxhat - \betax
			\end{pmatrix}
			 & \to
			\Normal{0}{I_d}.
		\end{aligned}
	\end{equation*}
\end{theorem}

The proof of Theorem \ref{thm:main-subgamma} reduces to three key lemmas. Lemma \ref{lem:true-x-point-ests} shows that least squares estimates are asymptotically normally distributed when the true latent positions are known. This is a standard result of M-estimation theory for regression.

\begin{lemma}[\cite{boos2013}, Theorems 7.2]
	\label{lem:true-x-point-ests}

	Define $\betawtilde, \betaxtilde$ and $\Thetatilde$ analogously to the estimators in Definition \ref{def:reg-point-estimators}, but using the true latent positions $\X$ rather than the spectral embedding $\Xhat$. Under Assumptions~\ref{ass:subgamma},~\ref{ass:causal-linearity}, and~\ref{ass:regularity},

	\begin{equation*} \begin{aligned}
			\sqrt{n} \begin{pmatrix} \Thetavectilde - \Thetavec \end{pmatrix}
			 & \to
			\Normal{0}{\Sigmathetavec}, \text{ and } \\
			\sqrt{ n }
			\begin{pmatrix}
				\betawtilde - \betaw \\
				\betaxtilde - \betax
			\end{pmatrix}
			 & \to
			\Normal{0}{\Sigma_\beta}.
		\end{aligned} \end{equation*}
\end{lemma}

Further, when the true latent positions are known, the covariance of the least squares coefficients is also estimable using the robust covariance estimator.

\begin{lemma}[\cite{boos2013}, Theorems 7.3, 7.4]
	\label{lem:true-x-covariance}

	Define $\Sigmatildebeta$ and $\Sigmatildetheta$ analogously to the estimators in Definition \ref{def:covariance}, but using the true latent positions $\X$ rather than the spectral embedding $\Xhat$. Under Assumptions \ref{ass:subgamma}, \ref{ass:causal-linearity}, and \ref{ass:regularity},
	\begin{align*}
		\Sigmatildebeta
		 & \to \Sigmabeta  ~~~\text{ in probability, and} \\
		\Sigmatildetheta
		 & \to \Sigmathetavec ~~~\text{ in probability. }
	\end{align*}
\end{lemma}

Since the true latent positions $\X$ are unobserved, we must use estimates $\Xhat$ in place of $\X$ in least squares estimators. This does not change the asymptotic distribution of the least squares coefficients.

\begin{lemma} \label{lem:op1}
	Under Assumptions~\ref{ass:subgamma},~\ref{ass:causal-linearity},~\ref{ass:regularity},~\ref{ass:second-stage}, and~\ref{ass:growth-rates},
	\begin{equation*}
		\sqrt{ n }
		\begin{pmatrix}
			\betawhat - \betawtilde \\
			Q_n \, \betaxhat - \betaxtilde
		\end{pmatrix}
		= \op{1}
		~~~\text{ and }~~~
		\sqrt{ n }
		\begin{pmatrix}
			\Thetahat \, Q_n^T - \Thetatilde
		\end{pmatrix}
		= \op{1}.
	\end{equation*}
\end{lemma}

We must show a similar result for the covariance estimators in order to obtain a consistent estimator based on $\Xhat$ rather than $\X$.

\begin{lemma}
	\label{lem:covariance}
	Under Assumptions \ref{ass:subgamma}, \ref{ass:regularity}, \ref{ass:second-stage}, and \ref{ass:growth-rates}, we have
	\begin{equation*}
		\Sigmahatbeta \to \Sigmatildebeta \qquad \text{and} \qquad
		\Sigmahattheta \to \Sigmatildetheta
		~~~\text{ both in probability.}
	\end{equation*}
\end{lemma}

Proofs of Lemmas~\ref{lem:op1} and~\ref{lem:covariance} are provided in Sections~\ref{sec:proof:lemop1} and~\ref{sec:proof:lemcovar}, respectively.
Additional technical results are also given in Section~\ref{sec:proofs:lemop1-lemcovar}.

With these preliminaries in place, we are ready to prove Theorem \ref{thm:main-subgamma}. Note that in this proof, and subsequently, we often drop the subscript $n$ from $Q_n$ for convenience and write $Q$ instead.

\begin{proof}[Proof of Theorem \ref{thm:main-subgamma}]

	First, we show that $\betahat$ and $\Thetahat$ are asymptotically normal. By Lemmas~\ref{lem:op1} and~\ref{lem:true-x-point-ests}, asymptotically we have
	\begin{align*}
		\sqrt n \paren*{\Thetahat \, \Q^T - \Theta} & =
		\underbrace{\sqrt n \paren*{\Thetahat \, \Q^T - \Thetatilde}}_{\op{1}} +
		\underbrace{\sqrt n \paren*{\Thetatilde - \Theta}}_{\Normal{0}{\SigmaTheta}}.
	\end{align*}
	Slutsky's theorem is then sufficient to establish asymptotic normality. An analogous argument holds for $\betahat$.

	We use a similar argument to establish consistency of covariance estimation. By Lemmas~\ref{lem:covariance} and~\ref{lem:true-x-covariance}
	\begin{align*}
		\Sigmahattheta - \Sigmathetavec =
		\underbrace{\Sigmahattheta - \Sigmatildetheta}_{\op{1}} +
		\underbrace{\Sigmatildetheta - \Sigmathetavec}_{\op{1}},
	\end{align*}
	\noindent such that $\Sigmahattheta \to \Sigmathetavec$ in probability. Again, an analogous argument holds for the covariance of $\betahat$. A final application of Slutsky's theorem to combine the above two results completes the proof.
\end{proof}

\subsection{Proof of Lemma \ref{lem:op1}} \label{sec:proof:lemop1}

To prove Lemma \ref{lem:op1}, we will first consider the mediator regression coefficients, and then the outcome regression coefficients. We partition the outcome regression coefficients $\beta = (\betaw, \betax)$ using the Frisch-Waugh-Lowell theorem to deal with identified and unidentified coefficients separately. Here we present some important supporting lemmas that outline the proof; some tedious and less illuminating supporting lemmas are relegated to a later portion of the Appendix.

\begin{lemma}[Sub-gamma mediator coefficient bound]
	\label{lem:Theta-op1}

	Suppose Assumptions \ref{ass:subgamma}, \ref{ass:regularity}, \ref{ass:second-stage} and~\ref{ass:growth-rates} hold. Let $\set{Q_n}_{n=1}^\infty$ be the sequence of orthogonal matrices guaranteed by Lemma~\ref{lem:subgamma-2toinfty}.
	Then
	\begin{equation*}
		\norm*{\sqrt n \paren*{\Thetahat \, Q_n^T- \Thetatilde}}
		= \op{1}.
	\end{equation*}
\end{lemma}

\begin{proof}
	Using basic properties of the spectral norm,
	\begin{equation} \label{eq:thetaprod} \begin{aligned}
			\norm*{\sqrt n \paren*{\Thetahat \, \Q^T- \Thetatilde}}
			 & = \sqrt{n} \, \norm*{\paren*{\W^T \W}^{-1} \W^T \paren*{\Xhat - \X \Q}} \\
			 & \le \sqrt{n} \, \frac{1}{n}
			\norm*{\paren*{\frac{1}{n} \W^T \W}^{-1}}
			\norm*{\W^T \paren*{\Xhat - \X \Q}}.
		\end{aligned} \end{equation}
	By the independence and moment conditions of Assumption \ref{ass:regularity}, $\paren*{\frac{1}{n} \W^T \W}^{-1}$ converges to the inverse covariance matrix of $\W$, so that
	\begin{equation*}
		\norm*{\paren*{\frac{1}{n} \W^T \W}^{-1}} = \bigoh{ 1 }.
	\end{equation*}
	By Lemma~\ref{lem:XhatXW},
	\begin{equation*}
		\norm*{\W^T \paren*{\Xhat - \X \Q}} = \op{ \sqrt{n} }.
	\end{equation*}
	Applying the above two displays to Equation~\eqref{eq:thetaprod},
	\begin{equation*}
		\norm*{\sqrt n \paren*{\Thetahat \, \Q^T- \Thetatilde}}
		= \op{1},
	\end{equation*}
	completing the proof.
\end{proof}

\begin{theorem}[\citealt{lovell1963, frisch1933}]
	\label{thm:FWL}
	Let $\betahat$ be as in Definition \ref{def:reg-point-estimators}. Then
	\begin{align}
		\begin{bmatrix}
			\betawhat \\
			\betaxhat
		\end{bmatrix}
		=
		\begin{bmatrix}
			\paren*{W^T W}^{-1} W^T \paren*{Y - \Xhat \, \betaxhat} \\
			\paren*{\Xhat^T M \Xhat}^{-1} \Xhat^T M Y
		\end{bmatrix}
	\end{align}
	\noindent where $M = I - W^T \paren*{W^T W}^{-1} W$.
\end{theorem}

We note that $M$ projects vectors onto the orthogonal complement of the column space of $W$. Note that $\norm*{M} = 1$ since $M$ is a projection matrix.

\begin{lemma}[Sub-gamma outcome coefficient bound] \label{lem:betax-op1}
	Suppose Assumptions~\ref{ass:subgamma},~\ref{ass:regularity},~\ref{ass:second-stage} and~\ref{ass:growth-rates} hold, and let $\set{Q_n}_{n=1}^\infty$ be the sequence of orthogonal matrices guaranteed by Lemma~\ref{lem:subgamma-2toinfty}.
	Then
	\begin{equation*}
		\sqrt{ n } \paren*{Q_n \, \betaxhat - \betaxtilde} = \op{1}.
	\end{equation*}
\end{lemma}
\begin{proof}
	Applying the definition of $\betaxhat$ and $\betaxtilde$ from Theorem~\ref{thm:FWL} and adding and subtracting appropriate quantities, we have
	\begin{align}
		\sqrt n \paren*{\Q \, \betaxhat - \betaxtilde}
		 & = \sqrt n \brac*{
			\Q \, \paren*{\Xhat^T M \Xhat}^{-1} \Xhat^T
		- \paren*{\X^T M \X}^{-1} \X^T } M Y \notag               \\
		 & = \sqrt{n} \brac*{
			\Q \, \paren*{\Xhat^T M \Xhat}^{-1}
		- \paren*{\X^T M \X}^{-1} \, \Q } \Xhat^T M Y \label{xb1} \\
		 & \quad \quad
		+ \sqrt n \paren*{\X^T M \X}^{-1} \paren*{\Q \Xhat^T - \X^T} M Y.
		\label{xb2}
	\end{align}

	We bound the quantities~\eqref{xb1} and~\eqref{xb2} separately, starting with \eqref{xb1}.
	Recalling the definition of $M$ in Theorem~\ref{thm:FWL} and expanding $Y$,
	\begin{equation*}
		\norm*{ \Xhat^T MY}
		= \norm*{\Xhat^T\left( M \X \beta + M W \xi + M \varepsilon \right) }
		= \norm*{\Xhat^T\left( M \X \beta + M \varepsilon \right) }.
	\end{equation*}
	Applying the triangle inequality, submultiplicativity and the fact that $M$ is a projection matrix,
	\begin{equation*}
		\norm*{ \Xhat^T MY}
		\le \norm*{ \Xhat } \norm*{ \beta } + \norm*{ \Xhat^T M \varepsilon }.
	\end{equation*}
	Lemma~\ref{lem:spectralnorms} and the fact that $\beta$ is a constant control the first term, while Lemma~\ref{lem:norm-of-errors:Heps} with $H = \Xhat^T M$ controls the second term, and we have
	\begin{equation} \label{eq:XhatMY:bound:prelim}
		\norm*{ \Xhat^T MY}
		\le  C \lambda_1^{1/2} + \sqrt{ B \trace \Xhat^T M M^T \Xhat }.
	\end{equation}
	By cyclicity of the trace and Von Neumann's trace inequality,
	\begin{equation*}
		\trace \Xhat^T M M^T \Xhat
		= \trace \Xhat \Xhat^T M M^T
		\le \trace \Xhat \Xhat^T = \trace \Shat,
	\end{equation*}
	where we have used the fact that $M$ is a projection matrix and the definition of $\Xhat = \Uhat \Shat^{1/2}$.
	Bounding $\trace \Shat \le d \| \Shat \|$, Lemma~\ref{lem:spectralnorms} implies
	\begin{equation*}
		\trace \Xhat^T M M^T \Xhat \le C d \lambda_1.
	\end{equation*}
	Plugging this bound back into Equation~\eqref{eq:XhatMY:bound:prelim} and using the fact that $B$ and $d$ are assumed constant, we conclude that
	\begin{equation} \label{eq:XhatMY:bound}
		\norm*{ \Xhat^T MY} = \Op{ \lambda_1^{1/2} }.
	\end{equation}

	Noting that for conformable matrices $A$ and $B$, $A^{-1} - B^{-1} = A^{-1} (B - A) B^{-1}$, submultiplicativity of the spectral norm implies
	\begin{equation*} \begin{aligned}
			 & \norm*{\Q \, \paren*{\Xhat^T M \Xhat}^{-1}
			- \paren*{\X^T M \X}^{-1} \, \Q}              \\
			 & \quad \quad = \norm*{
				\Q \, \paren*{\Xhat^T M \Xhat}^{-1}
				\brac*{\Q^T \X^T M \X - \Xhat^T M \Xhat \, \Q^T}
			\paren*{\X^T M \X}^{-1} \Q }                  \\
			 & \quad \quad
			\le \norm*{\paren*{\Xhat^T M \Xhat}^{-1}}
			\norm*{\Q^T \X^T M \X - \Xhat^T M \Xhat \, \Q^T}
			\norm*{\paren*{\X^T M \X}^{-1}}.
		\end{aligned} \end{equation*}
	Applying Lemmas~\ref{lem:XMXXhatMXhat} and~\ref{lem:xmx-spectral-bound}, it follows that
	\begin{equation} \label{eq:XMXinv}
		\norm*{\Q \, \paren*{\Xhat^T M \Xhat}^{-1} - \paren*{\X^T M \X}^{-1} \, \Q}
		= \op{ \frac{ 1 }{ \sqrt{n} \lambda_1^{1/2} } }.
	\end{equation}
	Applying submultiplicativity, we can bound \eqref{xb1} as
	\begin{equation*} \begin{aligned}
			 & \norm*{
				\sqrt n \brac*{\Q \,
					\paren*{\Xhat^T M \Xhat}^{-1} - \paren*{\X^T M \X}^{-1} \, \Q}
			\Xhat^T M Y }               \\
			 & \qquad \le \sqrt{ n } \,
			\norm*{\Q \, \paren*{\Xhat^T M \Xhat}^{-1} - \paren*{\X^T M \X}^{-1} \, \Q}
			\norm*{\Xhat^T M Y}.
		\end{aligned} \end{equation*}
	Applying Equations~\eqref{eq:XhatMY:bound} and~\eqref{eq:XMXinv}, the quantity in Equation~\eqref{xb1} is bounded as
	\begin{equation} \label{eq:term1:op1}
		\norm*{ \sqrt{ n } \brac*{\Q \, \paren*{\Xhat^T M \Xhat}^{-1}
				- \paren*{\X^T M \X}^{-1} \, \Q} \Xhat^T M Y }
		= \op{1}.
	\end{equation}

	Turning our attention to the quantity on line \eqref{xb2}, by applying submultiplicativity and the triangle inequality, along with the fact that $M$ is a projection matrix,
	\begin{equation*} \begin{aligned}
			 & \norm*{\sqrt n \paren*{\X^T M \X}^{-1}
			\paren*{\Q \, \Xhat^T - \X^T} M Y}                                                                                           \\
			 & ~~~~~~~~~\le \sqrt n \, \norm*{\paren*{\X^T M \X}^{-1}}
			\norm*{\paren*{\Q \, \Xhat^T - \X^T} M Y}                                                                                    \\
			 & ~~~~~~~~~ = \sqrt n \, \norm*{\paren*{\X^T M \X}^{-1}}
			\norm*{\paren*{\Q \, \Xhat^T - \X^T} M \X \beta
			\paren*{\Q \, \Xhat^T - \X^T} M \varepsilon}                                                                                 \\
			 & ~~~~~~~~~ \le \sqrt{ n } \, \norm*{\paren*{\X^T M \X}^{-1}}
			\norm*{\paren*{\Q \, \Xhat^T - \X^T} M \X} \norm*{\beta}                                                                     \\
			 & ~~~~~~~~~ \quad \quad + \sqrt n \, \norm*{\paren*{\X^T M \X}^{-1}} \norm*{\paren*{\Q \, \Xhat^T - \X^T} M \varepsilon}  .
		\end{aligned} \end{equation*}
	Applying Lemmas~\ref{lem:XhatXMX},~\ref{lem:XhatXMeps} and~\ref{lem:xmx-spectral-bound} and using the fact that $\norm*{\beta}$ is a constant,
	\begin{equation*}
		\norm*{\sqrt n \paren*{\X^T M \X}^{-1} \paren*{\Q \, \Xhat^T - \X^T} M Y}
		= \op{1}.
	\end{equation*}

	Using the above display and Equation~\eqref{eq:term1:op1}, respectively, to bound the terms on lines~\eqref{xb1} and~\eqref{xb2}, we conclude that
	\begin{equation*}
		\sqrt n \paren*{\Q \, \betaxhat - \betaxtilde}
		= \op{1},
	\end{equation*}
	completing the proof.
\end{proof}

\begin{lemma} \label{lem:betaw-op1}
	Suppose that Assumptions~\ref{ass:subgamma},~\ref{ass:causal-linearity},~\ref{ass:regularity},~\ref{ass:second-stage} and~\ref{ass:growth-rates} hold.
	Then, letting $\betawtilde$ and $\betaxtilde$ be as specified in Definition \ref{def:reg-point-estimators},
	\begin{equation*}
		\sqrt{ n } \paren*{\betawhat - \betawtilde} = \op{1}.
	\end{equation*}
\end{lemma}
\begin{proof}
	Applying Theorem~\ref{thm:FWL} and using basic properties of norms, we have
	\begin{equation}
		\label{eq:betaw:step1}
		\begin{aligned}
			\norm*{
				\sqrt n \paren*{
					\betawhat - \betawtilde
				}
			}
			 & \le \sqrt{ n } \, \norm*{\paren*{W^T W}^{-1} W^T
			\paren*{\X - \Xhat \Q^T}} \norm*{\betaxtilde}               \\
			 & \quad + \sqrt n \, \norm*{\paren*{W^T W}^{-1} W^T \Xhat}
			\norm*{\Q^T \betaxtilde - \betaxhat}.
		\end{aligned} \end{equation}

	By Lemmas~\ref{lem:true-x-point-ests} and~\ref{lem:Theta-op1},
	\begin{equation*}
		\norm*{\betaxtilde}
		= \Op{1} ~~~ \text{ and } ~~~
		\sqrt{ n } \, \norm*{\paren*{W^T W}^{-1} W^T \paren*{\X - \Xhat \, \Q^T }}
		= \op{1},
	\end{equation*}
	and it follows that
	\begin{equation} \label{eq:betaw:rate1}
		\sqrt{ n } \, \norm*{\paren*{W^T W}^{-1} W^T \paren*{\X - \Xhat \Q^T}} \norm*{\betaxtilde} = \op{1}.
	\end{equation}
	By Lemma~\ref{lem:betax-op1},
	\begin{equation*}
		\sqrt n \, \norm*{\Q^T \, \betaxtilde - \betaxhat} = \op{1},
	\end{equation*}
	and by Lemmas~\ref{lem:true-x-point-ests} and~\ref{lem:Theta-op1},
	\begin{equation*}
		\norm*{\paren*{W^T W}^{-1} W^T \Xhat}
		= \norm*{\Thetahat \, \Q^T}
		\le \norm*{\Thetahat \, \Q^T - \Thetatilde} + \norm*{\Thetatilde}
		= \op{n^{-1/2}} + \Op{1} = \Op{1}.
	\end{equation*}
	Combining the above two displays,
	\begin{equation*}
		\sqrt n \, \norm*{\paren*{W^T W}^{-1} W^T \Xhat}
		\norm*{\Q^T \betaxtilde - \betaxhat}
		= \op{1}.
	\end{equation*}
	Using this and Equation~\eqref{eq:betaw:rate1} to bound Equation~\eqref{eq:betaw:step1} completes the proof.
\end{proof}

\subsection{Technical Preliminaries for Supporting Lemmas}

The main technical components of our proofs are a series of concentration bounds similar to those in \cite{levin2022a}. See \cite{athreya2018} for an overview of proof techniques specialized to the RDPG setting.

Many of our results rely on the concentration $\Xhat$ around $\X$ and $A$ around $\Apop$.

\begin{lemma}[\cite{levin2022a}, Theorem 6]
	\label{lem:subgamma-2toinfty}
	Under Assumptions \ref{ass:subgamma} and \ref{ass:second-stage}, with probability at least $1 - \bigoh{n^{-2}}$, there exists an orthogonal matrix $\Q \in \mathbb{R}^{d \times d}$ such that
	\begin{equation} \label{eq:ASE:2-to-infty:bound}
		\norm*{\Xhat - \X \Q}_{2, \infty} \leq \eta_n
	\end{equation}
	where $\eta_n$ is defined to be
	\begin{equation} \label{eq:def:etan}
		\eta_n = \frac{C \, d}{\lambda_{d}^{1 / 2}} (\nu_n + b_n^2)^{1/2}\log n
		+ \frac{C \, d \, n \, \lambda_1}{\lambda_{d}^{5 / 2}} \paren*{\nu_n + b_n^2} \log^{2} n.
	\end{equation}
\end{lemma}

While Lemma~\ref{lem:subgamma-2toinfty} holds for any model satisfying Assumptions \ref{ass:subgamma} and Assumption \ref{ass:second-stage}, the result is not very interesting unless $\eta_n$ is $\littleoh{1}$. $\eta_n$ must be $\littleoh{1}$ for convergence of $\betahat$ and $\Thetahat$ in the general sub-gamma case. In the special case of a random dot product graph (Example \ref{ex:rdpg}), one can show that $\eta_n = \Op{n^{-1/2} \log n}$.
Under our growth assumptions outlined in Assumption~\ref{ass:growth-rates}, $\eta_n = \littleoh{1}$, as the next lemma shows.

\begin{lemma} \label{lem:subgamma-2toinfty:littleoh1}
	Letting $\eta_n$ be as defined in Lemma~\ref{lem:subgamma-2toinfty}, under the growth conditions of Assumption~\ref{ass:growth-rates}, we have $\eta_n = o(1)$.
	Further, under the additional Assumptions~\ref{ass:subgamma} and~\ref{ass:second-stage},
	\begin{equation*}
		\norm*{\Xhat - \X \Q}_{2, \infty} = \op{1}.
	\end{equation*}
\end{lemma}
\begin{proof}
	By Lemma~\ref{lem:subgamma-2toinfty} and the Borel-Cantelli lemma, there exists a sequence of orthogonal matrices $Q \in \mathbb{R}^{d \times d}$ such that, eventually,
	\begin{equation} \label{eq:bigohetan}
		\norm*{\Xhat - \X \Q}_{2, \infty} \le \eta_n.
	\end{equation}
	Applying the definition of $\eta_n$ given in Lemma~\ref{lem:subgamma-2toinfty},
	\begin{equation} \label{eq:eta:recall}
		\eta_n
		= \frac{C \sqrt{\nu_n + b_n^2}\log n }{ \lambda_d^{1/2} }
		+ \frac{ C \lambda_1 (\nu_n+b_n^2) n \log^2 n }{ \lambda_d^{5/2} }.
	\end{equation}
	Using the fact that $\lambda_d \le \lambda_1$ by definition and applying our growth assumptions in Equations~\eqref{eq:growth:lambda1:UB}\eqref{eq:growth:46MT1},
	\begin{equation*}
		\frac{C \sqrt{\nu_n + b_n^2}\log n }{ \lambda_d^{1/2} } = o(1).
	\end{equation*}
	Similarly, using $\lambda_d \le \lambda_1 = \bigoh{n}$ and Equation~\eqref{eq:growth:46MT5a},
	\begin{equation*}
		\frac{ C \lambda_1 (\nu_n+b_n^2) n \log^2 n }{ \lambda_d^{5/2} } = o(1).
	\end{equation*}
	Applying the above two displays to Equation~\eqref{eq:eta:recall},
	\begin{equation*}
		\eta_n = o( 1 ),
	\end{equation*}
	establishing our desired growth rate on $\eta_n$.
	Applying this to Equation~\eqref{eq:bigohetan},
	\begin{equation*}
		\norm*{\Xhat - \X \Q}_{2, \infty} = \op{ 1 },
	\end{equation*}
	as we set out to show.
\end{proof}

Several other technical results will also prove useful. We collect them below.

\begin{lemma} \label{lem:spectralnorms}
	Under Assumptions \ref{ass:subgamma} and \ref{ass:second-stage}, with probability at least $1 - \bigoh{n^{-2}}$ we have
	\begin{equation*}
		\norm*{\Shat^{-1/2}} \le C \lambda_d^{-1/2} \quad \text{and} \quad
		\norm*{\Shat^{1/2}}  \le C \lambda_1^{1/2}
	\end{equation*}
	for some universal constant $C > 0$.
\end{lemma}
\begin{proof}
	Both of these facts are shown in the course of proving Lemma 4 of \cite{levin2022a}, in particular see Equations (28) and (32) in that work.
\end{proof}

The following two lemmas are fundamental for determining rates of concentration throughout our proofs. Our goal is to produce bounds under very general assumptions on $A$, and as a result, under additional assumptions, it will often be possible to improve rates of convergence under specialized assumptions. For example, under the additional assumption that $A$ is binary, \cite[Theorem 5.2]{lei2015} produces a notable improvement over a generic sub-gamma bound. We do not pursue specialized bounds here.

\begin{lemma}[\cite{levin2022a}, Lemma 5, taking $N = 1$]
	\label{lem:Aconcentrates}

	Under Assumption \ref{ass:subgamma} and \ref{ass:second-stage}, with probability at least $1 - \bigoh{n^{-2}}$,
	\begin{align*}
		\norm*{A - \Apop} \le C \sqrt{\nu_n + b_n^2} \sqrt{n} \log n.
	\end{align*}
\end{lemma}

\begin{lemma}
	\label{lem:UAPHUT:frob}

	Suppose that Assumptions \ref{ass:subgamma} and \ref{ass:second-stage} hold and let $H \in \R^{n \times n}$ be a fixed matrix satisfying
	\begin{equation} \label{eq:Hbound}
		\max_{i \in [n]} \sum_{j=1}^n H_{ij}^2 \le C_H
	\end{equation}
	for some constant $C_H \ge 0$.  Then, with notation as above, with probability at least $1- \bigoh{n^{-2}}$,
	\begin{equation*}
		\norm*{\Upop^T (A - \Apop) H \Upop}_F \le C d \sqrt{\nu_n + b_n^2} \log n
	\end{equation*}
\end{lemma}

\begin{proof}
	We will show that $\norm*{\Upop^T (A - \Apop) H \Upop}_F^2 \ge C d^2 (\nu_n + b_n^2) \log^2 n$ with probability no larger than $\bigoh{n^{-2}}$, whence taking square roots will yield the result.

	For each $k, \ell \in [d]$, define
	\begin{equation*}
		S_{k, \ell}
		= \brac*{\Upop^T (A - \Apop) H \Upop}_{k, \ell}
		= \sum_{i=1}^n \sum_{j=1}^n (A - \Apop)_{ij} \Upop_{ik} (H \Upop)_{j \ell}
	\end{equation*}
	and note that
	\begin{equation} \label{eq:frob:decomp}
		\norm*{\Upop^T(A - \Apop) H \Upop}_F^2 = \sum_{k=1}^d \sum_{\ell=1}^d S_{k, \ell}^2.
	\end{equation}

	Since $(A - \Apop)_{ij}$ are i.i.d. $(\nu_n, b_n)$-sub-gamma, we have
	\begin{equation*}
		\sum_{i=1}^n \sum_{j=1}^n \E{\brac*{(A - \Apop)_{ij} \Upop_{ik} (H \Upop)_{j \ell}}^2}
		< \nu_n \sum_{i=1}^n \sum_{j=1}^n \Upop_{ik}^2 (H \Upop)_{j \ell}^2
	\end{equation*}
	and thus by Corollary 2.11 in \cite{boucheron2013}, for any $t > 0$,
	\begin{align*}
		\P{\abs{S_{k, \ell}} \ge t}
		\le 2 \exp \set*{\frac{-t^2}{2 \paren*{\nu_n \sum_{i=1}^n \sum_{j=1}^n \Upop_{ik}^2 (H \Upop)_{j \ell}^2 + b_n t}}}.
	\end{align*}

	By Cauchy-Schwarz and our assumption in Equation~\eqref{eq:Hbound},
	\begin{equation*}
		(H \Upop)_{j \ell}^2
		= \paren*{\sum_{t=1}^n H_{jt} \Upop_{t \ell}}^2
		\le \paren*{\sum_{t=1}^n H_{jt}^2} \paren*{\sum_{t=1}^n \Upop_{t\ell}^2}
		\le C_H,
	\end{equation*}
	and it follows that
	\begin{equation*}
		\P{\abs*{S_{k, \ell}} \ge t}
		\le 2 \exp \set*{\frac{-t^2}{2(C_H \nu_n + b_n t)}}.
	\end{equation*}

	Taking $t = C(\nu_n + b_n^2)^{1/2} \log n$ for $C > 0$ suitably large, it follows that
	\begin{equation*}
		\P{\abs*{S_{k, \ell}} \ge C(\nu_n + b_n^2)^{1/2} \log n} \le 2n^{-4}.
	\end{equation*}

	A union bound over all $k, \ell \in [d]$ implies that

	\begin{equation*}
		\P{ \exists~k, \ell \in [d] : \abs*{S_{k, \ell}} \ge C(\nu_n + b_n^2)^{1/2} \log n} \le \frac{2d^2}{n^4} \le 2n^{-2},
	\end{equation*}

	and it follows from Equation~\eqref{eq:frob:decomp} that

	\begin{equation*}
		\P{ \norm*{\Upop^T (A - \Apop) H \Upop}_F^2 \ge C d^2 (\nu_n + b_n^2) \log^2 n}
		\le 2n^{-2},
	\end{equation*}

	completing the proof.
\end{proof}

We now define of a convenient decomposition of $\Xhat - \X \Q$.

\begin{lemma}[\citealt{levin2022a}, Lemma 4]
	\label{lem:decomposition}

	Define the following three matrices:
	\begin{equation*}
		\begin{aligned}
			R_1 & = \Upop \Upop^T \Uhat - \Upop \Q  \\
			R_2 & = \Q \Shat^{1/2} - \Spop^{1/2} \Q \\
			R_3 &                                   %
			= \Uhat - \Upop \Upop^T \Uhat + R_1
			= \Uhat - \Upop \Q.
		\end{aligned}
	\end{equation*}
	Then
	\begin{equation*}
		\begin{aligned}
			\Xhat - \X \Q
			 & = \Uhat \Shat^{1/2} - \Upop \Spop^{1/2} \Q                                                  \\
			 & = (A - \Apop) \Upop \Spop^{-1/2} \Q + (A - \Apop) \Upop (\Q \Shat^{-1/2} - \Spop^{-1/2} \Q) \\
			 & \qquad + \Upop \Upop^T (A - \Apop) \Upop \Q \Shat^{-1/2} + R_1 \Shat^{1/2} + \Upop R_2      \\
			 & \qquad + (I - \Upop \Upop^T)(A - \Apop)R_3 \Shat^{-1/2}.
		\end{aligned}
	\end{equation*}
\end{lemma}

Our proofs will rely on bounding each of the terms in the decomposition given in Lemma~\ref{lem:decomposition}.
The next few technical results will be used to ensure these bounds.

\begin{proposition}[\citealt{levin2022a}, Proposition 19]
	\label{prop:principalangles}

	With notation as above, under Assumptions \ref{ass:subgamma} and \ref{ass:second-stage}, it holds with probability at least $1-\bigoh{n^{-2}}$ that

	\begin{align*}
		\norm*{R_1}_F =
		\norm*{U \paren*{\Upop^T \Uhat - \Q}}_F
		= \norm*{\Upop^T \Uhat - \Q}_F
		\le \frac{d \norm*{A - \Apop}^2}{\lambda_d^2}
		\le \frac{C d (\nu_n + b_n^2) \, n \log^2 n}{\lambda_d^2}.
	\end{align*}
\end{proposition}
\begin{proof}
	By Proposition 19 of \cite{levin2022a} and Lemma \ref{lem:Aconcentrates}.
\end{proof}

\begin{lemma} \label{lem:components}
	Under Assumptions \ref{ass:subgamma} and \ref{ass:second-stage}, with probability at least $1 - \bigoh{n^{-2}}$,
	\begin{equation} \label{eq:halfR1}
		\norm*{\Uhat - \Upop \Upop^T \Uhat}_F
		\le \frac{C \sqrt{d} \ \norm*{A - \Apop}}{\lambda_d}
		\le \frac{C \sqrt{d} \sqrt{\nu_n + b_n^2} \sqrt{n} \log n}{\lambda_d}.
	\end{equation}

	Furthermore,

	\begin{align}
		\label{eq:comp1}
		\norm*{\Q \Shat - \Spop \Q}_F
		 & \le \frac{C \lambda_1 (\nu_n + b_n^2) n \log^2 n }{\lambda_d^2}
		+ C d \sqrt{\nu_n + b_n^2} \log n                                      \\
		\label{eq:comp2}
		\norm*{R_2}_F = \norm*{\Q \Shat^{1/2} - \Spop^{1/2} \Q}_F
		 & \le \frac{C \lambda_1 (\nu_n + b_n^2) n \log^2 n }{\lambda_d^{5/2}}
		+ \frac{C d \sqrt{\nu_n + b_n^2} \log n}{\lambda_d^{1/2}}
		 & \text{and,}                                                         \\
		\label{eq:comp3}
		\norm*{\Q \Shat^{-1/2} - \Spop^{-1/2} \Q}_F
		 & \le \frac{C \lambda_1 (\nu_n + b_n^2) n \log^2 n }{\lambda_d^{7/2}}
		+ \frac{C d \sqrt{\nu_n + b_n^2} \log n}{\lambda_d^{3/2}}
	\end{align}
\end{lemma}
\begin{proof}
	By Proposition 20 of \cite{levin2022a} and an application of Lemma \ref{lem:Aconcentrates} we obtain \eqref{eq:halfR1}. Further, by Proposition 20 of \cite{levin2022a}, we have
	\begin{align*}
		\norm*{\Q \Shat - \Spop \Q}_F
		 & \le \frac{C \norm*{A-\Apop}^2 \lambda_1}{\lambda_d^2}
		+ \norm*{\Upop^T(A - \Apop) \Upop}_F                          \\
		\norm*{\Q \Shat^{1/2} - \Spop^{1/2} \Q}_F
		 & \le \frac{\norm*{\Q \Shat - \Spop \Q}_F}{\lambda_d^{1/2}}, \\
		\text{and }
		\norm*{\Q \Shat^{-1/2} - \Spop^{-1/2} \Q}_F
		 & \le \frac{\norm*{\Q \Shat - \Spop \Q}_F}{\lambda_d^{3/2}}.
	\end{align*}

	First we apply Lemma \ref{lem:Aconcentrates} and Lemma \ref{lem:UAPHUT:frob} to bound the top term
	\begin{equation*} \begin{aligned}
			\norm*{\Q \Shat - \Spop \Q}_F
			 & \le \frac{C \norm*{A - \Apop}^2 \lambda_1}{\lambda_d^2}
			+ \norm*{\Upop^T(A - \Apop) \Upop}_F                              \\
			 & \le \frac{C (\nu_n + b_n^2) n \log^2 n \lambda_1}{\lambda_d^2}
			+ C d \sqrt{\nu_n + b_n^2} \log n
		\end{aligned} \end{equation*}
	and Equations~\eqref{eq:comp2} and~\eqref{eq:comp3} follow immediately.
\end{proof}

\begin{lemma} \label{lem:UhattoUQ}
	Under Assumptions \ref{ass:subgamma} and \ref{ass:second-stage}, it holds with probability $1 - \bigoh{n^{-2}}$ that
	\begin{equation*}
		\norm*{\Uhat - \Upop \Q}
		= \norm*{R_3}
		\le \frac{C \sqrt{d} \sqrt{\nu_n + b_n^2} \sqrt{n} \log n}{\lambda_d} +
		\frac{C d (\nu_n + b_n^2) \, n \log^2 n}{\lambda_d^2}
	\end{equation*}
\end{lemma}
\begin{proof}
	Adding and subtracting appropriate quantities, applying the triangle inequality and using basic properties of the Frobenius norm,
	\begin{equation*} \begin{aligned}
			\norm*{\Uhat - \Upop \Q}
			 & \le  \norm*{\Uhat - \Upop \Upop^T \Uhat} +
			\norm*{\Upop \Upop^T \Uhat - \Upop \Q}_F      \\
			 & \le \norm*{\Uhat - \Upop \Upop^T \Uhat} +
			\norm*{\Upop^T \Uhat - \Q}_F .
		\end{aligned} \end{equation*}

	Applying Lemmas \ref{prop:principalangles} and \ref{lem:components}, it follows that with probability at least $1 - \bigoh{n^{-2}}$,
	\begin{equation*}
		\norm*{\Uhat - \Upop \Q}
		= \norm*{R_3}
		\le \frac{C \sqrt{d} \sqrt{\nu_n + b_n^2} \sqrt{n} \log n}{\lambda_d} + \frac{C d (\nu_n + b_n^2) \, n \log^2 n}{\lambda_d^2},
	\end{equation*}
	completing the proof.
\end{proof}

\subsection{Supporting Results for Lemma~\ref{lem:op1} and Lemma~\ref{lem:covariance}}
\label{sec:proofs:lemop1-lemcovar}

When we introduced Lemma~\ref{lem:op1} and Lemma~\ref{lem:covariance} earlier, we presented a broad proof sketch and deferred the technical details to supporting lemmas, which we now present.

\begin{lemma} \label{lem:XhatXMX}
	Suppose that Assumptions~\ref{ass:subgamma},~\ref{ass:second-stage} and~\ref{ass:growth-rates} hold.
	Then
	\begin{equation*}
		\norm*{(\Xhat \Q^T - \X)^T M \X} = \op{ \frac{ \lambda_d }{ \sqrt{n} } }.
	\end{equation*}
\end{lemma}
\begin{proof}
	Applying Lemma~\ref{lem:decomposition},
	\begin{equation} \label{eq:XhatXMX-decomp}
		\begin{aligned}
			 & (\Xhat \Q^T - \X)^T M \X                                                    \\
			 & \qquad = \Q \paren*{\Uhat \Shat^{1/2} - \Upop \Spop^{1/2} \Q}^T M \X        \\
			 & \qquad = \Q \Spop^{-1/2} \Upop^T (A - \Apop) M \X
			+ \Q (\Q \Shat^{-1/2} - \Spop^{-1/2} \Q)^T \Upop^T (A - \Apop) M \X            \\
			 & \qquad \qquad + \Q \Shat^{-1/2} \Q^T \Upop^T (A - \Apop) \Upop \Upop^T M \X
			+ \Q \Shat^{1/2} R_1^T M \X + \Q R_2^T \Upop^T M \X                            \\
			 & \qquad \qquad + \Q \Shat^{-1/2} R_3^T (I - \Upop \Upop^T)(A - \Apop) M \X.
		\end{aligned}
	\end{equation}
	We will bound each of the six terms on the right-hand side in turn.

	Considering the first term, expanding the definition of $\X$ and using submultiplicativity of the norm, with probability $1 - \bigoh{n^{-2}}$,
	\begin{equation*}
		\begin{aligned}
			\norm*{\Q \Spop^{-1/2} \Upop^T (A - \Apop) M \X}
			 & \le \norm*{\Spop^{-1/2}} \norm*{\Upop^T (A - \Apop) M \Upop} \norm*{\Spop^{1/2}} \\
			 & \le \frac{C d \lambda_1^{1/2} \sqrt{\nu_n + b_n^2} \log n}{\lambda_d^{1/2}},
		\end{aligned} \end{equation*}
	where the second bound follows from Lemma \ref{lem:UAPHUT:frob}.
	Applying Equations~\eqref{eq:growth:lambda1:UB} and~\eqref{eq:growth:46MT1},
	\begin{equation} \label{eq:XhatXMX:bound1}
		\norm*{\Q \Spop^{-1/2} \Upop^T (A - \Apop) M \X}
		= \op{ \frac{ \lambda_d }{ \sqrt{n} } }.
	\end{equation}

	For the second term in Equation~\eqref{eq:XhatXMX-decomp}, we again use submultiplicativity of the spectral norm, along with Equation~\eqref{eq:comp3} from Lemma~\ref{lem:components} and Lemma~\ref{lem:UAPHUT:frob}, to show that with probability $1 - \bigoh{n^{-2}}$,
	\begin{equation*} \begin{aligned}
			 & \norm*{\Q (\Q \Shat^{-1/2} - \Spop^{-1/2} \Q)^T \Upop^T (A - \Apop) M \X}                                      \\
			 & \qquad \le \norm*{\Q \Shat^{-1/2} - \Spop^{-1/2} \Q} \norm*{\Upop^T (A - \Apop) M \Upop} \norm*{ \Spop^{1/2} } \\
			 & \qquad \le
			C \lambda_1^{1/2} \paren*{
				\frac{\lambda_1 (\nu_n + b_n^2) n \log^2 n }{\lambda_d^{7/2}}
				+ \frac{ d \sqrt{\nu_n + b_n^2} \log n}{\lambda_d^{3/2}}
			} d \sqrt{\nu_n + b_n^2} \log n                                                                                   \\
			 & \qquad \le
			\frac{C d \lambda_1^{3/2} (\nu_n + b_n^2)^{3/2} n \log^3 n }
			{\lambda_d^{7/2}}
			+ \frac{C d^2 \lambda_1^{1/2} (\nu_n + b_n^2) \log^2 n }
			{\lambda_d^{3/2}}.
		\end{aligned} \end{equation*}
	Applying Equation~\eqref{eq:growth:lambda1:UB} and cubing the quantity in Equation~\eqref{eq:growth:46MT1} implies that the first of these two right-hand side quantities is $o(\lambda_d/\sqrt{n})$.
	Similarly, Equations~\eqref{eq:growth:lambda1:UB} and~\eqref{eq:growth:46MT5a} imply that the second right-hand term is $o(\lambda_d/\sqrt{n})$, whence
	\begin{equation} \label{eq:XhatXMX:bound2}
		\norm*{\Q (\Q \Shat^{-1/2} - \Spop^{-1/2} \Q)^T \Upop^T (A - \Apop) M \X} \\
		= \op{ \frac{ \lambda_d }{ \sqrt{n} } }.
	\end{equation}

	For the third term in Equation~\eqref{eq:XhatXMX-decomp}, by Lemmas~\ref{lem:spectralnorms} and~\ref{lem:UAPHUT:frob}, we have
	\begin{equation*} \begin{aligned}
			\norm*{\Q \Shat^{-1/2} \Q^T \Upop^T (A - \Apop) \Upop \Upop^T M \X}
			 & \le \norm*{\Shat^{-1/2}} \norm*{\Upop^T (A - \Apop) \Upop} \norm*{\Spop^{1/2}} \\
			 & \le \frac{ C d \lambda_1^{1/2} \sqrt{ \nu_n + b_n^2 } \log n }
			{ \lambda_d^{1/2} }.
		\end{aligned} \end{equation*}
	Applying Equations~\eqref{eq:growth:lambda1:UB} and~\eqref{eq:growth:46MT1} along with the trivial $\log n = \Omega( 1 )$, we have
	\begin{equation} \label{eq:XhatXMX:bound3}
		\norm*{\Q \Shat^{-1/2} \Q^T \Upop^T (A - \Apop) \Upop \Upop^T M \X}
		= \op{ \frac{ \lambda_d }{ \sqrt{n} } }.
	\end{equation}

	For the fourth term in Equation~\eqref{eq:XhatXMX-decomp}, recalling $R_1 = \Upop \Upop^T \Uhat - \Upop \Q = \Upop(\Upop^T \Uhat - \Q)$, observe that
	\begin{equation*}
		\Q \Shat^{1/2} R_1^T M \X
		= \Q \Shat^{1/2} (\Upop^T \Uhat - \Q)^T \Upop^T M \Upop \Spop^{1/2},
	\end{equation*}
	whence Lemma~\ref{lem:spectralnorms} and Proposition~\ref{prop:principalangles} imply that
	\begin{equation*} \begin{aligned}
			\norm*{\Q \Shat^{1/2} R_1^T M \X}
			 & \le \norm*{\Shat^{1/2}} \norm*{\Upop^T \Uhat - \Q}_F \norm*{M \Upop} \norm*{\Spop^{1/2}} \\
			 & \le \frac{C d \lambda_1 (\nu_n + b_n^2) n \log^2 n}{\lambda_d^2}.
		\end{aligned} \end{equation*}
	Using Equation~\eqref{eq:growth:lambda1:UB} and the trivial $\lambda_1 \ge \lambda_d$, Equation~\eqref{eq:growth:46MT5a} then implies
	\begin{equation} \label{eq:XhatXMX:bound4}
		\norm*{\Q \Shat^{1/2} R_1^T M \X} = \op{ \frac{ \lambda_d }{ \sqrt{n} } }.
	\end{equation}

	Similarly, for the fifth term in Equation~\eqref{eq:XhatXMX-decomp}, applying submultiplicativity followed by Equation~\eqref{eq:comp2} and Lemma~\ref{lem:spectralnorms},
	\begin{equation*}
		\norm*{\Q R_2^T \Upop^T M \X}
		\le \norm*{R_2} \norm*{\Spop^{1/2}}
		\le
		\frac{C \lambda_1^{3/2} (\nu_n + b_n^2) n \log^2 n }{\lambda_d^{5/2}}
		+ \frac{C d\lambda_1^{1/2} \sqrt{\nu_n + b_n^2} \log n }
		{\lambda_d^{1/2}}.
	\end{equation*}
	Equations~\eqref{eq:growth:lambda1:UB} and~\eqref{eq:growth:46MT5a} control the first of these terms, while Equations~\eqref{eq:growth:lambda1:UB} and~\eqref{eq:growth:46MT1} control the second, and we conclude that
	\begin{equation} \label{eq:XhatXMX:bound5}
		\norm*{\Q R_2^T \Upop^T M \X} = \op{ \frac{ \lambda_d }{ \sqrt{n} } }.
	\end{equation}

	For the sixth term in Equation~\eqref{eq:XhatXMX-decomp}, we expand $R_3$ to write
	\begin{align}
		 & \Q \Shat^{-1/2} R_3^T (I - \Upop \Upop^T)(A - \Apop) M \X                             \notag           \\
		 & \qquad
		\label{eq:XhatXMX:bound6:1}
		= \Q \Shat^{-1/2}(\Uhat - \Upop \Upop^T \Uhat)^T (I - \Upop \Upop^T)(A - \Apop) M \X                      \\
		\label{eq:XhatXMX:bound6:2}
		 & \qquad \qquad + \Q \Shat^{-1/2}(\Upop \Upop^T \Uhat - \Upop \Q)^T (I - \Upop \Upop^T)(A - \Apop) M \X.
	\end{align}
	To bound~\eqref{eq:XhatXMX:bound6:1}, we use submultiplicativity, Lemma~\ref{lem:spectralnorms}, \eqref{eq:halfR1} of Lemma~\ref{lem:components}, and Lemma~\ref{lem:Aconcentrates} to write
	\begin{equation*} \begin{aligned}
			 & \norm*{\Q \Shat^{-1/2}(\Uhat - \Upop \Upop^T \Uhat)^T (I - \Upop \Upop^T)(A - \Apop) M \X}                                                              \\
			 & \qquad \le C \norm*{\Shat^{-1/2}} \norm*{\Uhat - \Upop \Upop^T \Uhat} \norm*{A-\Apop} \norm*{\Spop^{1/2}}                                               \\
			 & \qquad \le C \lambda_d^{-1/2} \frac{C \sqrt{d} \sqrt{\nu_n + b_n^2} \sqrt{n} \log n}{\lambda_d}  C \lambda_1^{1/2} \sqrt{\nu_n + b_n^2} \sqrt{n} \log n \\
			 & \qquad \le \frac{C \sqrt{d} \lambda_1^{1/2} (\nu_n + b_n^2) n \log^2 n}{\lambda_d^{3/2}} ,
		\end{aligned} \end{equation*}
	Applying Equations~\eqref{eq:growth:lambda1:UB} and~\eqref{eq:growth:46MT5a},
	\begin{equation} \label{eq:XhatXMX:bound6:1:done}
		\norm*{\Q \Shat^{-1/2}(\Uhat - \Upop \Upop^T \Uhat)^T (I - \Upop \Upop^T)(A - \Apop) M \X}
		= \op{ \frac{ \lambda_d }{ \sqrt{n} } } .
	\end{equation}

	To bound \eqref{eq:XhatXMX:bound6:2}, we apply Lemma~\ref{lem:spectralnorms}, Proposition~\ref{prop:principalangles} and Lemma~\ref{lem:Aconcentrates} to see that
	\begin{equation*} \begin{aligned}
			 & \norm*{\Q \Shat^{-1/2}(\Upop \Upop^T \Uhat - \Upop \Q)^T (I - \Upop \Upop^T)(A - \Apop) M \X}                                            \\
			 & \qquad \le \norm*{\Shat^{-1/2}} \norm*{\Upop^T \Uhat - \Q} \norm*{A - \Apop} \norm*{M \Upop} \norm*{\Spop^{1/2}}                         \\
			 & \qquad \le C \lambda_d^{-1/2} \frac{C d (\nu_n + b_n^2) \, n \log^2 n}{\lambda_d^2} \lambda_1^{1/2} \sqrt{\nu_n + b_n^2} \sqrt{n} \log n \\
			 & \qquad \le \frac{C \lambda_1^{1/2} d (\nu_n + b_n^2)^{3/2} \, n^{3/2} \log^3 n}{\lambda_d^{5/2}} .
		\end{aligned} \end{equation*}
	Applying Equations~\eqref{eq:growth:lambda1:UB},~\eqref{eq:growth:46MT1} and~\eqref{eq:growth:46MT5a},
	\begin{equation} \label{eq:XhatXMX:bound6:2:done}
		\norm*{\Q \Shat^{-1/2}(\Upop \Upop^T \Uhat - \Upop \Q)^T
			(I - \Upop \Upop^T)(A - \Apop) M \X}
		= \op{ \frac{ \lambda_d }{ \sqrt{n} } }.
	\end{equation}

	Using Equations~\eqref{eq:XhatXMX:bound6:1:done} and~\eqref{eq:XhatXMX:bound6:2:done}, respectively, to bound~\eqref{eq:XhatXMX:bound6:1} and~\eqref{eq:XhatXMX:bound6:2}, we conclude that
	\begin{equation} \label{eq:XhatXMX:bound6}
		\norm*{ \Q \Shat^{-1/2} R_3^T (I - \Upop \Upop^T)(A - \Apop) M \X }
		= \op{ \frac{ \lambda_d }{ \sqrt{n} } }.
	\end{equation}

	Applying Equations~\eqref{eq:XhatXMX:bound1},~\eqref{eq:XhatXMX:bound2},~\eqref{eq:XhatXMX:bound3},~\eqref{eq:XhatXMX:bound4},~\eqref{eq:XhatXMX:bound5} and~\eqref{eq:XhatXMX:bound6} to the right-hand side of Equation~\eqref{eq:XhatXMX-decomp}, we conclude that
	\begin{equation*}
		\norm*{(\Xhat \Q^T - \X)^T M \X} = \op{ \frac{ \lambda_d }{ \sqrt{n} } },
	\end{equation*}
	completing the proof.
\end{proof}

\begin{lemma} \label{lem:subgamma-XhatX:spectral}
	Under Assumptions \ref{ass:subgamma}, \ref{ass:second-stage} and~\ref{ass:growth-rates}, it holds with high probability that
	\begin{equation*} %
		\norm*{ \X \Q - \Xhat }
		= \Op{ \frac{ \lambda_1 (\nu_n + b_n^2) n \log^2 n }{ \lambda_d^{5/2} } }
		+ \Op{ \frac{ \lambda_1^{1/2} \sqrt{\nu_n + b_n^2} \sqrt{n} \log n}                    {\lambda_d} }.
	\end{equation*}
\end{lemma}
\begin{proof}
	Using basic properties of the norm and the triangle inequality,
	\begin{equation} \label{eq:XQXhat:spec:tri} \begin{aligned}
			\norm*{ \X \Q - \Xhat }
			 & \le
			\norm*{ \Upop (\Spop^{1/2} \Q - \Q \Shat^{1/2}) }
			+ \norm*{ \Upop \Q \Shat^{1/2} - \Uhat \Shat^{1/2} } \\
			 & \le
			\norm*{ \Spop^{1/2} \Q - \Q \Shat^{1/2} }
			+ \norm*{ \Upop \Q - \Uhat } \norm*{ \Shat^{1/2} }.
		\end{aligned} \end{equation}
	By Lemma~\ref{lem:components}, it holds with high probability that
	\begin{equation*}
		\norm*{ \Spop^{1/2} \Q - \Q \Shat^{1/2} }
		\le
		\frac{C \lambda_1 (\nu_n + b_n^2) n \log^2 n }{ \lambda_d^{5/2} }
		+ \frac{C d \sqrt{\nu_n + b_n^2} \log n}{ \lambda_d^{1/2} }
	\end{equation*}
	Applying Lemmas~\ref{lem:spectralnorms} and~\ref{lem:UhattoUQ}, it holds with high probability that
	\begin{equation*} %
		\norm*{ \Upop \Q - \Uhat } \norm*{ \Shat^{1/2} }
		\le
		\frac{C \sqrt{d} \lambda_1^{1/2} \sqrt{\nu_n + b_n^2} \sqrt{n} \log n}{\lambda_d} +
		\frac{C d \lambda_1^{1/2} (\nu_n + b_n^2) \, n \log^2 n}{\lambda_d^2}.
	\end{equation*}
	Applying the above two displays to Equation~\eqref{eq:XQXhat:spec:tri},
	\begin{equation*} \begin{aligned}
			\norm*{ \X \Q - \Xhat }
			 & = \Op{ \frac{ \lambda_1 (\nu_n + b_n^2) n \log^2 n }{ \lambda_d^{5/2} } }
			+ \Op{ \frac{ \sqrt{\nu_n + b_n^2} \log n}{ \lambda_d^{1/2} } }              \\
			 & ~~~~~~+ \Op{ \frac{ \lambda_1^{1/2} \sqrt{\nu_n + b_n^2} \sqrt{n} \log n}
				{\lambda_d} }
			+ \Op{ \frac{\lambda_1^{1/2} (\nu_n + b_n^2) n \log^2 n}{\lambda_d^2} }.
		\end{aligned} \end{equation*}
	Collecting terms and using the fact that $\lambda_d \le \lambda_1$ ,
	\begin{equation*}
		\norm*{ \X \Q - \Xhat }
		= \Op{ \frac{ \lambda_1 (\nu_n + b_n^2) n \log^2 n }{ \lambda_d^{5/2} } }
		+ \Op{ \frac{ \lambda_1^{1/2} \sqrt{\nu_n + b_n^2} \sqrt{n} \log n}                    {\lambda_d} },
	\end{equation*}
	completing the proof.
\end{proof}

\begin{lemma} \label{lem:XMXXhatMXhat}
	Under Assumptions \ref{ass:subgamma}, \ref{ass:second-stage} and~\ref{ass:growth-rates}, for any orthogonal projection matrix $M$,
	\begin{equation*}
		\norm*{\Q^T \X^T M \X - \Xhat^T M \Xhat \, \Q^T}
		= \op{ \frac{ \lambda_d^2 }{ \sqrt{n} \lambda_1^{1/2} } }.
	\end{equation*}
\end{lemma}
\begin{proof}
	Adding and subtracting appropriate quantities, applying the triangle inequality and using basic properties of the spectral norm,
	\begin{equation*} \begin{aligned}
			\norm*{\Q^T \X^T M \X - \Xhat^T M \Xhat \, \Q^T}
			 & \le
			\norm*{ (\X \Q)^T M (\X-\Xhat \Q^T) }
			+ \norm*{ (\X \Q - \Xhat)^T M (\Xhat \Q^T - \X) }          \\
			 & ~~~~~~~~~~~~~~~~~~~~~+ \norm*{ (\X \Q - \Xhat)^T M \X } \\
			 & \le
			2 \norm*{ (\Xhat \Q^T - \X)^T M \X }
			+ \norm*{ \X \Q - \Xhat }^2 \norm*{ M }.
		\end{aligned} \end{equation*}
	Using the fact that $\norm*{M} \le 1$ and applying Lemma~\ref{lem:XhatXMX},
	\begin{equation} \label{eq:XMX:XhatMXhat:inter} \begin{aligned}
			\norm*{\Q^T \X^T M \X - \Xhat^T M \Xhat \, \Q^T}
			 & \le \norm*{ \X \Q - \Xhat }^2 + \op{ \frac{ \lambda_d }{ \sqrt{n} } }
		\end{aligned} \end{equation}
	By Lemma~\ref{lem:subgamma-XhatX:spectral},
	\begin{equation*}
		\norm*{ \X \Q - \Xhat }^2
		= \Op{ \frac{ \lambda_1^2 (\nu_n + b_n^2)^2 n2 \log^4 n }{ \lambda_d^5 } }
		+ \Op{ \frac{ \lambda_1 (\nu_n + b_n^2) n \log^2 n}{\lambda_d^2} }.
	\end{equation*}
	Applying our growth assumption in Equation~\eqref{eq:growth:46MT5a} controls the first of these terms as $o( \lambda_d/\sqrt{n} )$, while Equations~\eqref{eq:growth:lambdad:LB} and~\eqref{eq:growth:46MT5a} control the second of these terms as $o( \lambda_d/\sqrt{n} )$, so that
	\begin{equation*}
		\norm*{ \X \Q - \Xhat }^2 = \op{ \frac{ \lambda_d }{ \sqrt{n} } }.
	\end{equation*}
	Applying this bound to Equation~\eqref{eq:XMX:XhatMXhat:inter},
	\begin{equation*}
		\norm*{\Q^T \X^T M \X - \Xhat^T M \Xhat \, \Q^T}
		= \op{ \frac{ \lambda_d }{ \sqrt{n} } }.
	\end{equation*}
	Applying our assumption in Equation~\eqref{eq:lambdaratio} completes the proof.
\end{proof}

\begin{lemma} \label{lem:norm-of-errors:Heps}
	Under Assumption \ref{ass:second-stage}, for any (possibly random) matrix $H$ independent of $\varepsilon$,
	\begin{equation*}
		\norm*{ H \varepsilon} = \Op{ \sqrt{B \trace H^T H } }.
	\end{equation*}
	In particular, taking $H=I$, $\norm*{ \varepsilon} = \Op{ \sqrt{B n } }$.
\end{lemma}
\begin{proof}
	We begin by noting that since $\varepsilon$ is a vector of independent mean-zero random variables,
	\begin{equation*}
		\mathbb{E} \norm*{ H \varepsilon }^2
		= \mathbb{E} \varepsilon^T H^T H \varepsilon
		\le B \trace H^T H,
	\end{equation*}
	where $B > 0$ is the bound on the variance guaranteed by Assumption~\ref{ass:second-stage}.
	Applying Markov's inequality, for any $t > 0$ and $\delta > 0$,
	\begin{equation*}
		\P{ \frac{ \norm*{ H \varepsilon }^2 }{ t } > \delta }
		\le \frac{ \mathbb{E} \norm*{ H \varepsilon }^2 }{ t \delta }
		\le \frac{ B \trace H^T H }{ t \delta }.
	\end{equation*}

	Let $r_n$ be any function of $n$ growing such that $r_n = \omega( B \trace H^TH )$.
	Then taking $t = r_n$,
	\begin{equation*}
		\lim_{n \rightarrow \infty}
		\P{ \frac{ \norm*{ H \varepsilon }^2 }{ r_n } > \delta }
		= 0.
	\end{equation*}
	Thus, $\norm*{ H \varepsilon }^2 = \op{ r_n }$ for any $r_n = \omega( B \trace H^T H )$, and it follows that
	\begin{equation*}
		\norm*{ H \varepsilon }^2
		= \Op{ B \trace H^T H }.
	\end{equation*}
	Taking square roots completes the proof.
\end{proof}

\begin{lemma} \label{lem:UAPMeps:frob}
	Under Assumptions \ref{ass:subgamma} and \ref{ass:second-stage}, let $M \in \R^{n \times n}$ satisfy $\norm*{M} = 1$. Then
	\begin{equation*}
		\norm*{\Upop^T (A - \Apop) M \varepsilon} = \Op{\sqrt{d} \log n \sqrt{\nu_n B n \log n + b_n^2}}.
	\end{equation*}
\end{lemma}
\begin{proof}
	For each $k \in [d]$, define
	\begin{equation*}
		S_k
		= \brac*{\Upop^T (A - \Apop) M \varepsilon}_{k}
		= \sum_{i=1}^n \sum_{j=1}^n (A - \Apop)_{ij} \Upop_{ik} (M \varepsilon)_{j}
	\end{equation*}
	and note that
	\begin{equation*}
		\norm*{\Upop^T(A - \Apop) M \varepsilon}_2^2 = \sum_{k=1}^d S_k^2.
	\end{equation*}

	By Corollary 2.11 in \cite{boucheron2013}, for any $t > 0$,
	\begin{align*}
		\P[\X, \varepsilon]{\abs*{S_k} \ge t}
		 & \le 2 \exp \set*{
			\frac{
				-t^2
			}{
				2 \paren*{\nu_n \sum_{i=1}^n \sum_{j=1}^n \Upop_{ik}^2 (M \varepsilon)_j^2 + b_n t}
			}
		}                                                                                         \\
		 & = 2 \exp \set*{\frac{-t^2}{2 \paren*{\nu_n \sum_{j=1}^n (M \varepsilon)_j^2 + b_n t}}} \\
		 & = 2 \exp \set*{\frac{-t^2}{2 \paren*{\nu_n \norm{M \varepsilon}^2 + b_n t}}}           \\
		 & \le 2 \exp \set*{\frac{-t^2}{2 \paren*{\nu_n \norm{\varepsilon}^2 + b_n t}}}.
	\end{align*}

	Note that we can drop the conditioning on $\X$ in the above since the bound is free of terms that depend on $\X$. We now need to drop the conditioning on $\varepsilon$. Let $G_n$ denote the event $\set*{\norm*{\varepsilon}^2 < n B \log n}$ and $G_n^c$ denote the complement of $G_n$. By a slight modification of the proof of Lemma~\ref{lem:norm-of-errors:Heps} with $H=I$, $G_n$ occurs with probability at least $1 - \frac{1}{\log n}$. Thus
	\begin{align*}
		\P{\abs*{S_k} \ge t}
		 & = \P[G_n]{\abs*{S_k} \ge t} \cdot
		\P{G_n} + \P[G_n^c]{\abs*{S_k} \ge t} \cdot
		\P{G_n^c}                                           \\
		 & \le \P[G_n]{\abs*{S_k} \ge t} + \P{G_n^c}        \\
		 & \le \P[G_n]{\abs*{S_k} \ge t} + \frac{1}{\log n}
	\end{align*}
	using our previous bound on $\P[\X, \varepsilon]{\abs*{S_k} \ge t}$. Let $\delta > 0$ be arbitrary. Taking $t = C \log n \paren*{\sqrt{\nu_n B n \log n} + b_n} \delta$ for $C > 0$ suitably large, it follows that
	\begin{align*}
		\P[G_n]{\abs*{S_k} \ge C \log n \paren*{\sqrt{\nu_n B n \log n} + b_n} \delta} \le 2 n^{-3}.
	\end{align*}

	A union bound over all $k \in [d]$ implies that
	\begin{equation*}
		\P[G_n]{\set*{\exists k \in [d] : \abs*{S_k} \ge C \log n \paren*{\sqrt{\nu_n B n \log n} + b_n} \delta}}
		\le \frac{2 d}{n^3}
		\le 2 n^{-2}
	\end{equation*}
	from we which we see that
	\begin{equation*} \begin{aligned}
			 & \P{\set*{\exists k \in [d] : \abs*{S_k} \ge C \log n \paren*{\sqrt{\nu_n B n \log n} + b_n} \delta}}                       \\
			 & \quad = \P[G_n]{\set*{\exists k \in [d] : \abs*{S_k} \ge C \log n \paren*{\sqrt{\nu_n B n \log n} + b_n} \delta}} \cdot
			\P{G_n}                                                                                                                       \\
			 & \qquad + \P[G_n^c]{\set*{\exists k \in [d] : \abs*{S_k} \ge C \log n \paren*{\sqrt{\nu_n B n \log n} + b_n} \delta}} \cdot
			\P{G_n^c}                                                                                                                     \\
			 & \quad \le \P[G_n]{\set*{\exists k \in [d] : \abs*{S_k} \ge C \log n \paren*{\sqrt{\nu_n B n \log n} + b_n} \delta}}
			+ \P{G_n^c}                                                                                                                   \\
			 & \quad \le \frac{2}{n^2} + \frac{1}{\log n}.
		\end{aligned} \end{equation*}

	Thus, $\norm*{\Upop^T (A - \Apop) M \varepsilon}_2^2 = \sum_{k=1}^d S_k^2 \ge C d \paren*{\nu_n B n \log n + b_n^2} \delta^2 \log^2 n$ with probability $2 n^{-2} + 1 / \log n$; that is,
	\begin{equation*} \begin{aligned}
			\norm*{\Upop^T (A - \Apop) M \varepsilon} = \op{\sqrt{d} \sqrt{\nu_n B n \log n + b_n^2} \log n},
		\end{aligned} \end{equation*}
	completing the proof.
\end{proof}

\begin{lemma} \label{lem:XhatXMeps}
	Suppose that Assumptions~\ref{ass:subgamma}, \ref{ass:second-stage} and~\ref{ass:growth-rates} hold.
	Then
	\begin{equation*}
		\norm*{(\Xhat \Q^T - \X)^T M \varepsilon}
		= \op{ \frac{ \lambda_d }{\sqrt{n}} }.
	\end{equation*}
\end{lemma}
\begin{proof}
	Applying Lemma~\ref{lem:decomposition} to expand $\Xhat - \X \Q$,
	\begin{equation} \label{eq:XhatXMeps-decomp} \begin{aligned}
			 & (\Xhat \Q^T - \X)^T M \varepsilon                                         \\
			 & \qquad = \Q \paren*{\Uhat \Shat^{1/2} - \Upop \Spop^{1/2} \Q}^T
			M \varepsilon                                                                \\
			 & \qquad = \Q \Spop^{-1/2} \Upop^T (A - \Apop) M \varepsilon
			+ \Q (\Q \Shat^{-1/2} - \Spop^{-1/2} \Q)^T \Upop^T (A - \Apop) M \varepsilon \\
			 & \qquad \qquad + \Q \Shat^{-1/2} \Q^T \Upop^T (A - \Apop) \Upop
			\Upop^T M \varepsilon
			+ \Q \Shat^{1/2} R_1^T M \varepsilon + \Q R_2^T \Upop^T M \varepsilon        \\
			 & \qquad \qquad + \Q \Shat^{-1/2} R_3^T (I - \Upop \Upop^T)(A - \Apop)
			M \varepsilon.
		\end{aligned} \end{equation}

	We will bound each of the six terms on the right-hand side in turn. In the first term, expanding the definition of $\X$ and using submultiplicativity of the norm,
	\begin{equation*}
		\norm*{\Q \Spop^{-1/2} \Upop^T (A - \Apop) M \varepsilon}
		\le \norm*{\Spop^{-1/2}} \norm*{\Upop^T (A - \Apop) M \varepsilon}.
	\end{equation*}
	Applying Lemmas~\ref{lem:spectralnorms} and~\ref{lem:UAPMeps:frob},
	\begin{equation*} \begin{aligned}
			\norm*{\Q \Spop^{-1/2} \Upop^T (A - \Apop) M \varepsilon}
			 & = \Op{ \frac{ (\nu_n B n \log n + b_n^2)^{1/2} \log n}{ \lambda_d^{1/2} }} \\
			 & = \Op{ \frac{  \sqrt{\nu_n + b_n^2} \sqrt{n} \log^{3/2} n}
				{ \lambda_d^{1/2} } },
		\end{aligned} \end{equation*}
	where we have used the trivial upper bound
	\begin{equation} \label{eq:bound:nubn}
		\nu_n n \log n + b_n^2 \le n (\nu_n+b_n^2) \log n
	\end{equation}
	along with the assumption that $d$ and $B$ are constant in $n$.
	Equation~\eqref{eq:growth:46MT1} then implies that
	\begin{equation} \label{eq:XhatXMeps:bound1}
		\norm*{\Q \Spop^{-1/2} \Upop^T (A - \Apop) M \varepsilon}
		= \op{ \frac{ \lambda_d }{ \sqrt{n} } }.
	\end{equation}

	For the second term in Equation~\eqref{eq:XhatXMeps-decomp}, we use submultiplicativity of the spectral norm again to write
	\begin{equation*}
		\norm*{\Q (\Q \Shat^{-1/2} - \Spop^{-1/2} \Q)^T \Upop^T (A - \Apop) M \varepsilon}
		\le \norm*{\Q \Shat^{-1/2} - \Spop^{-1/2} \Q}
		\norm*{\Upop^T (A - \Apop) M \varepsilon}.
	\end{equation*}
	Equation~\eqref{eq:comp3} from Lemma~\ref{lem:components} bounds the first multiplicand on the right-hand side, while Lemma~\ref{lem:UAPMeps:frob} bounds the second, and we have
	\begin{equation*}
		\begin{aligned}
			 & \norm*{\Q (\Q \Shat^{-1/2} - \Spop^{-1/2} \Q)^T \Upop^T (A - \Apop) M \varepsilon} \\
			 & ~~~~~~\le
			C\paren*{
				\frac{ \lambda_1 (\nu_n + b_n^2) n \log^2 n }{\lambda_d^{7/2}}
				+
				\frac{ \sqrt{\nu_n + b_n^2} \log n}{\lambda_d^{3/2}}
			} \op{ \sqrt{\nu_n B n \log n + b_n^2} \log n}                                        \\
			 & ~~~~~~=
			\op{ \frac{ \lambda_1 (\nu_n + b_n^2)^{3/2} n^{3/2} \log^{7/2} n }
				{ \lambda_d^{7/2} } }
			+
			\op{ \frac{  (\nu_n + b_n^2) n^{1/2} \log^{5/2} n }{ \lambda_d^{3/2} } }
		\end{aligned} \end{equation*}
	where we have again used the bound in Equation~\eqref{eq:bound:nubn} and our assumption that $B$ and $d$ are constants.
	Raising the quantity in Equation~\eqref{eq:growth:46MT1} to the third power and applying Equation~\eqref{eq:growth:lambda1:UB}, the first of these as $\op{ \lambda_d/\sqrt{n} }$.
	Trivially using $\lambda_1 \ge \lambda_d$ and $\log n = o( n )$, Equation~\eqref{eq:growth:46MT5a} bounds the second term by the same rate, and it follows that
	\begin{equation} \label{eq:XhatXMeps:bound2}
		\norm*{\Q (\Q \Shat^{-1/2} - \Spop^{-1/2} \Q)^T \Upop^T (A - \Apop) M \varepsilon}
		= \op{ \frac{ \lambda_d }{ \sqrt{n} } }.
	\end{equation}

	For the third term in Equation~\eqref{eq:XhatXMeps-decomp}, submultiplicativity followed by Lemmas~\ref{lem:spectralnorms},~\ref{lem:UAPHUT:frob} and~\ref{lem:norm-of-errors:Heps} yields
	\begin{equation*} \begin{aligned}
			\norm*{\Q \Shat^{-1/2} \Q^T \Upop^T (A - \Apop) \Upop \Upop^T M \varepsilon}
			 & \le \norm*{\Shat^{-1/2}} \norm*{\Upop^T (A - \Apop) \Upop}
			\norm*{\varepsilon}                                                         \\
			 & \le C \lambda_d^{-1/2} d \sqrt{\nu_n + b_n^2} \norm*{\varepsilon} \log n \\
			 & = \Op{ \frac{ \sqrt{\nu_n + b_n^2} \sqrt{n} \log n }
				{ \lambda_d^{1/2} } }.
		\end{aligned} \end{equation*}
	The trivial bound $\log^{1/2} n = \Omega(1)$ and our growth bound in Equation~\eqref{eq:growth:46MT1} imply
	\begin{equation} \label{eq:XhatXMeps:bound3}
		\norm*{\Q \Shat^{-1/2} \Q^T \Upop^T (A - \Apop) \Upop \Upop^T
			M \varepsilon}
		= \op{ \frac{ \lambda_d }{ \sqrt{n} } }.
	\end{equation}

	For the fourth term in Equation~\eqref{eq:XhatXMeps-decomp}, recalling the definition of $R_1 = \Upop \Upop^T \Uhat - \Upop \Q = \Upop(\Upop^T \Uhat - \Q)$, observe that
	\begin{equation*}
		\Q \Shat^{1/2} R_1^T M \varepsilon
		= \Q \Shat^{1/2} (\Upop^T \Uhat - \Q)^T \Upop^T M \varepsilon.
	\end{equation*}
	Applying submultiplicativity followed by Lemma~\ref{lem:spectralnorms}, Proposition~\ref{prop:principalangles} and Lemma~\ref{lem:norm-of-errors:Heps}, and using our assumption that $B$ is constant in $n$,
	\begin{equation*} \begin{aligned}
			\norm*{\Q \Shat^{1/2} R_1^T M \varepsilon}
			 & \le \norm*{\Shat^{1/2}} \norm*{\Upop^T \Uhat - \Q}_F \norm*{\varepsilon} \\
			 & = \op{
				\frac{\lambda_1^{1/2} (\nu_n + b_n^2) n^{3/2} \log^2 n}
				{\lambda_d^2} }.
		\end{aligned} \end{equation*}
	Using the trivial upper bound $\lambda_1 \ge \lambda_d$ and our growth assumption in Equation~\eqref{eq:growth:46MT5a}, it follows that
	\begin{equation} \label{eq:XhatXMeps:bound4}
		\norm*{\Q \Shat^{1/2} R_1^T M \varepsilon}
		= \op{ \frac{ \lambda_d }{\sqrt{n} } }.
	\end{equation}

	Similarly, for the fifth term in Equation~\eqref{eq:XhatXMeps-decomp}, applying submultiplicativity followed by Equation~\eqref{eq:comp2} and Lemma~\ref{lem:norm-of-errors:Heps},
	\begin{equation} \label{eq:XhatXMeps:inter5} \begin{aligned}
			\norm*{\Q R_2^T \Upop^T M \varepsilon}
			 & \le \norm*{R_2} \norm*{\varepsilon}                        \\
			 & \le C \brac*{
				\frac{\lambda_1 (\nu_n + b_n^2) n \log^2 n }{\lambda_d^{5/2}}
				+ \frac{ (\nu_n + b_n^2)^{1/2} \log n}{\lambda_d^{1/2}}
			} \Op{ \sqrt{n} }                                             \\
			 & = \Op{ \frac{ \lambda_1 (\nu_n + b_n^2) n^{3/2} \log^2 n }
				{\lambda_d^{5/2}} }
			+
			\Op{
				\frac{ \sqrt{\nu_n + b_n^2} \sqrt{n} \log n }
				{\lambda_d^{1/2} }
			}.
		\end{aligned} \end{equation}
	Our growth assumption in Equation~\eqref{eq:growth:46MT5a} states that the first of these two rates is $\op{ \lambda_d/\sqrt{n} }$.
	Our bound in Equation~\eqref{eq:growth:46MT1} along with the trivial bound $\log^{1/2} n = \Omega(1)$ implies that the second is $\op{ \lambda_d/\sqrt{n} }$, whence
	\begin{equation} \label{eq:XhatXMeps:bound5}
		\norm*{\Q R_2^T \Upop^T M \varepsilon}
		= \op{ \frac{ \lambda_d }{ \sqrt{n} } }.
	\end{equation}

	For the sixth term in Equation~\eqref{eq:XhatXMeps-decomp}, we expand $R_3$ to write
	\begin{align}
		 & \Q \Shat^{-1/2} R_3^T (I - \Upop \Upop^T)(A - \Apop)
		M \varepsilon  \notag                                                                                              \\
		 & \qquad
		\label{eq:XhatXMeps:bound6:1}
		= \Q \Shat^{-1/2}(\Uhat - \Upop \Upop^T \Uhat)^T (I - \Upop \Upop^T)(A - \Apop) M \varepsilon                      \\
		\label{eq:XhatXMeps:bound6:2}
		 & \qquad \qquad + \Q \Shat^{-1/2}(\Upop \Upop^T \Uhat - \Upop \Q)^T (I - \Upop \Upop^T)(A - \Apop) M \varepsilon.
	\end{align}
	To bound \eqref{eq:XhatXMeps:bound6:1}, we use submultiplicativity followed by Lemma~\ref{lem:spectralnorms}, Equation~\eqref{eq:halfR1} of Lemma~\ref{lem:components}, and Lemma~\ref{lem:Aconcentrates} to see
	\begin{equation*} \begin{aligned}
			 & \norm*{\Q \Shat^{-1/2}(\Uhat - \Upop \Upop^T \Uhat)^T
			(I - \Upop \Upop^T)(A - \Apop) M \varepsilon}            \\
			 & \qquad
			\le C\norm*{\Shat^{-1/2}} \norm*{\Uhat - \Upop \Upop^T \Uhat}
			\norm*{A-\Apop} \norm*{\varepsilon}                      \\
			 & \qquad
			\le \frac{ C (\nu_n + b_n^2) n^{3/2} \log^2 n }{ \lambda_d^{3/2} }.
		\end{aligned} \end{equation*}
	Applying the trivial upper bound $\lambda_1 \ge \lambda_d$ and our growth bound in Equation~\eqref{eq:growth:46MT5a},
	\begin{equation} \label{eq:XhatXMeps:bound6:1:resolved}
		\norm*{ \Q \Shat^{-1/2}(\Uhat - \Upop \Upop^T \Uhat)^T (I - \Upop \Upop^T)(A - \Apop) M \varepsilon }
		= \op{ \frac{ \lambda_d }{ \sqrt{n} } }.
	\end{equation}

	To bound \eqref{eq:XhatXMeps:bound6:2}, we apply submultiplicativity followed by Lemma~\ref{lem:spectralnorms}, Proposition~\ref{prop:principalangles}, Lemma \ref{lem:Aconcentrates} and Lemma \ref{lem:norm-of-errors:Heps} to see
	\begin{equation*} \begin{aligned}
			 & \norm*{\Q \Shat^{-1/2}(\Upop \Upop^T \Uhat - \Upop \Q)^T (I - \Upop \Upop^T)(A - \Apop) M \varepsilon} \\
			 & \qquad \le C\norm*{\Shat^{-1/2}} \norm*{\Upop^T \Uhat - \Q} \norm*{A - \Apop} \norm*{M \varepsilon}
			\le \frac{ C (\nu_n + b_n^2)^{3/2} n^2 \log^{3} n }{ \lambda_d^{5/2} }.
		\end{aligned} \end{equation*}
	Applying Equations~\eqref{eq:growth:46MT1} and~\eqref{eq:growth:46MT5a} along with the trivial bound $\lambda_d \le \lambda_1 = \bigoh{n}$ from Equation~\eqref{eq:growth:lambda1:UB},
	\begin{equation} \label{eq:XhatXMeps:bound6:2:resolved}
		\norm*{\Q \Shat^{-1/2}(\Upop \Upop^T \Uhat - \Upop \Q)^T (I - \Upop \Upop^T)(A - \Apop) M \varepsilon}
		= \op{ \frac{  \lambda_d }{ \sqrt{n} } }.
	\end{equation}
	Applying Equations~\eqref{eq:XhatXMeps:bound6:1:resolved} and~\eqref{eq:XhatXMeps:bound6:2:resolved} to bound the respective quantities on lines~\eqref{eq:XhatXMeps:bound6:1} and~\eqref{eq:XhatXMeps:bound6:2}, we conclude that
	\begin{equation} \label{eq:XhatXMeps:bound6}
		\norm*{\Q \Shat^{-1/2} R_3^T (I - \Upop \Upop^T)(A - \Apop) M \varepsilon}
		= \op{ \frac{  \lambda_d }{ \sqrt{n} } }.
	\end{equation}
	Applying
	Equations~\eqref{eq:XhatXMeps:bound1}, \eqref{eq:XhatXMeps:bound2}, \eqref{eq:XhatXMeps:bound3}, \eqref{eq:XhatXMeps:bound4}, \eqref{eq:XhatXMeps:bound5} and~\eqref{eq:XhatXMeps:bound6} to control the terms of Equation~\eqref{eq:XhatXMeps-decomp}, we conclude that
	\begin{equation*}
		\norm*{(\Xhat \Q^T - \X)^T M \varepsilon}
		=
		\op{ \frac{ \lambda_d }{ \sqrt{n} } },
	\end{equation*}
	completing the proof.
\end{proof}

\begin{lemma} \label{lem:xmx-spectral-bound}

	Under Assumptions \ref{ass:subgamma}, \ref{ass:causal-linearity}, and \ref{ass:second-stage}
	\begin{equation*}
		\norm*{\paren*{\X^T M \X}^{-1}}
		= \Op{\lambda_d^{-1}} ~~~\text{ and }~~~
		\norm*{\paren*{\Xhat^T M \Xhat}^{-1}}
		= \Op{\lambda_d^{-1}}.
	\end{equation*}
\end{lemma}
\begin{proof}
	Recall that $M = I - \W \paren*{\W^T \W}^{-1} \W^T$. Consider the full singular value decomposition
	\begin{equation*}
		W =
		\begin{array}{@{}c@{}}{
			\begin{bmatrix}
				\Uw & \Uwtilde \\
			\end{bmatrix}} \\
			\vphantom{
				\begin{bmatrix}
					\Uw & \Uwtilde \\
				\end{bmatrix}
			}              \\
		\end{array}
		\begin{bmatrix}
			\Sw & 0        \\
			0   & \Swtilde \\
		\end{bmatrix}
		\begin{bmatrix}
			\Vw^T \\
			\Vwtilde^T
		\end{bmatrix},
	\end{equation*}
	where $\Uw, \Vw \in \bbO_p$ and $\Uwtilde, \Vwtilde \in \bbO_{n - p}$. It can be shown that $M = \Uwtilde \Uwtilde^T$.

	Recall that, for conformable real-valued matrices $A$ and $B$, $\paren*{A B}^\dagger = B^\dagger A^\dagger$ when $B = A^T$, or when either $A$ or $B$ is an orthogonal matrix \citep{greville1966}. Using this fact repeatedly, together with submultiplicativity and orthogonal invariance of the spectral norm, and Lemma \ref{lem:spectralnorms}, we obtain:
	\begin{align*}
		\norm*{\paren*{\X^T M \X}^{-1}}
		 & = \norm*{\paren*{\X^T \Uwtilde \Uwtilde^T \X}^{-1}}
		= \norm*{\paren*{\X^T \Uwtilde \Uwtilde^T \X}^\dagger}                         = \norm*{\paren*{\Uwtilde^T \X}^\dagger
		\paren*{\X^T \Uwtilde}^\dagger}                                                                 \\
		 & \le \norm*{\paren*{\Uwtilde^T \X}^\dagger} \norm*{\paren*{\X^T \Uwtilde}^\dagger}
		= \norm*{\X^\dagger \paren*{\Uwtilde^T}^\dagger} \norm*{\Uwtilde^\dagger \paren*{\X^T}^\dagger} \\
		 & \le \norm*{\X^\dagger} \norm*{\Uwtilde^\dagger} \norm*{\Uwtilde^\dagger} \norm*{\X^\dagger}
		= \norm*{\paren*{\Spop^{1/2}}^\dagger} \norm*{\paren*{\Spop^{1/2}}^\dagger}                     \\
		 & = C \lambda_d^{-1/2} \cdot \lambda_d^{-1/2} = C \lambda_d^{-1}.
	\end{align*}

	The proof for the $\Xhat$ case is analogous, and uses the fact that $\Shat$ concentrates around $S$, as characterized by Lemma \ref{lem:spectralnorms}. The fact that the inverses exist asymptotically follows from Assumption~\ref{ass:regularity}, which takes the regression coefficients $\beta$ to be identified.
\end{proof}

\begin{lemma}[Concentration of Sub-Gaussian Norms] \label{lem:W-norms}
	Let $\W \in \mathbb{R}^{n \times p}$ obey Assumption~\ref{ass:second-stage}, so that $\W$ has independent rows, with the entries of each row being possibly dependent, but each marginally sub-Gaussian with parameter $\sigma > 0$.
	Then there exists a constant $C>0$ such that with probability $1- \bigoh{n^{-2}}$,
	\begin{equation*}
		\norm*{ \W } \le C \sqrt{ p n \sigma^2 }.
	\end{equation*}
	and, also with probability $1-\bigoh{n^{-2}}$, it holds for all $j \in [p]$ that
	\begin{equation*}
		\norm*{ \W_{\cdot j}}^2 \le C n \sigma^2.
	\end{equation*}
\end{lemma}
\begin{proof}
	To prove the spectral norm bound on $\W$, we adapt the argument given in Theorem 4.6.1 in \cite{vershynin2020}, for which we must first establish the Orlicz norm of each row $\W_{i \cdot}$ \citep[][Definition 3.4.1]{vershynin2020}. We begin by noting that for any unit vector $u \in \mathbb{R}^p$, any integer $q \ge 1$ and $i \in [n]$,
	\begin{equation*}
		\E{\left( u^T \W_{i \cdot} \right)^{2q}}
		\le \E{p^{2q-1} \sum_{j=1}^p \left(  u_j \W_{i j} \right)^{2q}} \\
		= p^{2q-1} \sum_{j=1}^p u_j^{2q} \E{\W_{i j}^{2q}},
	\end{equation*}
	where the inequality follows from the convexity of $x \mapsto x^{2q}$. Since $\W_{ij}$ is sub-Gaussian with parameter $\sigma$, using basic properties of sub-Gaussian random variables \citep[see, e.g.,][Theorem 2.1]{boucheron2013},
	\begin{equation*}
		\E{\W_{ij}^{2q}} \le q!(4\sigma)^{2q},
	\end{equation*}
	and it follows that, trivially upper bounding $q! \le q^q$ and using $\|u\|=1$,
	\begin{equation*}
		\E{\left( u^T \W_{i \cdot} \right)^{2q}}
		\le q! (4p\sigma)^{2q} \frac{1}{p} \sum_{j=1}^p u_j^{2q}
		\le q^q (4p\sigma)^{2q}.
	\end{equation*}
	Thus, we find that for any unit $u \in \mathbb{R}^p$, the random variable $u^T \W_{i \cdot}$ satisfies
	\begin{equation} \label{eq:orlicz:q}
		\left( \E{\left( u^T \W_{i \cdot} \right)^{2q}} \right)^{1/2q}
		\le \sqrt{2q} \left( 2\sqrt{2} p \sigma^2 \right).
	\end{equation}
	That is, the random variable $u^T \W_{i \cdot}$ has Orlicz norm \citep[][Proposition 2.5.2]{vershynin2020}
	\begin{equation*}
		\| u^T \W_{i \cdot} \|_{\Psi_2} \le C p \sigma^2.
	\end{equation*}
	Taking the supremum over all unit $u \in \mathbb{R}^p$, the random vector $\W_{i \cdot}$ has Orlicz norm \citep[see, e.g.,][Definition 3.4.1]{vershynin2020}
	\begin{equation} \label{eq:W:orlicz}
		\| \W_{i \cdot} \|_{\Psi_2} \le  C p \sigma^2.
	\end{equation}

	Following the argument of Theorem 4.6.1 in \cite{vershynin2020}, let $\mathcal{N}$ be a $(1/4)$-net for the unit sphere in $\mathbb{R}^p$, which can be constructed with cardinality at most $9^p$. It follows that
	\begin{equation*}
		\left\| \frac{1}{n} \W^T \W \right\|
		\le 2 \max_{u \in \mathcal{N}}
		\left| \frac{1}{n} u^T \W^T \W u \right|
		= 2 \max_{u \in \mathcal{N}} \frac{1}{n} \| \W u \|^2.
	\end{equation*}
	Fixing $u \in \mathcal{N}$, note that
	\begin{equation*}
		\frac{1}{n} \| \W u \|^2
		= \frac{1}{n} \sum_{i=1}^n \left( u^T \W_{i \cdot} \right)^2.
	\end{equation*}
	Since the random vector $\W_{i \cdot}$ has Orlicz norm as given in Equation~\eqref{eq:W:orlicz}, $u^T \W_{i \cdot}$ is subgaussian with parameter $C p \sigma^2$ and it follows that
	\begin{equation*}
		\frac{1}{n} \left( \| \W u \|^2 - \E{\| \W u \|^2} \right)
		= \frac{1}{n}
		\sum_{i=1}^n \left[ ( u^T \W_{i \cdot})^2 - \E{( u^T \W_{i \cdot})^2} \right]
	\end{equation*}
	is the sample mean of $n$ independent sub-exponential random variables, each with parameter $C p \sigma^2$ (adjusting $C$ by a suitable constant multiple). Define $\delta = C(\sqrt{p} + t)/\sqrt{n}$ for $t \ge 0$ and set $\epsilon = C p \sigma^2 \max\{ \delta, \delta^2 \}$. Applying Bernstein's inequality \citep[][Corollary 2.8.3]{vershynin2020}, an argument essentially identical to that in Step 2 of Theorem 4.6.1 in \cite{vershynin2020}, yields that
	\begin{equation*}
		\P{\left| \frac{1}{n} \left( \| \W u \|^2 - \E{\| \W u \|^2}
			\right) \right|
			\ge \frac{ \epsilon }{ 2 }}
		\le 2 \exp\left\{ -C (p + t^2) \right\}.
	\end{equation*}
	Setting $t = C \log^{1/2} n$ for suitably large $C > 0$ and taking a union bound over all at most $9^p$ vectors $u \in \mathcal{N}$, it follows that with probability at least $1 - \bigoh{n^{-2}}$, it holds for all $u \in \mathcal{N}$ that
	\begin{equation*}
		\left| \frac{1}{n} \left( \| \W u \|^2 - \E{\| \W u \|^2}
		\right) \right|
		\le \frac{ C \sigma^2 p ( \sqrt{p} + \log^{1/2}n )^2 }{ \sqrt{n} }
		\le \frac{ C p \sigma^2( p + \log n ) }{ \sqrt{n} }.
	\end{equation*}
	Thus, it follows that with probability at least $1-\bigoh{n^{-2}}$, for all $u \in \mathcal{N}$
	\begin{equation} \label{eq:probbound:inter}
		\| \W u \|^2 \le \E{\| \W u \|^2} + C p n^{1/2} \sigma^2( p + \log n ).
	\end{equation}
	Expanding $\W u$ and setting $q=1$ in Equation~\eqref{eq:orlicz:q},
	\begin{equation*}
		\E{\| \W u \|^2}
		= \sum_{i=1}^n \E{(u^T \W_{i \cdot})^2}
		\le C p n \sigma^2.
	\end{equation*}
	Applying this bound to Equation~\eqref{eq:probbound:inter},
	with probability at least $1-\bigoh{n^{-2}}$, it holds for all $u \in \mathcal{N}$ that
	\begin{equation*}
		\| \W u \|^2 \le C p \sigma^2 ( n + n^{1/2} \log n ).
	\end{equation*}
	Thus, with probability at least $1-\bigoh{n^{-2}}$,
	\begin{equation*}
		\| \W \|^2 \le 2\max_{u \in \mathcal{N}} \| \W u \|^2
		\le C p \sigma^2 n.
	\end{equation*}
	Taking square roots yields our desired bound on the spectral norm. To prove the column-wise bound on $\W$, observe that for any $j \in [p]$, we use a straight-forward adaptation of the proof of Theorem 3.1.1 in \cite{vershynin2020}.
	We observe that
	\begin{equation*}
		\frac{1}{n} \norm*{ \W_{\cdot j}}^2 - \frac{1}{n}\sum_{i=1}^n \E{\W_{ij}^2}
		= \frac{1}{n} \sum_{i=1}^n\left( \W_{i j}^2 - \E{\W_{i j}^2} \right)
	\end{equation*}
	is the sample mean of independent mean-zero random variables, each of which is sub-exponential with parameter $C \sigma$ for suitably chosen constant $C > 0$. An application of Bernstein's inequality \citep{boucheron2013} then yields that with probability $1-\bigoh{n^{-2}}$,
	\begin{equation*}
		\left| \frac{1}{n} \sum_{i=1}^n\left( \W_{i j}^2 - \E{\W_{i j}^2}
		\right) \right|
		\le \frac{ C \sqrt{ \sigma^2 \log n } }{ n^{1/2} }.
	\end{equation*}
	Thus, with probability at least $1-\bigoh{n^{-2}}$,
	\begin{equation*}
		\| \W_{\cdot j} \|^2
		= \E{\| \W_{\cdot j} \|^2} + \frac{ C  \sqrt{ \sigma^2 \log n }}{ n^{1/2} }
		\le C n \sigma^2
	\end{equation*}
	for $C > 0$ chosen suitably large.
\end{proof}

\begin{lemma} \label{lem:UAPW:frob}
	Under Assumptions \ref{ass:subgamma} and \ref{ass:second-stage}, with notation as above, it holds with probability at least $1- \bigoh{n^{-2}}$ that
	\begin{equation*}
		\norm*{\Upop^T (A - \Apop) \W}_F \le C \sqrt{d p (\nu_n + b_n^2) n \log n} .
	\end{equation*}
\end{lemma}
\begin{proof}
	We will show that $\norm*{\Upop^T (A - \Apop) \W}_F^2 \ge C d p (\nu_n + b_n^2) n \log n$ with probability no larger than $\mathcal O (n^{-2})$, whence taking square roots will yield the result.

	For each $k \in [d], \ell \in [p + 2]$, define
	\begin{equation*}
		S_{k, \ell}
		= \brac*{\Upop^T (A - \Apop) \W}_{k, \ell}
		= \sum_{i=1}^n \sum_{j=1}^n (A - \Apop)_{ij} \Upop_{ik} \W_{j \ell}
	\end{equation*}
	and note that
	\begin{equation*}
		\norm*{\Upop^T(A - \Apop) \W}_F^2 = \sum_{k=1}^d \sum_{\ell=1}^{p + 2} S_{k, \ell}^2.
	\end{equation*}

	By Corollary 2.11 in \cite{boucheron2013}, for any $t > 0$,

	\begin{equation*}
		\P[\X, \W]{\abs*{S_{k, \ell}} \ge t}
		\le 2 \exp \set*{\frac{-t^2}
			{2\paren*{\nu_n \sum_{i=1}^n \sum_{j=1}^n \Upop_{ik}^2 \W_{j \ell}^2 + b_n t}}}.
	\end{equation*}

	Let $G_n$ denote the event $\set*{\norm*{\W_{\cdot \ell}}^2 \le C_W \, n}$ for some constant $C_W$, and $G_n^c$ denote the complement of $G_n$. By Lemma~\ref{lem:W-norms}, $G_n$ occurs with probability at least $1 - \bigoh{n^{-2}}$. Thus
	\begin{equation*} \begin{aligned}
			\P{\abs*{S_{k, \ell}} \ge t}
			 & = \P[G_n]{\abs*{S_{k, \ell}} \ge t} \cdot \P{G_n}
			+ \P[G_n^c]{\abs*{S_{k, \ell}} \ge t} \cdot \P{G_n^c}      \\
			 & \le \P[G_n]{\abs*{S_{k, \ell}} \ge t} + \P{G_n^c}       \\
			 & \le \P[G_n]{\abs*{S_{k, \ell}} \ge t} + \bigoh{n^{-2}}.
		\end{aligned} \end{equation*}

	Now observe that $\sum_{j=1}^n \W_{j \ell}^2$ is the squared $\ell_2$ norm of a column of $\W$. By the definition of $G_n$,
	\begin{equation*}
		\P[G_n]{\abs*{S_{k, \ell}} \ge t}
		\le 2 \exp \set*{\frac{-t^2}{2(C_W n \nu_n + bt)}}.
	\end{equation*}

	Thus

	\begin{align*}
		\P{\abs*{S_{k, \ell}} \ge t}
		 & \le 2 \exp \set*{\frac{-t^2}{2(C_W n \nu_n + bt)}} + \bigoh{n^{-2}}
	\end{align*}

	Taking $t = C(\nu_n + b_n^2)^{1/2} \sqrt{n \log n}$ for $C > 0$ suitably large, it follows that

	\begin{equation*}
		\P{\abs*{S_{k, \ell}} \ge C(\nu_n + b_n^2)^{1/2} \sqrt{n \log n}} \le 2n^{-4} + \bigoh{n^{-2}}.
	\end{equation*}

	A union bound over all $k \in [d], \ell \in [p + 2]$ implies that

	\begin{equation*}
		\P{ \exists k \in [d], \ell \in [p] : \abs*{S_{k, \ell}} \ge C(\nu_n + b_n^2)^{1/2} \sqrt{n \log n}} \le \frac{2d(p + 2)}{n^4} + \bigoh{\frac{d(p + 2)}{n^2}} = \bigoh{n^{-2}},
	\end{equation*}
	as $d$ and $p$ are fixed as a function of $n$. It follows that
	\begin{equation*}
		\P{ \norm*{\Upop^T (A - \Apop) \W}_F^2 \ge C d p (\nu_n + b_n^2) n \log^2 n}
		\le 2 n^{-2},
	\end{equation*}
	completing the proof.
\end{proof}

\begin{lemma} \label{lem:XhatXW}
	Suppose that Assumptions~\ref{ass:subgamma},~\ref{ass:second-stage} and~\ref{ass:growth-rates} hold.
	Then
	\begin{equation*}
		\norm*{(\Xhat \Q^T - \X)^T \W} = \op{ \sqrt{n} }.
	\end{equation*}
\end{lemma}
\begin{proof}
	Applying Lemma~\ref{lem:decomposition},
	\begin{equation} \label{eq:XhatXW-decomp} \begin{aligned}
			 & (\Xhat \Q^T - \X)^T \W                                                    \\
			 & \qquad = \Q \paren*{\Uhat \Shat^{1/2} - \Upop \Spop^{1/2} \Q}^T \W        \\
			 & \qquad = \Q \Spop^{-1/2} \Upop^T (A - \Apop) \W
			+ \Q (\Q \Shat^{-1/2} - \Spop^{-1/2} \Q)^T \Upop^T (A - \Apop) \W            \\
			 & \qquad \qquad + \Q \Shat^{-1/2} \Q^T \Upop^T (A - \Apop) \Upop \Upop^T \W
			+ \Q \Shat^{1/2} R_1^T \W + \Q R_2^T \Upop^T \W                              \\
			 & \qquad \qquad + \Q \Shat^{-1/2} R_3^T (I - \Upop \Upop^T)(A - \Apop) \W.
		\end{aligned} \end{equation}
	We observe that since $\lambda_d/\sqrt{n} \le \lambda_1\sqrt{n}$ trivially, our growth assumption in Equation~\eqref{eq:growth:lambda1:UB} implies that it is sufficient to show that the above terms are all bounded as $o( \lambda_d/\sqrt{n} )$.
	This is precisely the bound obtained in Lemma~\ref{lem:XhatXMeps}, and the proof of this lemma proceeds identically,
	but using Lemma~\ref{lem:UAPW:frob} instead of Lemma~\ref{lem:UAPMeps:frob}
	and using the bound in Lemma~\ref{lem:W-norms} instead of directly bounding $\| S^{1/2} \|$.
	The remaining details of the proof are omitted.
\end{proof}

\subsection{Proof of Lemma \ref{lem:covariance}} \label{sec:proof:lemcovar}

The following fact will be useful in the subsequent proofs.

\begin{proposition}
	\label{prop:outer-triangle}
	For $u,v \in \R^k$, $\norm*{u u^T - v v^T} \le 2 \norm*{u - v} \norm*{v} + \norm*{u - v}^2$.
\end{proposition}

\begin{proof}
	Follows from adding and subtracting appropriate quantities and repeatedly applying the triangle inequality.
\end{proof}

\begin{proposition} \label{prop:max-w-norm}
	Under Assumption \ref{ass:second-stage},
	\begin{equation*}
		\max_{i \in [n]} \, \norm*{\W_{i \cdot}} = \Op{\sqrt{\log n}}.
	\end{equation*}
\end{proposition}
\begin{proof}
	Since each $\W_{ij}$ is sub-Gaussian with variance parameter $\sigma^2$, by Theorem 2.1 of \cite{boucheron2013},
	\begin{equation*}
		\P{\abs*{\W_{ij}} \ge t} \le 2 \exp \paren*{-t^2 / 2 \sigma^2}.
	\end{equation*}
	By a union bound,
	\begin{equation*}
		\P{\max_{ij} \, \abs*{\W_{ij}} \ge t}
		\le 2 n (p + 2) \exp \paren*{-t^2 / 2 \sigma^2}.
	\end{equation*}
	Taking $t = \sqrt{C \sigma^2 \log n}$ and $C > 0$ sufficiently large,
	\begin{equation*}
		\P{\max_{i \in [n], j \in [p + 2]} \, \abs*{\W_{ij}} \ge \sqrt{C \sigma^2 \log n}}
		\le 2 n (p + 2) \exp \paren*{-C \sigma^2 \log n / 2 \sigma^2}
		= \frac{(4 p + 4)}{n^2}.
	\end{equation*}
	Observing that $\max_{i \in [n]} \, \norm*{\W_{i \cdot}} \le \sqrt{p + 2} \, \max_{ij} \W_{ij}$, we obtain the desired result.
\end{proof}

\begin{proposition} \label{prop:SigmabetaD}
	Suppose Assumptions \ref{ass:subgamma}, \ref{ass:regularity}, \ref{ass:second-stage} and \ref{ass:growth-rates} hold.
	Letting $\Q \in \bbO_d$ be as in Lemma~\ref{lem:subgamma-2toinfty}, define
	\begin{equation*}
		\Qtilde =
		\begin{bmatrix}
			I_{p+2} & 0  \\
			0       & \Q
		\end{bmatrix}.
	\end{equation*}
	Then $\Sigmahatbeta \to \Qtilde \Sigmatildebeta \Qtilde^T$ in probability, where $\Sigmahatbeta$ and $\Sigmatildebeta$ are as defined in Definition \ref{def:covariance}.
\end{proposition}
\begin{proof}
	By definition,
	\begin{align*}
		\norm*{\Sigmahatbeta - \Qtilde \Sigmatildebeta \Qtilde^T}
		= \norm*{
			\Abetahat^{-1} \cdot \Bbetahat \cdot \paren*{\Abetahat^{-1}}^T
			- \Qtilde^T \Abetatilde^{-1} \Qtilde \cdot \Qtilde^T \Bbetatilde \Qtilde \cdot \Qtilde^T \paren*{\Abetatilde^{-1}}^T \Qtilde
		}.
	\end{align*}
	By the continuous mapping theorem, it is sufficient to show $\Abetahat \to \Qtilde^T \Abetatilde \Qtilde$ and $\Bbetahat \to \Qtilde^T \Bbetatilde \Qtilde$, with both convergences holding in probability.

	We begin by observing that
	\begin{equation*}
		\Abetahat - \Qtilde^T \Abetatilde \Qtilde
		= \frac{1}{n}
		\begin{bmatrix}
			0                         & \W^T \paren{\Xhat - X \Q}     \\
			\paren{\Xhat - X \Q}^T \W & \Xhat^T \Xhat - \Q^T X^T X \Q
		\end{bmatrix},
	\end{equation*}
	and so, applying Lemma \ref{lem:XhatXW} to bound $\norm{\paren{\Xhat - X \Q}^T \W} = \op{\sqrt n}$,
	\begin{equation} \label{eq:AhatAtilde:inter} \begin{aligned}
			\norm*{\Abetahat - \Qtilde^T \Abetatilde \Qtilde}
			 & \le \frac{ 2 }{n } \norm*{\paren{\Xhat - X \Q}^T \W}
			+ \frac{1}{n} \norm*{\Xhat^T \Xhat - \Q^T X^T X \Q}     \\
			 & = \frac{1}{n} \norm*{\Xhat^T \Xhat - \Q^T X^T X \Q}
			+ \op{ n^{-1/2} }.
		\end{aligned} \end{equation}

	Applying Lemma~\ref{lem:XMXXhatMXhat} with $M=I$,
	\begin{equation*}
		\frac{1}{n} \norm*{\Xhat^T \Xhat - \Q^T X^T X \Q}
		= \op{ \frac{ \lambda_d^2 }{ n^{3/2} \lambda_1^{1/2} } }.
	\end{equation*}
	Using the fact that $\lambda_d \le \lambda_1$ by definition and applying Equation~\eqref{eq:growth:lambda1:UB},
	\begin{equation*}
		\frac{1}{n} \norm*{\Xhat^T \Xhat - \Q^T X^T X \Q} = \op{ 1 }.
	\end{equation*}

	Applying this to Equation~\eqref{eq:AhatAtilde:inter},
	\begin{equation*}
		\norm*{\Abetahat - \Qtilde^T \Abetatilde \Qtilde}
		= \op{ 1 } + \op{ n^{-1/2} }
		= \op{ 1 }.
	\end{equation*}
	The continuous mapping theorem then implies that
	\begin{equation*}
		\Abetahat^{-1} \rightarrow \Qtilde^T \Abetatilde^{-1} \Qtilde~~~\text{ in probability.}
	\end{equation*}

	It remains to show that $\Bbetahat$ converges to $\Qtilde^T \Bbetatilde \Qtilde$ in probability. Toward this end, recall the definitions $\varepsilonhat_i = Y_i - \Dhat_{i \cdot} \betahat$ and $\varepsilontilde_i = Y_i - \D_{i \cdot} \betatilde$. Adding and subtracting appropriate quantities,
	\begin{equation} \label{eq:Bbetahat:decomp} \begin{aligned}
			\norm*{\Bbetahat - \Qtilde^T \Bbetatilde \Qtilde}
			 & = \norm*{
				\frac{1}{n} \sum_{i=1}^n \varepsilonhat_i^2 \Dhat_{i \cdot}^T \Dhat_{i \cdot}
				- \varepsilontilde_i^2 \Qtilde^T \D_{i \cdot}^T \D_{i \cdot} \Qtilde
			}            \\
			 & = \norm*{
				\frac{1}{n} \sum_{i=1}^n
				\varepsilonhat_i^2 \brac*{
					\Dhat_{i \cdot}^T \Dhat_{i \cdot} - \Qtilde^T \D_{i \cdot}^T \D_{i \cdot} \Qtilde
				}
				+ \brac*{
					\varepsilonhat_i^2 - \varepsilontilde_i^2
				} \Qtilde^T \D_{i \cdot}^T \D_{i \cdot} \Qtilde
			}            \\
			 & \le
			\frac{1}{n} \sum_{i=1}^n
			\varepsilonhat_i^2 \norm*{
				\Dhat_{i \cdot}^T \Dhat_{i \cdot} - \Qtilde^T \D_{i \cdot}^T \D_{i \cdot} \Qtilde
			}
			+
			\abs*{\varepsilonhat_i^2 - \varepsilontilde_i^2}
			\norm*{\D_{i \cdot}^T \D_{i \cdot}}.
		\end{aligned} \end{equation}

	Adding and subtracting appropriate quantities and applying the triangle inequality,
	\begin{equation*} \begin{aligned}
			\max_{i \in [n]} \, \abs*{\varepsilonhat_i - \varepsilontilde_i}
			 & = \max_{i \in [n]} \, \abs*{\Dhat_{i \cdot} \Qtilde^T \Qtilde \betahat - \D_{i \cdot} \betatilde} \\
			 & \le \max_{i \in [n]} \, \norm*{\Dhat_{i \cdot} \Qtilde^T} \norm*{\Qtilde \betahat - \betatilde} +
			\norm*{\Dhat_{i \cdot} \Qtilde^T - \D_{i \cdot}} \norm*{\betatilde}.
		\end{aligned} \end{equation*}
	By definition, $\norm*{\Dhat_{i \cdot} - \D_{i \cdot} \Qtilde} = \norm*{\Xhat_{i \cdot} - \X_{i \cdot} \Q}$.
	Thus, applying Lemma~\ref{lem:subgamma-2toinfty:littleoh1} and using Lemma \ref{lem:spectralnorms} to ensure that $\norm*{\Dhat_{i \cdot}} = \Op{\lambda_1^{1/2}}$,
	\begin{equation*}
		\max_{i \in [n]} \, \abs*{\varepsilonhat_i - \varepsilontilde_i}
		\le \norm*{\Qtilde \betahat - \betatilde} \Op{ \lambda_1^{1/2} }
		+ \eta_n \norm*{ \betatilde }.
	\end{equation*}
	Applying Assumption~\ref{ass:growth-rates} and Lemma~\ref{lem:op1},
	\begin{equation*}
		\max_{i \in [n]} \, \abs*{\varepsilonhat_i - \varepsilontilde_i}
		\le \eta_n + \op{ 1 }.
	\end{equation*}
	Applying the bound on $\eta_n$ from Lemma~\ref{lem:subgamma-2toinfty:littleoh1},
	\begin{equation*}
		\max_{i \in [n]} \, \abs*{\varepsilonhat_i - \varepsilontilde_i} = \op{ 1 }.
	\end{equation*}

	This in turn yields, by Proposition \ref{prop:outer-triangle},
	\begin{equation} \label{eq:eps2:diff}
		\abs*{\varepsilonhat_i^2 - \varepsilontilde_i^2}
		\le \abs*{\varepsilonhat_i - \varepsilontilde_i} \abs*{\varepsilontilde_i}
		+ \abs*{\varepsilonhat_i - \varepsilontilde_i}^2
		= (1+ \abs*{\varepsilontilde_i}) \op{1}.
	\end{equation}

	By Proposition \ref{prop:outer-triangle} and Lemma~\ref{lem:subgamma-2toinfty},
	\begin{equation} \label{eq:DhatD:diff}
		\norm*{
			\Dhat_{i \cdot}^T \Dhat_{i \cdot} - \Qtilde^T \D_{i \cdot}^T \D_{i \cdot} \Qtilde
		}
		\le \norm*{\Dhat_{i \cdot} - \D_{i \cdot} \Qtilde} \norm*{\D_{i \cdot} \Qtilde} + \norm*{\Dhat_{i \cdot} - \D_{i \cdot} \Qtilde}^2
		= \eta_n \norm*{\D_{i \cdot}} + \eta_n^2.
	\end{equation}
	Adding and subtracting appropriate quantities,
	\begin{equation} \label{eq:varepsilonhat:decomp}
		\begin{aligned}
			\frac{1}{n} \sum_{i=1}^n \varepsilonhat_i^2 \norm*{
				\Dhat_{i \cdot}^T \Dhat_{i \cdot} - \Qtilde^T \D_{i \cdot}^T \D_{i \cdot} \Qtilde
			}
			 & \le \frac{1}{n} \sum_{i=1}^n \varepsilontilde_i^2 \norm*{
				\Dhat_{i \cdot}^T \Dhat_{i \cdot} - \Qtilde^T \D_{i \cdot}^T \D_{i \cdot} \Qtilde
			}                                                            \\
			 & ~~~~~~+ \frac{1}{n} \sum_{i=1}^n
			\brac*{\varepsilonhat_i^2 - \varepsilontilde_i^2} \norm*{
				\Dhat_{i \cdot}^T \Dhat_{i \cdot} - \Qtilde^T \D_{i \cdot}^T \D_{i \cdot} \Qtilde
			}
		\end{aligned} \end{equation}

	By Equation~\eqref{eq:DhatD:diff},
	\begin{equation*}
		\frac{1}{n} \sum_{i=1}^n \varepsilontilde_i^2 \norm*{
			\Dhat_{i \cdot}^T \Dhat_{i \cdot} - \Qtilde^T \D_{i \cdot}^T \D_{i \cdot} \Qtilde
		}
		\le \frac{\eta_n}{n} \sum_{i=1}^n \varepsilontilde_i^2 \norm*{\D_{i \cdot}} + \frac{\eta_n^2}{n} \sum_{i=1}^n \varepsilontilde_i^2
	\end{equation*}
	By the regularity conditions in Assumptions~\ref{ass:regularity} and~\ref{ass:second-stage}, both averages on the right-hand side converge to constants.
	Thus, since $\eta_n = o(1)$ by Lemma~\ref{lem:subgamma-2toinfty:littleoh1},
	\begin{equation} \label{eq:epsDhat:inter1}
		\frac{1}{n} \sum_{i=1}^n \varepsilontilde_i^2 \norm*{
			\Dhat_{i \cdot}^T \Dhat_{i \cdot} - \Qtilde^T \D_{i \cdot}^T \D_{i \cdot} \Qtilde
		}
		\le C \, \eta_n = \op{1}.
	\end{equation}

	Appealing to Equations~\eqref{eq:eps2:diff} and~\eqref{eq:DhatD:diff},
	\begin{equation*} \begin{aligned}
			\frac{1}{n} \sum_{i=1}^n
			\brac*{\varepsilonhat_i^2 - \varepsilontilde_i^2} \norm*{
				\Dhat_{i \cdot}^T \Dhat_{i \cdot} - \Qtilde^T \D_{i \cdot}^T \D_{i \cdot} \Qtilde
			}
			 & = \frac{ \op{1} }{n} \sum_{i=1}^n
			(1+ \abs*{\varepsilontilde_i}) \eta_n \left(\norm*{\D_{i \cdot}} +  \eta_n \right) \\
			 & = \frac{ \op{ \eta_n } }{n} \sum_{i=1}^n
			(1+ \abs*{\varepsilontilde_i}) \left(\norm*{\D_{i \cdot}} +  \eta_n \right).
		\end{aligned} \end{equation*}
	By Assumptions~\ref{ass:regularity} and~\ref{ass:second-stage}, the mean on the right-hand side converges to a constant, and again recalling that $\eta_n = o(1)$, we have

	\begin{equation*}
		\frac{1}{n} \sum_{i=1}^n
		\brac*{\varepsilonhat_i^2 - \varepsilontilde_i^2} \norm*{
			\Dhat_{i \cdot}^T \Dhat_{i \cdot} - \Qtilde^T \D_{i \cdot}^T \D_{i \cdot} \Qtilde
		}
		= \op{1}.
	\end{equation*}
	Applying this and Equation~\eqref{eq:epsDhat:inter1} to Equation~\eqref{eq:varepsilonhat:decomp},
	\begin{equation} \label{eq:varepsilonhat:bound}
		\frac{1}{n} \sum_{i=1}^n \varepsilonhat_i^2 \norm*{
			\Dhat_{i \cdot}^T \Dhat_{i \cdot} - \Qtilde^T \D_{i \cdot}^T \D_{i \cdot} \Qtilde
		}
		= \op{1}.
	\end{equation}

	By Equation~\eqref{eq:eps2:diff} and Proposition \ref{prop:outer-triangle},
	\begin{equation*}
		\frac{1}{n} \sum_{i=1}^n
		\left| \varepsilonhat_i^2 - \varepsilontilde_i^2 \right|
		\norm*{\D_{i \cdot}^T \D_{i \cdot}}
		\le \frac{ \op{1} }{n}  \sum_{i=1}^n (1+ \abs*{\varepsilontilde_i})
		\norm*{\D_{i \cdot}}^2.
	\end{equation*}
	Assumptions~\ref{ass:regularity} and~\ref{ass:second-stage} ensure that the mean on the right-hand side converges to a constant, and we have
	\begin{equation*}
		\frac{1}{n} \sum_{i=1}^n
		\left| \varepsilonhat_i^2 - \varepsilontilde_i^2 \right|
		\norm*{\D_{i \cdot}^T \D_{i \cdot}}
		= \op{ 1 }.
	\end{equation*}
	Applying this and Equation~\eqref{eq:varepsilonhat:bound} to Equation~\eqref{eq:Bbetahat:decomp}, we conclude that
	\begin{equation*}
		\norm*{\Bbetahat - \Qtilde^T \Bbetatilde \Qtilde} = \op{1},
	\end{equation*}
	completing the proof.
\end{proof}

\begin{remark}
	Proposition~\ref{prop:SigmathetaX} below states that the robust covariance estimator for $\Thetahat$ based on $\Xhat$ converges to the robust covariance estimator based on $\X$, but subject to some orthogonal non-identifiability. The orthogonal non-identifiability has a somewhat nasty form because we have vectorized $\Thetahat$ in for the stake of $M$-estimation. To understand the result more intuitively, suppose that $d = 1$. Then Theorem \ref{thm:main} states
	\begin{equation*}
		\sqrt{ n } \,
		\Sigmahattheta^{-1/2}
		\begin{pmatrix}
			\Thetahat \, \Q^T - \Theta
		\end{pmatrix}
		\to
		\Normal{0}{I_{p}},
	\end{equation*}
	such that $\Thetahat \in \R^{d \times 1}$. Proposition \ref{prop:SigmathetaX} then states that $\Sigmahattheta \to \Qtilde \Sigmatildetheta \Qtilde^T$, analogously to Proposition \ref{prop:SigmabetaD}.
\end{remark}

\begin{proposition} \label{prop:SigmathetaX}
	Suppose Assumptions \ref{ass:subgamma}, \ref{ass:regularity}, \ref{ass:second-stage} and \ref{ass:growth-rates} hold. Let
	\begin{equation*} \begin{aligned}
			\SigmatildethetaQ
			 & = \Athetatilde^{-1} \cdot \BthetaQtilde \cdot \paren*{\Athetatilde^{-1}}^T,
		\end{aligned} \end{equation*}
	where $\Athetatilde$ is as defined in Definition \ref{def:covariance}, and
	\begin{equation*} \begin{aligned}
			\BthetaQtilde
			 & = \frac{1}{n} \sum_{i=1}^n \Q^T \xitilde_{i \cdot}^T \, \xitilde_{i \cdot} \Q \otimes \W_{i \cdot}^T \W_{i \cdot}.
		\end{aligned} \end{equation*}

	Under Assumptions \ref{ass:subgamma}, \ref{ass:regularity}, and \ref{ass:second-stage}, and Definition \ref{def:covariance}, $\Sigmahattheta \to \SigmatildethetaQ$ in probability.
\end{proposition}

\begin{proof}
	By definition, $\Athetahat = \Athetatilde$. Thus, by the continuous mapping theorem, it will suffice to show that $\Bthetahat \to \Bthetatilde$.

	Let $\xitilde_{i \cdot} = \X_{i \cdot} - \W_{i \cdot} \Thetatilde$. Then, by the triangle inequality and properties of the Kronecker product,
	\begin{equation*} \begin{aligned}
			\norm*{\Bthetahat - \Bthetatilde}
			 & = \norm*{ \frac{1}{n} \sum_{i=1}^n
				\paren*{\xihat_{i \cdot}^T \xihat_{i \cdot} - \Q^T \xitilde_{i \cdot}^T \xitilde_{i \cdot} \Q} \otimes \W_{i \cdot}^T \W_{i \cdot}
			}                                                                                                                                    \\
			 & \le \frac{1}{n} \sum_{i=1}^n \norm*{
			\paren*{\xihat_{i \cdot}^T \xihat_{i \cdot} - \Q^T \xitilde_{i \cdot}^T \xitilde_{i \cdot} \Q} \otimes \W_{i \cdot}^T \W_{i \cdot} } \\
			 & \le \frac{1}{n} \sum_{i=1}^n
			\norm*{\xihat_{i \cdot}^T \xihat_{i \cdot} - \Q^T \xitilde_{i \cdot}^T \xitilde_{i \cdot} \Q}
			\norm*{\W_{i \cdot}^T \W_{i \cdot}}.
		\end{aligned} \end{equation*}
	Applying Cauchy-Schwarz, we obtain
	\begin{equation} \label{eq:CS}
		\norm*{\Bthetahat - \Bthetatilde}
		\le \sqrt{ \frac{1}{n} \sum_{i=1}^n
			\norm*{\xihat_{i \cdot}^T \xihat_{i \cdot} - \Q^T \xitilde_{i \cdot}^T
				\xitilde_{i \cdot} \Q}^2 }
		\sqrt{ \frac{1}{n} \sum_{i=1}^n \norm*{\W_{i \cdot}^T \W_{i \cdot}}^2 }.
	\end{equation}

	We will show that this product is $\op{1}$, which will complete the proof. Using basic properties of the norm,
	\begin{equation*}
		\E{\norm*{\W_{i \cdot}^T \W_{i \cdot}}^2}
		\le \E{\norm*{ \W_{i \cdot} }^4}
		= \E{\left( \sum_{j=1}^{p + 2} \W_{ij}^2 \right)^2}
		\le C \sum_{j=1}^{p + 2} \E{\W_{ij}^4}
	\end{equation*}
	Our sub-Gaussian assumption on the entries of $\W_{i \cdot}$ imply that $\E{\W_{ij}^4} \le C \sigma^4$ \citep[][Theorem 2.1]{boucheron2013}, whence
	\begin{equation*}
		\E{\norm*{\W_{i \cdot}^T \W_{i \cdot}}^2}
		\le C (p + 2) \sigma^4 = \bigoh{1}.
	\end{equation*}
	It follows that, by the law of large numbers,
	\begin{equation*}
		\frac{1}{n} \sum_{i=1}^n \norm*{\W_{i \cdot}^T \W_{i \cdot}}^2
		= \Op{1}.
	\end{equation*}
	Applying this fact to Equation~\eqref{eq:CS}, our proof will be complete if we can show that
	\begin{equation} \label{eq:xihat:goal}
		\frac{1}{n} \sum_{i=1}^n \norm*{\xihat_{i \cdot}^T \xihat_{i \cdot} - \Q^T \xitilde_{i \cdot}^T \xitilde_{i \cdot} \Q}^2
		= \op{ 1 }.
	\end{equation}
	Recalling the definition of $\xihat_{ij}$ from Definition~\ref{def:covariance}, we observe that
	\begin{equation*} \begin{aligned}
			\xihat_{i \cdot} - \xitilde_{i \cdot} \Q
			= \Xhat_{i \cdot} - \W_{i \cdot} \Thetahat
			- \X_{i \cdot} \Q + \W_{i \cdot} \Thetatilde \Q
			= \Xhat_{i \cdot} - \X_{i \cdot} \Q
			+ \W_{i \cdot} \paren*{\Thetatilde \Q - \Thetahat}.
		\end{aligned} \end{equation*}

	Using Lemma~\ref{lem:subgamma-2toinfty} to bound $\max_i \norm*{\Xhat_{i \cdot} - \X_{i \cdot} \Q}$, Lemma~\ref{lem:op1} to bound $\norm*{\Thetatilde \Q - \Thetahat}$ and Proposition \ref{prop:max-w-norm} to bound $\norm*{\W_{i \cdot}}$, it follows that
	\begin{equation*}
		\max_{i \in [n]} \, \norm*{\xihat_{i \cdot} - \xitilde_{i \cdot} \Q}
		\le \eta_n + \max_{i \in [n]} \, \norm*{\W_{i \cdot}} \, \op{n^{-1/2}}
		= \eta_n + \op{ n^{-1/2} \log^{1/2} n }.
	\end{equation*}
	Applying our growth assumptions to ensure that $\eta_n = o(1)$,
	\begin{equation} \label{eq:xihats:max}
		\max_{i \in [n]} \, \norm*{\xihat_{i \cdot} - \xitilde_{i \cdot} \Q}
		= \op{1}.
	\end{equation}
	Adding and subtracting appropriate quantities,
	\begin{equation*}
		\xitilde_{i \cdot}
		= \xitilde_{i \cdot} - \xi_{i \cdot} + \xi_{i \cdot}
		= \X_{i \cdot} - \W_{i \cdot} \Thetatilde - \X_{i \cdot}
		+ \W_{i \cdot} \Theta + \xi_{i \cdot}
		= \W_{i \cdot} \paren*{\Theta - \Thetatilde} + \xi_{i \cdot},
	\end{equation*}
	so that $\norm*{\xitilde_{i \cdot}} = \norm*{\xi_{i \cdot}} + \op{ n^{-1/2} \log^{1/2} n }$, and it follows that
	\begin{equation*}
		\begin{aligned}
			\norm*{\xihat_{i \cdot}^T \xihat_{i \cdot} - \Q^T \xitilde_{i \cdot}^T \xitilde_{i \cdot} \Q}^2
			 & = \norm*{
				\paren*{\xihat_{i \cdot}^T - \Q^T \xitilde_{i \cdot}^T}
				\paren*{\xihat_{i \cdot} - \xitilde_{i \cdot} \Q}
				+ \paren*{\xihat_{i \cdot}^T - \Q^T \xitilde_{i \cdot}^T} \xitilde_{i \cdot}
				+ \xitilde_{i \cdot}^T \paren*{\xihat_{i \cdot} - \xitilde_{i \cdot} \Q}
			}^2                                                                               \\
			 & \le \norm*{\xihat_{i \cdot} - \xitilde_{i \cdot} \Q}^4
			+ 2 \, \norm*{\xi_{i \cdot}} \norm*{\xihat_{i \cdot} - \xitilde_{i \cdot} \Q}^3
			+ \norm*{\xitilde_{i \cdot}}^2 \norm*{\xihat_{i \cdot} - \xitilde_{i \cdot} \Q}^2 \\
			 & = \op{1} \cdot \brac*{1 + \norm*{\xi_{i \cdot}} + \norm*{\xi_{i \cdot}}^2},
		\end{aligned} \end{equation*}
	where we have made several applications of Equation~\eqref{eq:xihats:max}.
	Thus,
	\begin{equation*}
		\frac{1}{n} \sum_{i=1}^n \norm*{\xihat_{i \cdot}^T \xihat_{i \cdot} - \Q^T \xitilde_{i \cdot}^T \xitilde_{i \cdot} \Q}^2
		\le \op{1} \cdot
		\frac{1}{n}
		\sum_{i=1}^n \left( 1 + \norm*{\xi_{i \cdot}}
		+ \norm*{\xi_{i \cdot}}^2 \right).
	\end{equation*}
	By the law of large numbers and Assumption~\ref{ass:second-stage}, the mean on the right-hand side converges in probability to a finite constant.  Thus,
	\begin{equation*}
		\frac{1}{n} \sum_{i=1}^n \norm*{\xihat_{i \cdot}^T \xihat_{i \cdot} - \Q^T \xitilde_{i \cdot}^T \xitilde_{i \cdot} \Q}^2
		= \op{1},
	\end{equation*}
	which establishes Equation~\eqref{eq:xihat:goal} and completes the proof.
\end{proof}

\section{Proofs for Causal Estimators}
\label{app:causal-proofs}

Here we collect the proofs of Theorems~\ref{thm:nde} and~\ref{thm:nie}.

\subsection{Proof of Theorem \ref{thm:nde}}

\begin{proof}
	Follows immediately from Theorem \ref{thm:main}, Proposition \ref{prop:identification-mediated} and the delta method.
\end{proof}

\subsection{Proof of Theorem \ref{thm:nie}}

Before giving a formal proof of Theorem~\ref{thm:nie}, we give a high-level sketch of the argument. As above, let $\nietilde$ denote the analogous estimator to $\niehat$, but based on $\X$ rather than $\Xhat$.

First, we show that $\nietilde$ converges to a normal distribution centered on $\nie$ by a stacked M-estimation argument. This requires an extension of Theorem~\ref{thm:main-subgamma} to show that $(\betatilde, \Thetatilde)$ are jointly asymptotically normal rather than marginally asymptotically normal. Then, using the delta method, we show
\begin{equation*}
	\sqrt n \paren*{\nietilde - \nie}
	\to \Normal{0}{\sigma^2_\text{nie}}.
\end{equation*}
Next we show that we can replace the true latent positions $\X$ with the estimates $\Xhat$ without changing the asymptotic distribution of our estimates. By Slutsky's theorem it is sufficient show that
\begin{equation*}
	\abs*{\niehat - \nietilde} = \op{\frac{1}{\sqrt n}}.
\end{equation*}
Finally, we argue that $\sigmahatnie$ is a consistent estimator for $\sigmatildenie$, which is itself a consistent estimator for $\sigmanie$ under some mild conditions \citep[Theorem 7.3, Theorem 7.4]{boos2013}.

\begin{proof}
	Define $\btheta = \paren*{\Theta_{\cdot 1}^T, ..., \Theta_{\cdot d}^T, \beta^T}^T$. Observe that $\paren*{\widehat{\Theta}_{\cdot 1}^T, ..., \widehat{\Theta}_{\cdot d}^T, \betahat^T}^T$ is an $M$-estimator with estimating function
	\begin{equation*}
		\psi \paren*{
			\paren*{Y_i, \W_{i \cdot}, \D_{i \cdot}},
			\btheta^*
		}
		= \begin{bmatrix}
			\paren*{\X_{i 1} - \W_{i \cdot} \Theta_{\cdot 1}} \W_{i \cdot} \\
			\vdots                                                         \\
			\paren*{\X_{i d} - \W_{i \cdot} \Theta_{\cdot d}} \W_{i \cdot} \\
			\paren*{Y_i - \D_{i \cdot} \beta} \D_{i \cdot}
		\end{bmatrix}.
	\end{equation*}
	Recall the prior definitions of $\D$ and $\W$ as the design matrices for the mediator and outcome regressions. Then, by straightforward calculation, or \citet[Chapter 7]{boos2013}, and recalling the definitions of $\Athetavec \in \R^{pd \times pd}$ and $\Abeta \in \R^{(d + p) \times (d + p)}$ from Assumption \ref{ass:regularity}, we have, under sufficient regularity conditions for limits and expectations to be interchangeable,
	\begin{equation*} \begin{aligned}
			A_{\btheta}
			 & = \lim_{n \to \infty}
			\E{
				- \frac{1}{n} \sum_{i=1}^n \psi' \paren*{
					\paren*{Y_i, \W_{i \cdot}, \D_{i \cdot}},
					\btheta
				}
			}                                                                                                                                                                           \\
			 & = \lim_{n \to \infty}
			\mathbb E
			\begin{bmatrix}
				\frac{1}{n} \sum_{i=1}^n \W_{i \cdot} \W_{i \cdot}^T & 0      & \dots                                                & 0                                                    \\
				0                                                    & \ddots & 0                                                    & 0                                                    \\
				\vdots                                               & 0      & \frac{1}{n} \sum_{i=1}^n \W_{i \cdot} \W_{i \cdot}^T & 0                                                    \\
				0                                                    & 0      & 0                                                    & \frac{1}{n} \sum_{i=1}^n \D_{i \cdot} \D_{i \cdot}^T
			\end{bmatrix} \\
			 & =
			\begin{bmatrix}
				\Athetavec & 0      \\
				0          & \Abeta
			\end{bmatrix}.
		\end{aligned} \end{equation*}
	Further, again recalling the definitions of $\Bthetavec \in \R^{pd \times pd}$ and $\Bbeta \in \R^{(d + p) \times (d + p)}$ from Assumption \ref{ass:regularity}, we have
	\begin{equation*} \begin{aligned}
			B_{\btheta}
			 & = \lim_{n \to \infty}
			\E{
				\frac{1}{n} \sum_{i=1}^n
				\psi \paren*{\paren*{Y_i, \W_{i \cdot}, \D_{i \cdot}}, \btheta} \,
				\psi \paren*{\paren*{Y_i, \W_{i \cdot}, \D_{i \cdot}}, \btheta}^T
			}                        \\
			 & = \lim_{n \to \infty}
			\mathbb E
			\brac*{
				\frac{1}{n} \sum_{i=1}^n
				\begin{bmatrix}
					\xi_{i 1} \W_{i \cdot} \W_{i \cdot}^T \xi_{i 1}^T & \dots  & \xi_{i 1} \W_{i \cdot} \W_{i \cdot}^T \xi_{i d}^T & 0                                           \\
					\vdots                                            & \ddots & \vdots                                            & 0                                           \\
					\xi_{i d} \W_{i \cdot} \W_{i \cdot}^T \xi_{i 1}^T & \dots  & \xi_{i d} \W_{i \cdot} \W_{i \cdot}^T \xi_{i d}^T & 0                                           \\
					0                                                 & 0      & 0                                                 & \varepsilon_i^2 \D_{i \cdot} \D_{i \cdot}^T
				\end{bmatrix}
			}                        \\
			 & =
			\begin{bmatrix}
				\Bthetavec & 0      \\
				0          & \Bbeta
			\end{bmatrix}.
		\end{aligned} \end{equation*}
	While the diagonal blocks of $B_{\btheta}$ correspond to $B_{\vecc \paren*{\Theta}}$, and $B_{\beta}$, we show the derivation of the off-diagonal blocks. By the tower law, definition of $D$ and $\xi$, and then Assumption~\ref{ass:regularity}, we have
	\begin{equation*}
		\begin{aligned}
			B_{14}
			 & = \frac{1}{n} \sum_{i=1}^n \E{\E[\D_{i \cdot}]{\xi_{i 1} \W_{i \cdot}  \D_{i \cdot}^T \varepsilon_i^T}}                                  \\
			 & = \frac{1}{n} \sum_{i=1}^n \E{(\X_{i 1} - \W_{i \cdot} \Theta_{\cdot 1}) \W_{i \cdot}  \D_{i \cdot}^T \E[\D_{i \cdot}]{\varepsilon_i^T}} \\
			 & = 0
		\end{aligned}
	\end{equation*}
	\noindent where the derivations to show that $B_{23} = B_{34} = 0$ are analogous to the $B_{14}$ case. Under Assumption \ref{ass:regularity}, Theorem 7.2 of \cite{boos2013} implies that
	\begin{equation*}
		\begin{bmatrix}
			\vecc \paren{\Thetahat} \\
			\betahat
		\end{bmatrix}
		\to
		\Normal{
			\begin{bmatrix}
				\vecc \paren*{\Theta} \\
				\beta
			\end{bmatrix}
		}{
			\begin{bmatrix}
				\Sigmathetavec & 0          \\
				0              & \Sigmabeta
			\end{bmatrix}
		},
	\end{equation*}
	\noindent where $\Sigmathetavec$ and $\Sigmabeta$ are the same marginal covariances for $\Thetavechat$ and $\betahat$ from Theorem \ref{thm:main-subgamma}. This is an extension of Theorem \ref{thm:main-subgamma} in that it shows that $\Thetahat$ and $\betahat$ are jointly, rather than marginally, asymptotically normal. Also note that the asymptotic covariances between $\Thetahat$ and $\betahat$ are zero, such that we can concatenate the previous marginal covariance estimators to obtain a joint covariance estimate.
	Next we show that
	\begin{equation*}
		\abs*{\niehat - \nietilde} = \op{\frac{1}{\sqrt n}}.
	\end{equation*}
	Recall that
	\begin{equation*}
		\niehat = \paren*{t - t^*} \, \thetathat \, \betaxhat
	\end{equation*}
	We apply the submultiplicativity to obtain
	\begin{equation*}
		\label{eq:nieXhattoX}
		\frac{\abs*{\niehat - \nietilde}}{\paren*{t - t^*}} \le \norm*{\thetathat \betaxhat - \thetattilde \betaxtilde}
	\end{equation*}
	We now bound the first term on the right-hand side via
	\begin{align*}
		\abs*{\thetathat \betaxhat - \thetattilde \betaxtilde}
		 & \le
		\norm*{\thetathat \paren*{\betaxhat - \betaxtilde}} +
		\norm*{\paren*{\thetathat - \thetattilde} \betaxtilde} \\
		 & \le
		\norm*{\thetathat - \thetattilde} \norm*{\betaxhat - \betaxtilde} +
		\norm*{\thetattilde} \norm*{\betaxhat - \betaxtilde} +
		\norm*{\thetathat - \thetattilde} \norm*{\betaxtilde}.
	\end{align*}

	By Lemma \ref{lem:true-x-point-ests}, $\norm*{\thetattilde}$ and $\norm*{\betaxtilde}$ are both $\Op{1}$. By Theorem \ref{thm:main}, $\norm*{\thetathat - \thetattilde}$ and $\norm*{\betaxhat - \betaxtilde}$ are both $\op{n^{-1/2}}$. Thus, we obtain that the upper display is $\op{n^{-1/2}}$, as desired.

	Finally, we show that $\sigmahatnie$ is a consistent estimator for $\sigmanie$. By Theorem 7.2 of \cite{boos2013} and the delta method, the asymptotic variance of $\nietilde$ is given by
	\begin{equation*} \begin{aligned}
			\sigmanie
			= \paren*{t - t^*}^T
			\begin{bmatrix}
				\betax \\
				\thetat
			\end{bmatrix}^T
			\begin{bmatrix}
				\Sigma_{\thetat} & 0               \\
				0                & \Sigma_{\betax} \\
			\end{bmatrix}
			\begin{bmatrix}
				\betax \\
				\thetat
			\end{bmatrix}
			\paren*{t - t^*}.
		\end{aligned} \end{equation*}

	By Theorem \ref{thm:main}, Proposition \ref{prop:SigmabetaD}, Proposition \ref{prop:SigmathetaX}, and the continuous mapping theorem, $\sigmahatnie$ converges to $\sigmatildenie$ in probability.
\end{proof}

\section{Additional Simulation Results}
\label{app:additional-sims}

In this section we report additional results from our simulation study.

\subsection{Convergence Rates for \texorpdfstring{$\betahat$ and $\Thetahat$}{Regression Coefficients}}

In Section~\ref{sec:simulations}, we showed that $\ndehat$ converged to $\nde$ and $\niehat$ converged to $\nie$. We now show similar convergence results for the regression estimators $\betahat$ and $\Thetahat$. Recall that $\betaxhat$ and $\Thetahat$ only recover $\betax$ and $\Theta$ up to an unknown orthogonal transformation $Q$. Luckily, since the true latent positions are known in our simulations, we can align $\Xhat$ with the latent $\X$ by solving a Procrustes alignment problem \citep{gower2004}. As a result, we can investigate element-wise parameter recovery even in the presence of orthogonal non-identifiability. Figures~\ref{fig:uninformative_theta_loss} and~\ref{fig:informative_theta_loss} visualize results for $\Thetahat$. Figures~\ref{fig:uninformative_beta_loss_average} and~\ref{fig:informative_beta_loss_average} visualize results $\betahat$. These figures show convergence of each element of $\Thetahat$ and $\betahat$ to the corresponding elements of $\Theta$ and $\beta$. This convergence occurs at the same $\sqrt n$-rate across several simulation settings, as expected under Theorem \ref{thm:main}.

\begin{figure}
	\centering
	\includegraphics[width=0.7\textwidth]{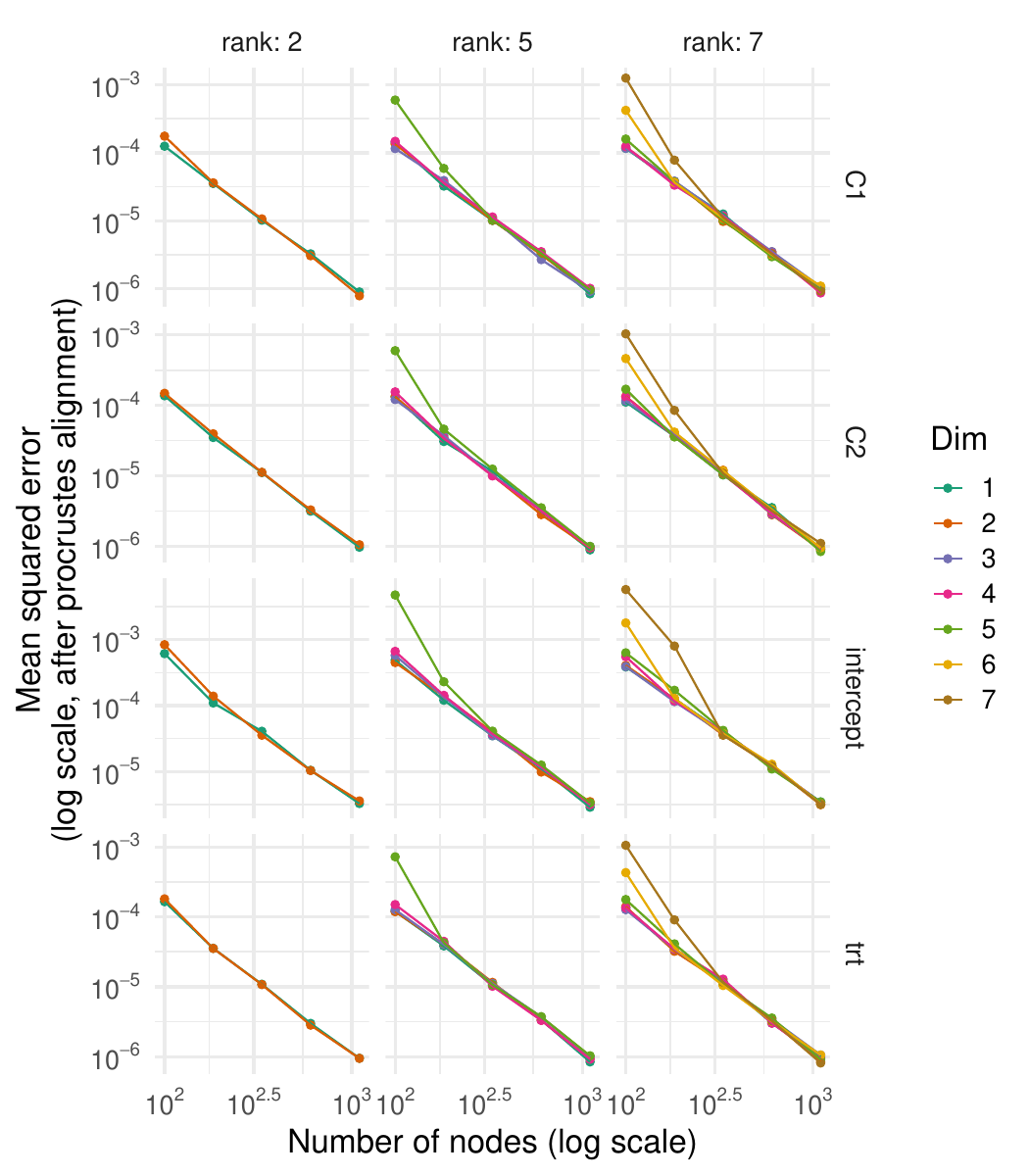}
	\caption{Elementwise $\ell_1$ convergence of $\Thetahat \, \Q^T$ to $\Theta$ under the uninformative model. Recall that $\Thetahat$ is a matrix-valued estimator. Each panel shows the $\ell_1$ error (vertical axis, log scale) of a portion of $\Thetahat$ as a function of the number of nodes in the network (horizontal axis, log scale). Within each panel, each line represents the error for a single coefficient corresponding to a particular dimension of the latent space. Panels vary horizontally by number of latent communities (left: two blocks, middle: five block, right: seven blocks) and vertically by column of the design matrix $\W$.}
	\label{fig:uninformative_theta_loss}
\end{figure}

\begin{figure}
	\centering
	\includegraphics[width=0.7\textwidth]{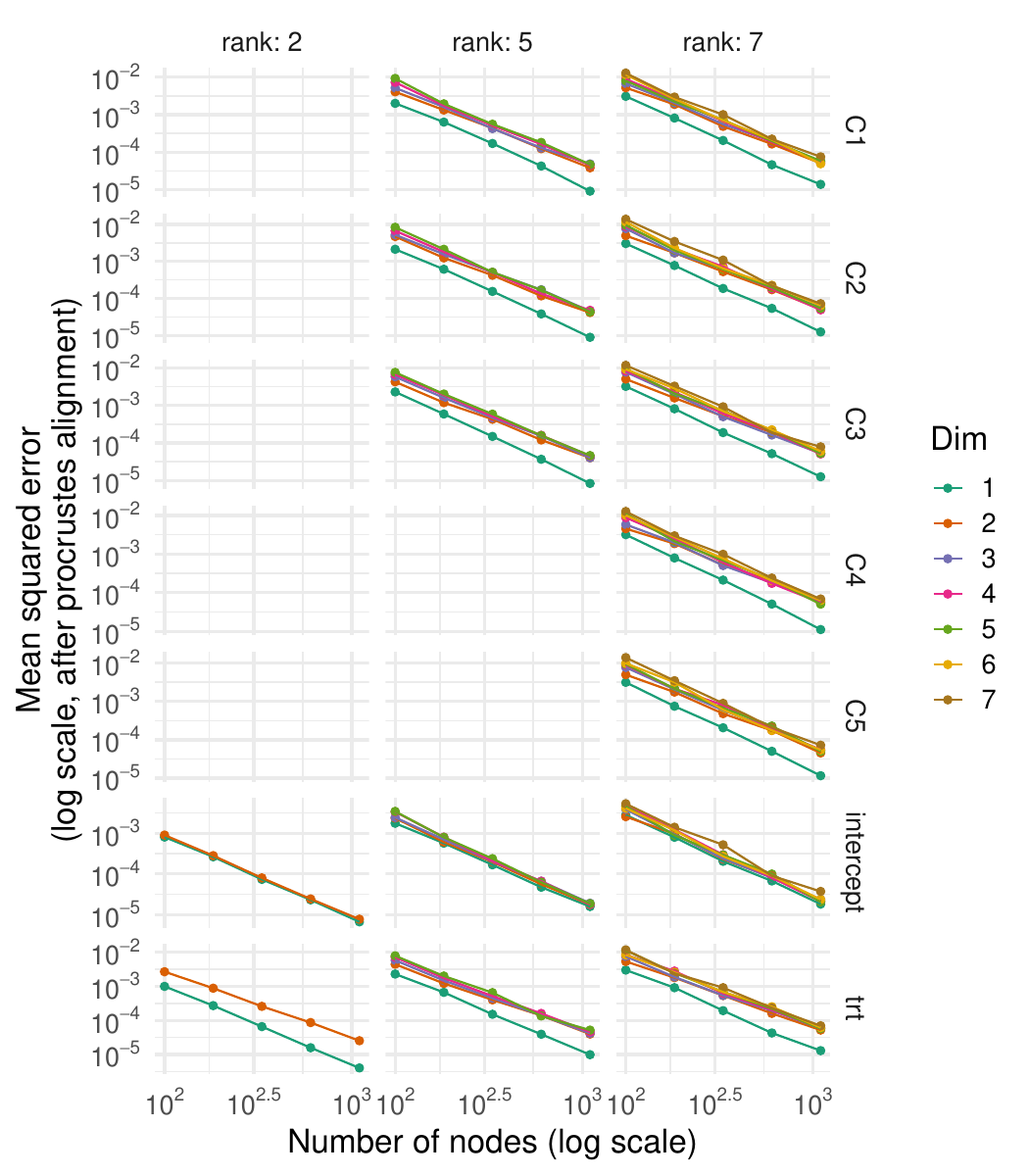}
	\caption{Elementwise $\ell_1$ convergence of $\Thetahat \, \Q^T$ to $\Theta$ under the informative model. Recall that $\Thetahat$ is a matrix-valued estimator. Each panel shows the $\ell_1$ error (vertical axis, log scale) of a portion of $\Thetahat$ as a function of the number of nodes in the network (horizontal axis, log scale). Within each panel, each line represents the error for a single coefficient corresponding to a particular dimension of the latent space. Panels vary horizontally by number of latent communities (left: two blocks, middle: five block, right: seven blocks) and vertically by column of the design matrix $\W$.}
	\label{fig:informative_theta_loss}
\end{figure}

\begin{figure}
	\centering
	\includegraphics[width=0.7\textwidth]{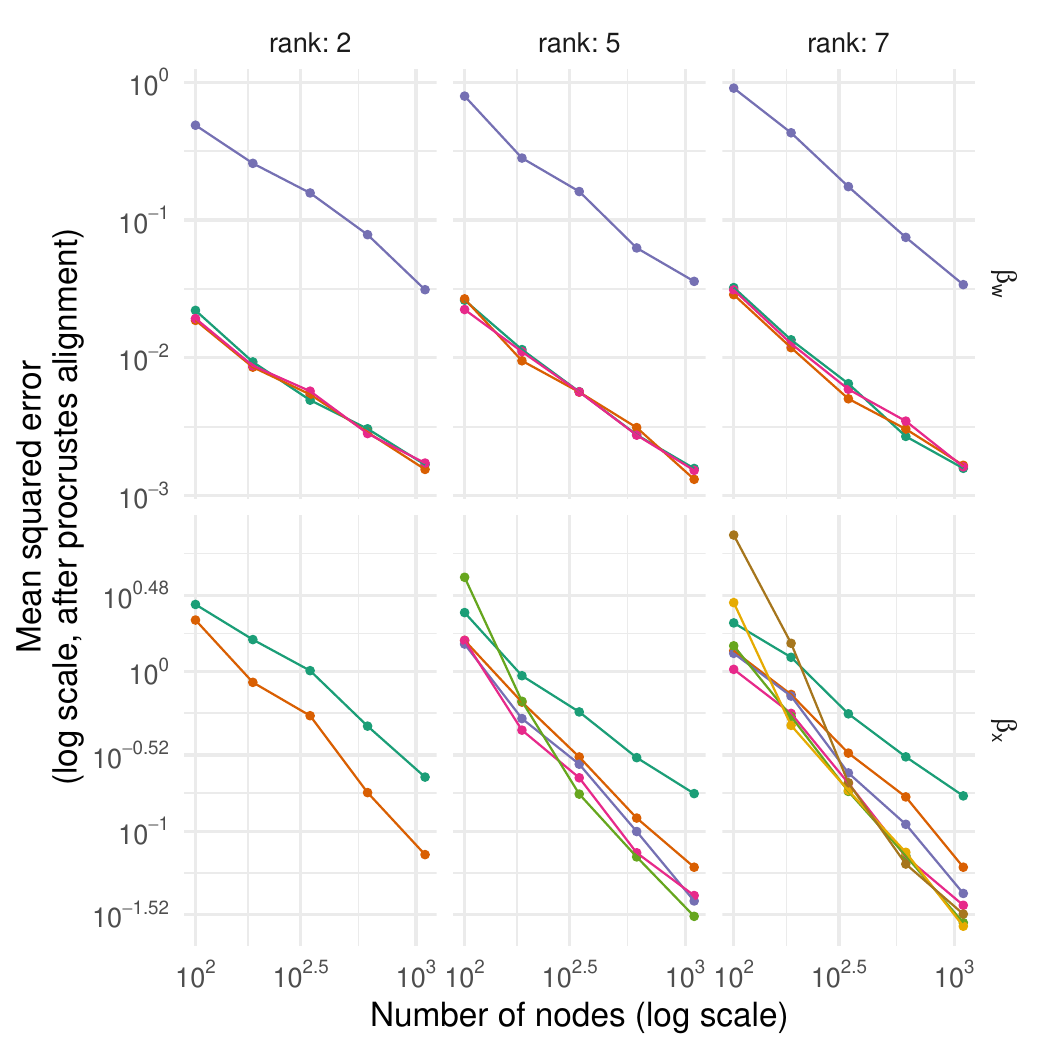}
	\caption{Convergence of $\betawhat$ to $\betaw$ and $\Q \betaxhat$ to $\betax$ under the uninformative model. Each panel shows the $\ell_1$ error (vertical axis, log scale) of a portion of $\betahat$ as a function of the number of nodes in the network (horizontal axis, log scale). Within each panel, each line represents the error for a single coefficient. We visualize results for $\betawhat$ and $\betaxhat$ in separate rows of panels, since only $\betaxhat$ is subject to rotational non-identifiability. Panels vary horizontally by number of latent communities (left: two blocks, middle: five block, right: seven blocks).}
	\label{fig:uninformative_beta_loss_average}
\end{figure}

\begin{figure}
	\centering
	\includegraphics[width=0.7\textwidth]{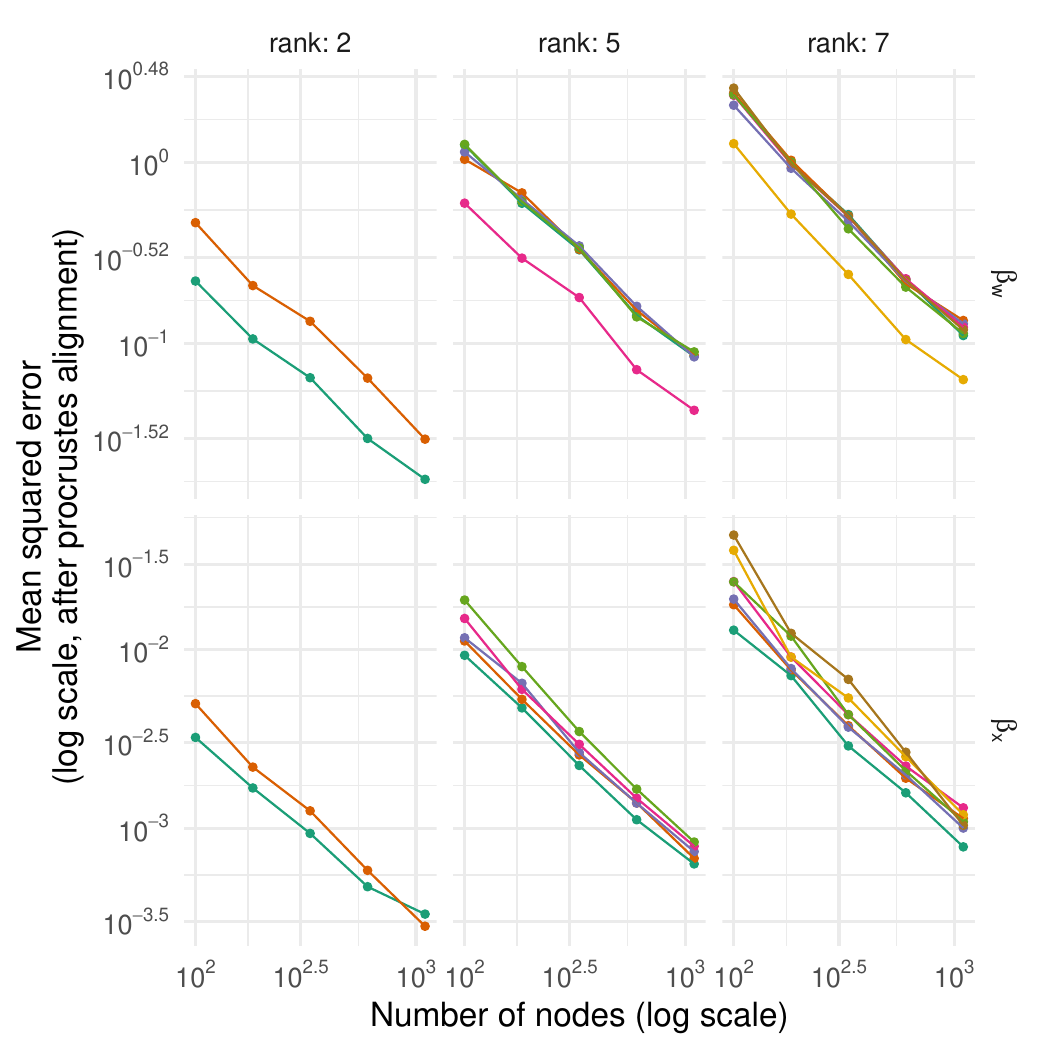}
	\caption{Convergence of $\betawhat$ to $\betaw$ and $\Q \betaxhat$ to $\betax$ under the informative model. Each panel shows the $\ell_1$ error (vertical axis, log scale) of a portion of $\betahat$ as a function of the number of nodes in the network (horizontal axis, log scale). Within each panel, each line represents the error for a single coefficient. We visualize results for $\betawhat$ and $\betaxhat$ in separate rows of panels, since only $\betaxhat$ is subject to rotational non-identifiability. Panels vary horizontally by number of latent communities (left: two blocks, middle: five block, right: seven blocks).}
	\label{fig:informative_beta_loss_average}
\end{figure}

\subsection{Finite Sample Bias in \texorpdfstring{$\betahat$}{Outcome Regression Coefficients}}

As explained in Remark \ref{rem:finite-sample-bias}, it is well-known that ordinary least squares estimates are biased when regression covariates are measured with error. Asymptotically, this is not an issue for $\betahat$, as the deviance of $\Xhat$ around $\X$ tends goes to zero in two-to-infinity norm, such that ``measurement error'' shrinks to zero as the number of nodes in the network grows. Nonetheless, we visualize the finite sample bias of $\betahat$ in Figures~\ref{fig:uninformative_beta_bias} and~\ref{fig:informative_beta_bias}. Unsurprisingly, $\betahat$ is biased, as expected due to the noise in $\Xhat$ around $\X$. As $n$ increases and $\Xhat$ converges to $\X$, the bias rapidly shrinks.

\begin{figure}
	\centering
	\includegraphics[width=0.7\textwidth]{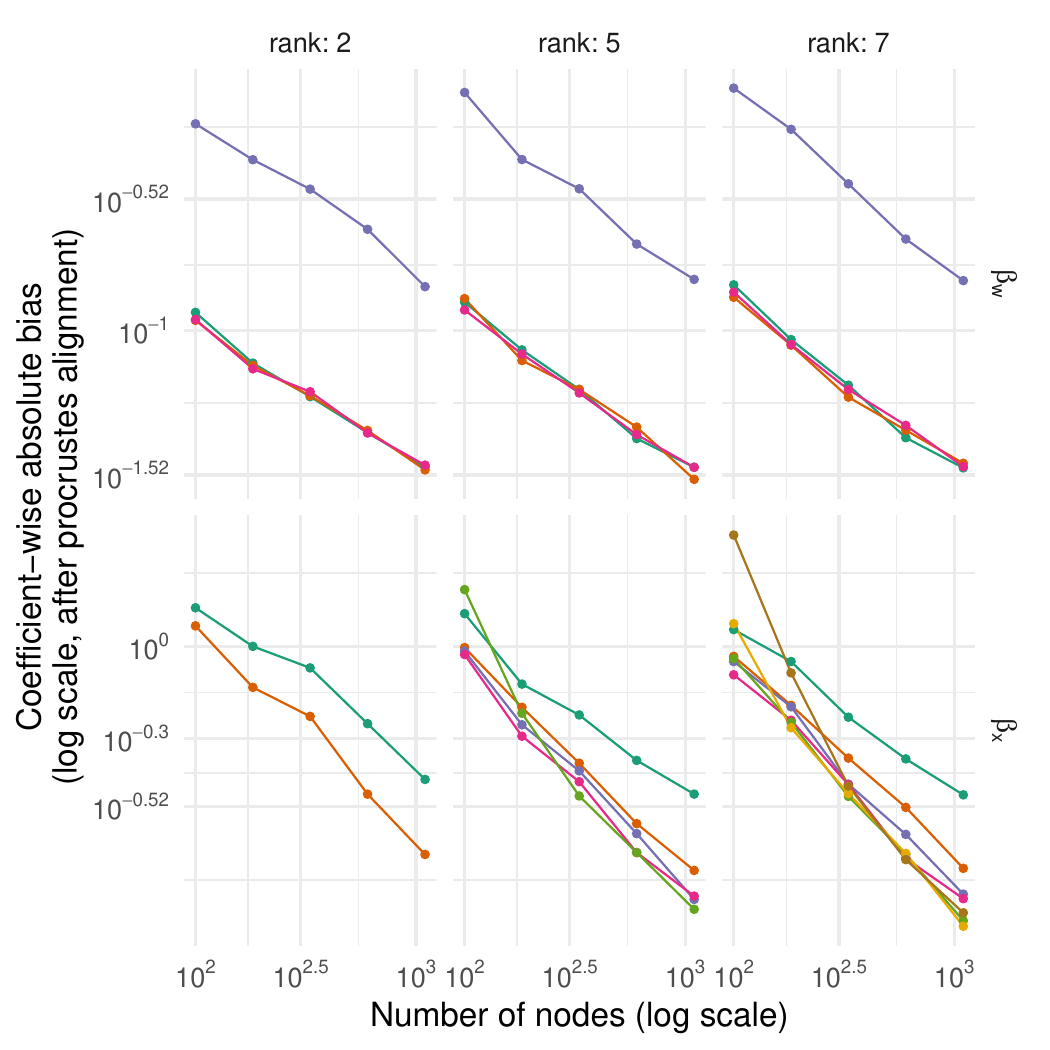}
	\caption{Elementwise bias of $\betahat$ for $\beta$ under the uninformative model. This can be thought of as measurement error bias induced by using $\Xhat$ in place of $\X$. Each column of panels corresponds to a distinct model, where models have varying numbers of latent communities. We visualize results for $\betawhat$ and $\betaxhat$ in separate rows of plots. Within each panel, each line represents results for a single coefficient. Note that the bias disappears asymptotically, as $\Xhat$ converges to $\X$.}
	\label{fig:uninformative_beta_bias}
\end{figure}

\begin{figure}
	\centering
	\includegraphics[width=0.7\textwidth]{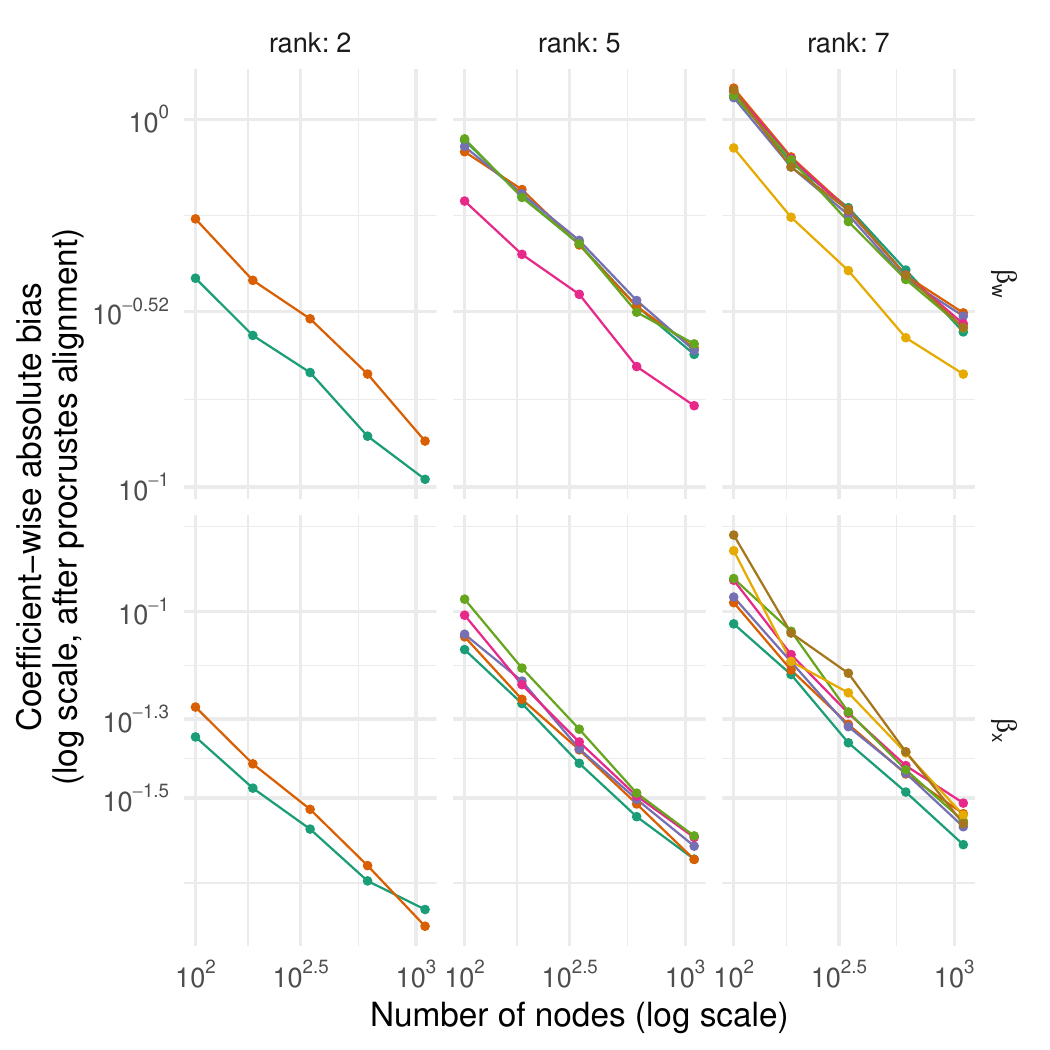}
	\caption{Elementwise bias of $\betahat$ for $\beta$ under the informative model. This can be thought of as measurement error bias induced by using $\Xhat$ in place of $\X$. Each column of panels corresponds to a distinct model, where models have varying numbers of latent communities. We visualize results for $\betawhat$ and $\betaxhat$ in separate rows of plots. Within each panel, each line represents results for a single coefficient. Note that the bias disappears asymptotically, as $\Xhat$ converges to $\X$.}
	\label{fig:informative_beta_bias}
\end{figure}

\subsection{Robustness of Causal Point Estimates to Rank Misspecification}

In Theorems \ref{thm:nde} and \ref{thm:nie}, the dimension $d$ of the latent space is taken to be known or otherwise correct specified. In Figure~\ref{fig:misspecification_mean_squared_error}, we investigate the estimation error of $\ndehat$ and $\niehat$ when the dimension of the latent space is misspecified. As in Figure~\ref{fig:causal_coverage} (which investigates coverage rates when $d$ is specified), we find that it is dramatically better, in terms of estimation error, to overestimate $d$ than it is to underestimate $d$. This aligns with previous results that suggest overestimating the dimension of the embedding space in stochastic blockmodels incurs a performance penalty but otherwise retains nice properties of estimators like consistency \protect\citep{fishkind2013}.

\begin{figure}
	\centering
	\includegraphics[width=0.8\textwidth]{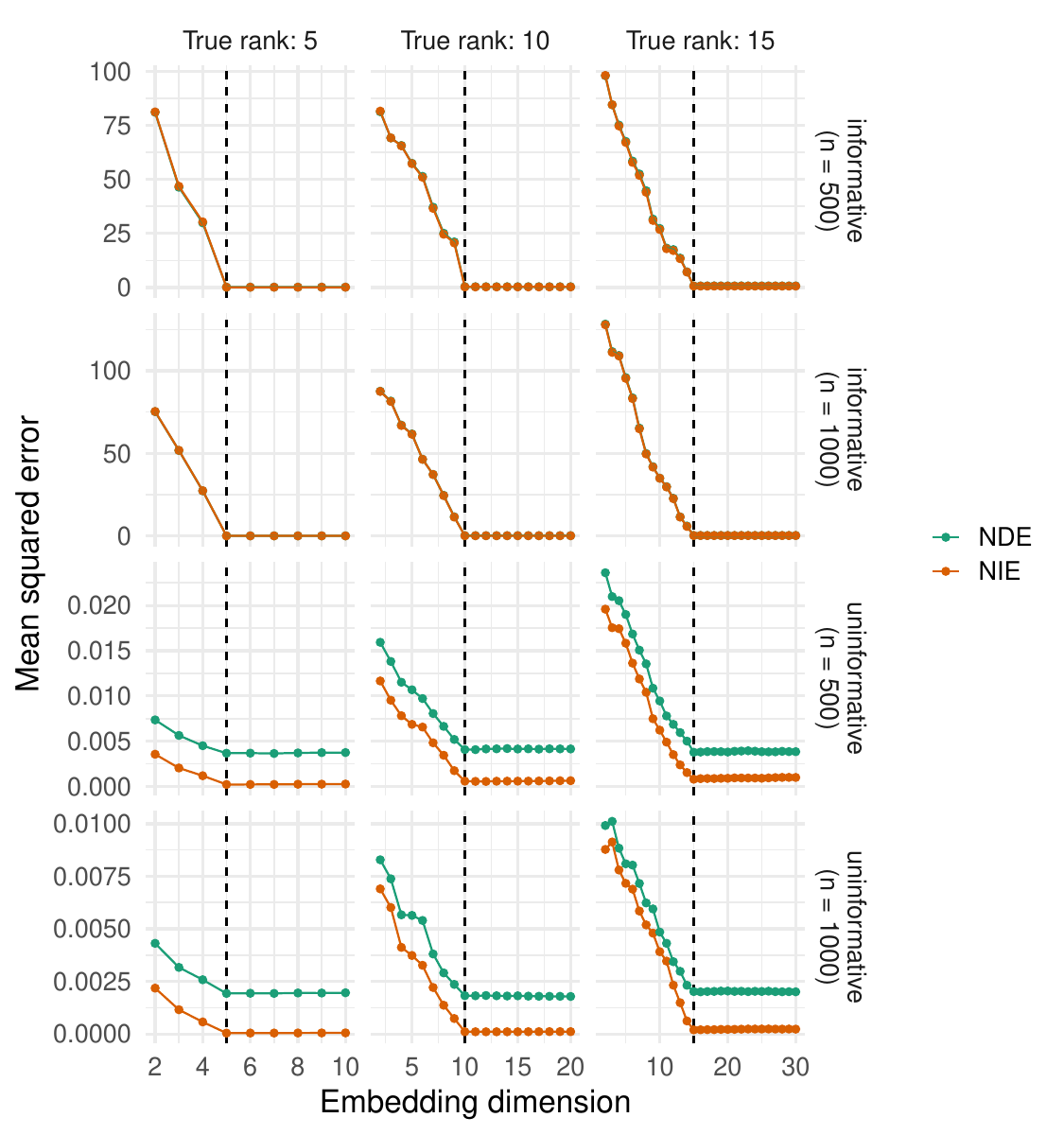}
	\caption{Mean squared error of $\ndehat$ and $\niehat$ when the dimension $d$ is misspecified. Each panel shows mean squared error (vertical axis) of $\nde$ (teal) and $\nie$ (orange) as a function of the embedding dimension $d$ (horizontal axis). The dashed vertical line denotes the true latent dimension. Panels vary horizontally by number of latent communities (left: five, middle: ten, right: fifteen) and vertically by the simulation model and number of nodes in the network.}
	\label{fig:misspecification_mean_squared_error}
\end{figure}

\subsection{Causal Estimation Error When Either \texorpdfstring{$\nde$ or $\nie$}{Natural Mediated Effect} is Zero}
\label{subsec:null-sims}

We additionally investigate if estimation error and converge rates behave as expected when either $\nde = 0$ or $\nie = 0$. To generate data where $\nde = 0$, we simulate from the informative model with $\betat = 0$, and so to generate data where $\nie = 0$, we simulate from the informative model with $\betax = 0$. The results, in Figures \ref{fig:loss_average_null} and \ref{fig:causal_coverage_null}, show that estimator error and coverage rates do not behave any differently in the setting where $\nde$ or $\nie$ is zero.

\begin{figure}
	\centering
	\includegraphics[width=0.9\textwidth]{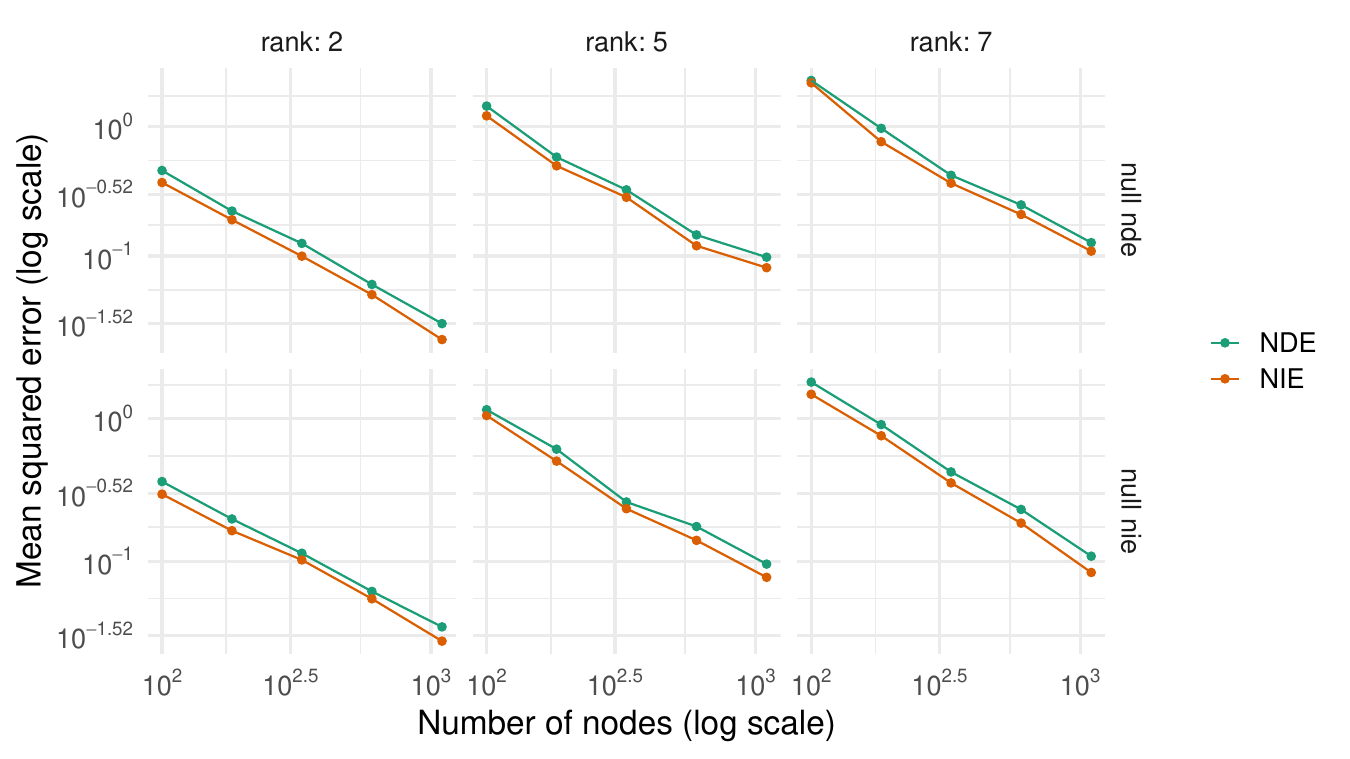}
	\caption{Convergence of $\ndehat$ to $\nde$ and $\niehat$ to $\nie$. Each panel shows the mean squared error (vertical axis, log scale) of $\ndehat$ (teal) and $\niehat$ (orange) as a function of the number of nodes in the network (horizontal axis, log scale). Panels vary horizontally by number of latent communities (left: two blocks, middle: five block, right: seven blocks) and vertically by the simulation model (top: informative, bottom: uninformative).}
	\label{fig:loss_average_null}
\end{figure}

\begin{figure}
	\centering
	\includegraphics[width=0.9\textwidth]{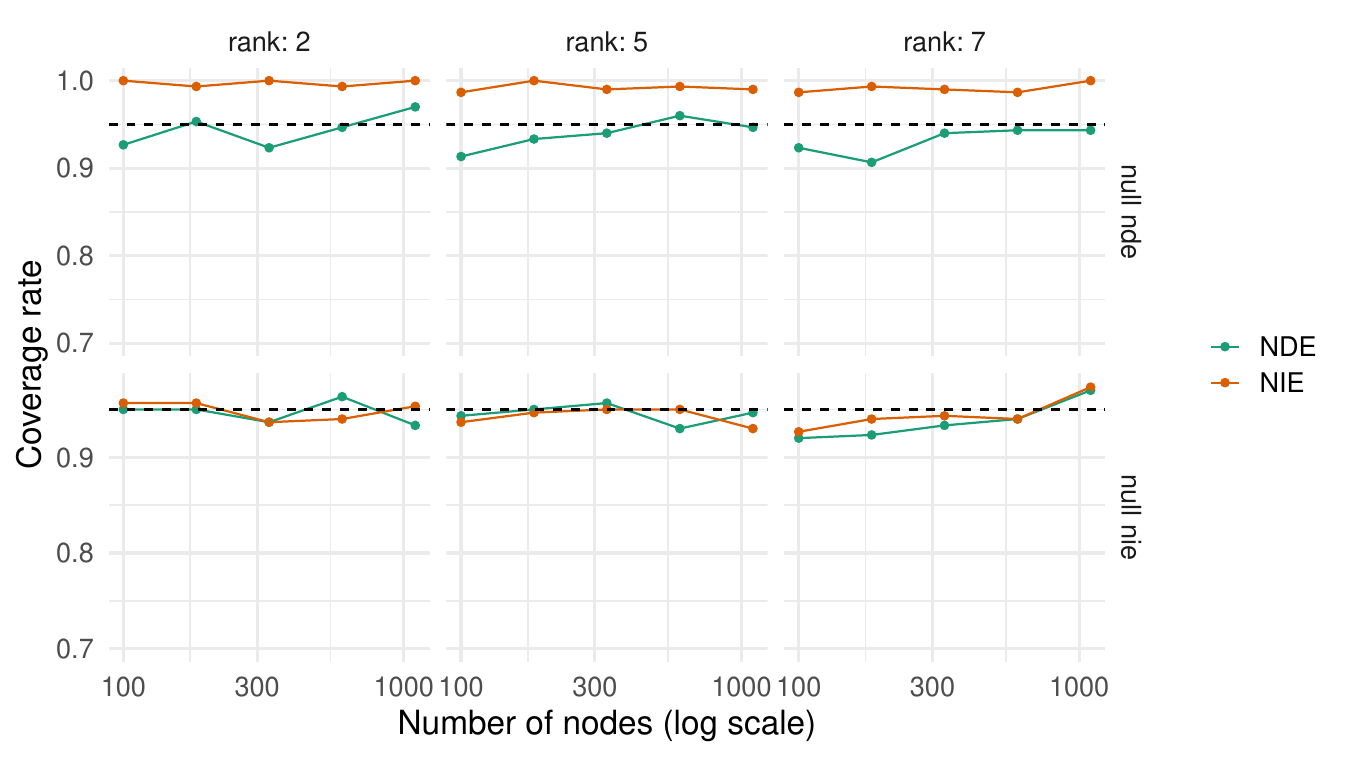}
	\caption{Finite sample coverage of asymptotic confidence intervals for $\nde$ and $\nie$. Each panel shows coverage (vertical axis) of $\nde$ (teal) and $\nie$ (orange) as a function of the number of nodes in the network (horizontal axis, log scale). The dashed horizontal line denotes the nominal coverage rate of 95\%. Panels vary horizontally by number of latent communities (left: two blocks, middle: five block, right: seven blocks) and vertically by the simulation model (top: informative, bottom: uninformative).}
	\label{fig:causal_coverage_null}
\end{figure}

\clearpage

\bibliography{references}

\end{document}